\documentclass[a4paper,11pt]{amsart} 
\usepackage{amsmath,amsxtra,amssymb,latexsym, amscd,amsthm}
\usepackage[mathscr]{eucal}
\usepackage{mathrsfs}
\usepackage{bbm}
\usepackage{enumerate}
\usepackage{pict2e}
\usepackage{tikz}
\usepackage{graphicx}
\usepackage[a4paper]{geometry}
\usepackage{color}
\usepackage[colorlinks=true,%
            linkcolor=red!75!black,%
            citecolor=blue!75!black,%
            urlcolor=darkgray]{hyperref}  % Needs to go last
\numberwithin{equation}{section}

\theoremstyle{plain}
\newtheorem{thm}{Theorem}[section]

\newtheorem{prp}[thm]{Proposition}
\newtheorem{cor}[thm]{Corollary}
\newtheorem{lem}[thm]{Lemma}
\theoremstyle{definition}
\newtheorem{defa}[thm]{Definition}

\newtheorem{rem}[thm]{Remark}

\newtheorem*{rem*}{Remark}

\newcommand{\dd}{\mathrm{d}}
\newcommand{\ee}{\mathrm{e}}
\newcommand{\ii}{\mathrm{i}}
\renewcommand{\Im}{\operatorname{Im}}
\renewcommand{\Re}{\operatorname{Re}}
\newcommand{\To}{\longrightarrow}
\newcommand{\Lim}{\mathop{\longrightarrow}\limits}
\newcommand{\N}{\mathbb{N}}
\newcommand{\Z}{\mathbb{Z}}

\newcommand{\R}{\mathbb{R}}
\newcommand{\C}{\mathbb{C}}

\newcommand{\T}{\mathbb{T}}

\newcommand{\cT}{\mathcal{T}}
\newcommand{\cN}{\mathcal{N}}
\newcommand{\cA}{\mathcal{A}}
\newcommand{\cC}{\mathcal{C}}

\newcommand{\cS}{\mathcal{S}}
\newcommand{\cB}{\mathcal{B}}

\newcommand{\cR}{\mathcal{R}}

\newcommand{\cO}{\mathcal{O}}
\newcommand{\cQ}{\mathcal{Q}}
\newcommand{\cP}{\mathcal{P}}
\newcommand{\cU}{\mathcal{U}}
\newcommand{\cE}{\mathcal{E}}
\newcommand{\cG}{\mathcal{G}}
\newcommand{\cK}{\mathcal{K}}
\newcommand{\cJ}{\mathcal{J}}
\newcommand{\cL}{\mathcal{L}}
\newcommand{\eps}{\epsilon}
\newcommand{\tilg}{{\tilde{g}}}
\DeclareMathOperator{\expect}{\mathbb{E}}

\DeclareMathOperator{\prob}{\mathbb{P}}

\DeclareMathOperator{\varie}{Var^{\mathit I}_{\eta_0}}

\newcommand{\Data}{{\textbf{(\hyperlink{lab:data}{Data})}}}
\newcommand{\EXP}{\textbf{(\hyperlink{lab:exp}{EXP})}}
\newcommand{\BST}{\textbf{(\hyperlink{lab:bst}{BST})}}
\newcommand{\Green}{\textbf{(\hyperlink{lab:green}{Green})}}
\newcommand{\NonDirichlet}{{\textbf{(\hyperlink{lab:non-dirichlet}%
{Non-Dirichlet})}}}
\newcommand{\Hol}{\textbf{(\hyperlink{lab:hol}{Hol})}}
\DeclareMathOperator{\Bad}{Bad}
\DeclareMathOperator{\Badp}{Badp}
\newcommand{\rmM}{\mathrm{M}}
\newcommand{\rmm}{\mathrm{m}}
\DeclareSymbolFont{extraup}{U}{zavm}{m}{n}
\DeclareMathSymbol{\varheart}{\mathalpha}{extraup}{86}
\DeclareMathSymbol{\vardiamond}{\mathalpha}{extraup}{87}
\DeclareMathOperator{\Lip}{Lip}
\title[Quantum ergodicity for quantum graphs in the AC regime]{Quantum ergodicity for expanding quantum
graphs in the regime of spectral delocalization}
\author{Nalini Anantharaman, Maxime Ingremeau, Mostafa Sabri, Brian Winn}
\address{Universit\'e de Strasbourg, CNRS, IRMA UMR 7501, F-67000 Strasbourg, France.}
\email{anantharaman@math.unistra.fr}
\address{Laboratoire J.A.Dieudonn\'e, UMR CNRS-UNS 7351, Universit\'e C\^ote d'Azur, 06108 Nice, France}
\email{maxime.ingremeau@univ-cotedazur.fr}
\address{Department of Mathematics, Faculty of Science, Cairo University, Giza 12613, Egypt.}
\email{mmsabri@sci.cu.edu.eg}
\address{Department of Mathematical Sciences, Loughborough University, Leicestershire, LE11 3TU, United Kingdom.}
\email{b.winn@lboro.ac.uk}
\usepackage{calc}
\usepackage{graphicx}

\makeatletter
\newlength{\temp@wc@width}
\newlength{\temp@wc@height}
\newcommand{\widecheck}[1]{%
  \setlength{\temp@wc@width}{\widthof{$#1$}}%
  \setlength{\temp@wc@height}{\heightof{$#1$}}%
  #1\hspace{-\temp@wc@width}%
  \raisebox{\temp@wc@height+2pt}[\heightof{$\widehat{#1}$}]%
     {\rotatebox[origin=c]{180}{\vbox to 0pt{\hbox{$\widehat{\hphantom{#1}}$}}}}%
}
\makeatother

\begin{document}

\begin{abstract}
We consider a sequence of finite quantum graphs with few loops, so that they converge, in the sense of Benjamini-Schramm, to a random infinite quantum tree. We assume these quantum trees are spectrally delocalized in some interval $I$, in the sense that their spectrum in $I$ is purely absolutely continuous and their Green's functions are well controlled near the real axis. We furthermore suppose that the underlying sequence of discrete graphs is expanding. We deduce a quantum ergodicity result, showing that the eigenfunctions with eigenvalues lying in $I$ are spatially delocalized.
\end{abstract}

\maketitle

\section{Introduction}
In their far-reaching work \cite{KS97}, Kottos and Smilansky suggested that the ideas and results of quantum chaos should apply to quantum graphs. By quantum graphs, we mean a metric graph, equipped with a differential operator and suitable boundary conditions at each vertex. We refer the reader to Section~\ref{subsec:QG} for a more precise definition.

In this paper, we study an analogue on quantum graphs of one of the most famous properties of quantum chaotic systems, namely quantum ergodicity. In its original context \cite{Shni, CdV,Zel}, quantum ergodicity says that, on a compact Riemannian manifold whose geodesic flow is ergodic, the eigenfunctions of the Laplace-Beltrami operator become equidistributed {\em{in the high-frequency limit}}.

On a {\em{fixed}} quantum graph (with Kirchhoff boundary conditions), it was shown in \cite{CdV15} that quantum ergodicity generically does not hold in the high-frequency limit, unless the graph is very simple (homeomorphic to an interval or a circle).

However, instead of studying the asymptotics of eigenfunctions on a fixed quantum graph, one may study quantum ergodicity on \emph{sequences} of quantum graphs whose size goes to infinity. Some positive results appeared in \cite{BKS07}, \cite{GKP10}, \cite{KaSc14} \cite{BrWi16} for several families of quantum graphs, while it was shown in \cite{BKW04} that quantum ergodicity does not hold on star graphs. We refer the reader to the introduction of \cite{QEQGEQ} for an up-to-date survey of these recent developments.

All the preceding results study the high frequency behaviour of eigenfunctions. In \cite{QEQGEQ}, we proved a quantum ergodicity result for regular equilateral quantum graphs $\cQ_N$ converging to the regular equilateral tree $\mathbf{T}_q$ (in the sense of Benjamini-Schramm), in the \emph{bounded energy regime}, more precisely for energies lying in some bounded interval $I\subset \sigma_{\mathrm{ac}}(H_{\mathbf{T}_q})$, assuming the underlying discrete graphs are expanders. Our aim in this paper is to extend this result to the case of non-regular non-equilateral quantum graphs satisfying $\delta$-conditions (\S~\ref{subsec:QG}).

As in \cite{QEQGEQ}, we suppose that the quantum graphs have few short cycles. This can be restated as supposing that our quantum graphs converge in the sense of Benjamini-Schramm to a measure $\mathbb{P}$ supported on the set of quantum trees. We introduced the notion of Benjamini-Schramm convergence for quantum graphs in \cite{BSQG}, in analogy to the case of discrete graphs. We also assume the underlying discrete graphs are expanders.

Our main assumption is that in the energy interval we consider, the spectrum at the limit is absolutely continuous. More precisely, we need a good control over the Green's function near the real axis (see hypothesis \Green{} below). Our results thus convert \emph{spectral delocalization} at the limit (AC spectrum) into \emph{spatial delocalization} for the eigenfunctions  (i.e.\ quantum ergodicity). We give in \S~\ref{sec:examp} two important examples in which our assumptions hold.

The results of this paper can be regarded as a quantum graph counterpart of the results of \cite{AS2} for discrete graphs, we refer the reader to that paper for a more detailed introduction of the problem of quantum ergodicity and its implications.

Let us discuss our results in more detail. Given a sequence $\cQ_N$ of growing quantum graphs, the aim is to show that for any orthonormal basis of eigenfunctions $(\psi_j^{(N)})$ on $\cQ_N$, the probability measure $|\psi_j^{(N)}(x)|^2\,\dd x$ approaches the uniform measure $\frac{1}{|\cQ_N|}\,\dd x$ when $N$ is large enough. We only aim to prove this for most eigenfunctions in an interval $I$, so we consider Ces\`aro means. More precisely, denoting by $\mathbf{N}_N(I)$ the number of eigenvalues of $\mathcal{Q}_N$ in $I$, we set to prove that
\begin{equation}\label{e:mainlimit}
 \lim_{N\to\infty} \frac{1}{\mathbf{N}_N(I)} \sum_{\lambda_j^{(N)}\in I} \left|\langle \psi_j^{(N)}, f_N\psi_j^{(N)}\rangle_{L^2(\cG_N)} -\langle f_N\rangle_{\lambda_j^{(N)}} \right| =0
\end{equation}
for any uniformly bounded observable $f_N\in L^\infty(\cG_N)$, where $\cG_N$ is the underlying metric graph. Since $\langle \psi_j,f\psi_j\rangle = \int_{\cG_N}f(x)|\psi_j(x)|^2\,\dd x$, if we had $\langle f\rangle_{\lambda_j^{(N)}}=\frac{1}{|\cG_N|}\int_{\cG_N}f(x)\,\dd x$, i.e., if $\langle f \rangle$ were the uniform averages independently of $\lambda_j^{(N)}$ this would show that $|\psi_j(x)|^2\,\dd x$ approaches $\frac{1}{|\cG_N|}\,\dd x$ in some weak sense.

It turns out that such perfect uniform distribution can only be true in very special cases, cf. \cite{QEQGEQ}. In fact, since we consider a regime of bounded energies (lying in a fixed $I$, not the high frequency regime), we expect the potential we put on the edges to have some influence over the probability of finding the wavefunction in various places of the graph, which is given by $|\psi_j^{(N)}(x)|^2\,\dd x$. The true ``limiting measure'' is thus not the uniform measure in general, but one with a possibly non-constant density. The density is very satisfactory as it is directly related to the spectral density of the limiting quantum tree. In fact, our results show in a weak sense that $|\psi_j^{(N)}(x)|^2\,\dd x$ tends to the measure $\frac{\Im \tilg^{\lambda_j^{(N)}+\ii 0}_N(\tilde{x},\tilde{x})}{\int_{\cG_N}\Im \tilg^{\lambda_j^{(N)}+\ii 0}_N(\tilde{y},\tilde{y})\,\dd y}\,\dd x$, where $\tilg^z_N(\tilde{x},\tilde{x})$ is the Green's function of the universal cover of $\cQ_N$ ($\tilg^z_N(\tilde{x},\tilde{x})$ approaches the Green's function of the limiting tree when $N\to\infty$, see Appendix~\ref{sec:BS}). Accordingly, the mean $\langle f_N\rangle_{\lambda_j^{(N)}}$ above will actually depend on the energy\footnote{Note that in the special case where $\tilg^z(x,x)$ is independent of $x$ we get the uniform measure $\frac{1}{|\cG_N|}\dd x$. Roughly speaking, the general quotient detects the inhomogeneities in the limit object, a tree in our setting.}  $\lambda_j^{(N)}$.

We now discuss the main steps of the proof:
\begin{enumerate}[(1)]
\item In Sections~\ref{sec:nonback} and \ref{Sec:Reduction} we reduce \eqref{e:mainlimit} to proving that analogous Ces\`aro means defined on the \emph{discrete graph} (which we call \emph{quantum variances}) vanish as $N\To\infty$. In this process the $L^2$ scalar product $\langle \psi_j,f\psi_j\rangle=\int_{\cG_N}|\psi_j(x)|^2f(x)\,\dd x$ is replaced by $\ell^2$ scalar products of the form $\sum_{v\in V_N}|\psi_j(v)|^2K_{f,j}(v)$ or $\sum_{v\in V_N}\sum_{w\sim v}\psi_j(v)\psi_j(w)M_{f,j}(v,w)$, for some auxiliary (energy-dependent) observables $K_{f,j},M_{f,j}$ built from $f$.
Such discretization philosophy is well established in the quantum graphs literature, especially when the quantum graph is equilateral, in which case the restrictions $\psi_j|_{V_N}$ become eigenfunctions of some nice adjacency matrix. It is known however that when the graph is not equilateral, the discretization produces a complicated energy-dependent Schr\"odinger operator. In this paper we circumvent this problem by using \emph{non-backtracking eigenfunctions} $f_j,f_j^\ast$ living on the directed edges of the graph, an idea that already proved fruitful in discrete graphs \cite{A,AS2}, and it is quite remarkable that it also works for quantum graphs. This new construction is explained in Section~\ref{sec:nonback}. The eigenvalue equation for $f_j,f_j^\ast$ implies that the corresponding quantum variance is invariant under simple averaging operators, weighted by the Green's functions of the quantum universal cover (Proposition \ref{p:VarInv}).
In the usual proof of quantum ergodicity on manifolds, we would instead be using the invariance of eigenfunctions under the wave propagator.
\item We can bound the quantum variance by Hilbert-Schmidt norms. For this we follow the general scheme of \cite{AS2}, but the procedure is complicated by two problems: first, neither the restrictions $\psi_j|_V$ of the eigenfunctions to the vertices, nor the non-backtracking eigenfunctions $f_j,f_j^\ast$ form an orthonormal basis; second we have less \textit{a priori} bounds on the auxiliary observables $K_{f,j},M_{f,j}$. This calls for several technical innovations (Sections~\ref{sec:Complex}--\ref{Sec:UpperBound}).
\item \phantomsection\label{page:step3}Applying the averaging operators from step (1), and developing the Hilbert-Schmidt norms (Section~\ref{Sec:HSestimate}) reduces the proof to some contraction estimates on a family of sub-stochastic operators. These estimates require a careful analysis involving the expanding properties of the graphs (Section~\ref{Sec:Contrac}). Compared to \cite{AS2}, this part contains several novelties which are necessary due to the more complicated recursion relations satisfied by the Green's functions in the quantum setting (Section~\ref{sec:Notation}).
\end{enumerate}

Each of the preceding steps is not ``exact'' in the sense that it holds modulo error terms. To control them, we derive in Appendix~\ref{sec:BS} some important implications of Benjamini-Schramm convergence and spectral delocalization.

\subsubsection*{Acknowledgments}
N.A. was supported by Institut Universitaire de France, by the ANR project GeRaSic ANR-13-BS01-0007 and by USIAS (University of Strasbourg Institute of Advanced Study).

M.I. was supported by the Labex IRMIA  during part of the realization of this project.

M.S. was supported by a public grant as part of the \textit{Investissement d'avenir} project,
reference ANR-11-LABX-0056-LMH, LabEx LMH. He thanks the Universit\'e Paris Saclay
for excellent working conditions, where part of this work was done.

\section{Main results}\label{sec:mainres}
\subsection{Quantum graphs}\label{subsec:QG}
Let $G=(V,E)$ be a graph with vertex set $V$ and edge set $E$.
For each vertex $v\in V$, we denote by $d(v)$ the degree of $v$.  If $v_1, v_2\in V$, we write $v_1\sim v_2$ if $\{v_1,v_2\}\in E$.
We let $B= B(G)$ be the set of oriented edges (or bonds), so that $|B|=2|E|$. We assume that there is at most one edge between two vertices, so we will view $B$ as a subset of $V\times V$.
 If $b\in B$, we shall denote by $\hat{b}$ the reverse bond. We write $o_b$ for the origin of $b$ and $t_b$ for the terminus of $b$. We will also write $e(b)\in E$ for the edge obtained by forgetting the orientation of $b$.

For us, a \emph{quantum graph} $\cQ=(V,E,L,W,\alpha)$ is the data of:
\begin{itemize}
\item A connected combinatorial graph $(V,E)$.
\item A map $L: E\rightarrow (0,\infty)$. If $b\in B$, we denote $L_b:= L(e(b))$.
\item A potential $W=(W_b)_{b\in B} \in \bigoplus_{b\in B} C^0([0,L_b]; \R)$ satisfying for $x\in [0,L_b]$,
\begin{equation}\label{eq:ReversePot}
W_{\widehat{b}}(L_b-x) = W_b(x)\,.
\end{equation}
\item Coupling constants $\alpha = (\alpha_v)_{v\in V} \in \R^V$.
\end{itemize}

The underlying \emph{metric graph} is defined by
\[
\mathcal{G}:= \{ x= (b, x_b); b\in B, x_b\in [0,L_b]\} \slash \simeq \,,
\]
where $(b,x_b)\simeq (b',x'_{b'})$ if $b'= \hat{b}$ and $x'_{b'}= L_b-x_b$. In the sequel, we will sometimes write $x\in b$ to indicate that $x=(b,x_b)\in \mathcal{G}$.
Condition \eqref{eq:ReversePot} then simply ensures that $W$ is well-defined on $\mathcal{G}$.

In general a function $f : \mathcal{G}\longrightarrow \R$ can be identified with a collection of maps $(f_b)_{b\in B}$ such that $f_b (L_b - \cdot) = f_{\hat{b}}(\cdot)$. We say that $f$ is \emph{supported on} $e$ for some $e\in E$ if $f_b \equiv 0$ unless $e(b)= e$. In the sequel, we will often write $f(x_b)$ instead of $f((b,x_b))$ or $f_b(x_b)$.

If each $f_b$ is positive and measurable, we define $\int_\mathcal{G} f(x) \mathrm{d}x := \frac{1}{2}\sum_{b\in B} \int_0^{L_b} f_{b}(x_b) \mathrm{d}x_b$. We may then define the spaces $L^p(\mathcal{G})$ for $p\in [1, \infty]$ in the natural way (see below).

In the sequel, we will always make the following assumptions:

\medskip 

\raisebox{1.2\ht\strutbox}{\hypertarget{lab:data}{}}\textbf{(Data)}
There exist $0<\mathrm{m}<\mathrm{M}$ and $D\in \N$ such that for all $v\in V$ and all $b\in B$:
\begin{align*}
3\le d(v) &\le D\\
|\alpha_{v}|&\le \mathrm{M}\\
\mathrm{m}\le L_b &\le \mathrm{M}\\
W_b\in \Lip([0,L_b]) \text{ and }&  \max\left(\| W_b\|_{\infty},\Lip(W_b)\right) \le \mathrm{M} \\
W_b(L_b- \cdot) &= W_b(\cdot).
\end{align*}

This last assumption says that the potential on each edge is symmetric (in particular, this is the case if $W\equiv 0$). The symmetry condition is used in the ``trigonometric'' relations \eqref{e:simplif}, but we believe one should be able to adapt our proof
to non-symmetric potentials --- at the expense of using a modified version of \eqref{e:simplif}.

Here $\Lip(I)$ is the set of Lipschitz-continuous functions on $I$ and $\Lip(f)$ is the Lipschitz constant of $f$.

Let $\cQ$ be a quantum graph. Consider the Hilbert space\footnote{In the sequel, all the scalar products in Hilbert spaces will be linear in the right variable, and anti-linear in the left one: $\langle z u , z' v \rangle = \overline{z} z' \langle u, v \rangle$.}
\begin{equation*}
L^2(\cG) :=\Big\{(f_b)_{b\in B} \in \mathop\bigoplus_{b\in B} L^2(0,L_b) \,\Big{|}\, f_{\widehat{b}}(L_b-\cdot) = f_{b} (\cdot) \text{ and } \sum_{b\in B} \|f_b\|^2_{L^2(0,L_b)}<\infty \Big\}
\end{equation*}
 and its subset 
\[
H^2(\cG):=\Big\{ (f_b)_{b\in B} \in \mathop\bigoplus_{b\in B} H^2(0,L_b) \,\Big{|}\, f_{\widehat{b}}(L_b-\cdot) = f_{b} (\cdot) \text{ and } \sum_{b\in B} \|f_b\|^2_{H^2(0,L_b)}<\infty \Big\}\,.
\]

We say that a function $f\in H^2(\cG)$ satisfies the $\delta$-boundary conditions at $v\in V$ if
\begin{equation}\label{eq:Continuity}
\begin{cases} \forall b,b'\in B \text{ such that } o_b=o_{b'}=v, \text{ we have } f_b(0)= f_{b'}(0) =: f(v),\\
\qquad\qquad\qquad\qquad \sum\limits_{b\in B,\, o_b=v} f'_{o_b}(0) = \alpha_v f(v).
\end{cases}
\end{equation}
 
The set of such functions is denoted by: 
\begin{equation*}
\mathcal{H}^{\cQ}:=\Big{\{}f\in H^2(\cG) \,\big{|}\, \forall v\in V, f \text{ satisfies \eqref{eq:Continuity}} \Big{\}}\,.
\end{equation*}

We then define an operator $H_\cQ$ acting on $\psi= (\psi_b)_{b\in B}\in  \mathcal{H}^\cQ$ by 
\begin{equation}\label{e:H}
(H_{\cQ}\psi_b)(x_b) = -\psi''_b(x_b) + W_b(x_b)\psi_b(x_b).
\end{equation}

Our assumption \Data{} implies in particular that the operator $H_\cQ : \mathcal{H}^{\cQ} \to L^2(\cG)$ is then self-adjoint \cite[Theorem 1.4.19]{BeKu13}.

We say that the quantum graph is \emph{finite} if $V$ and $E$ are finite. In this case we denote $\cL(\cQ) = \sum_{e\in E}L(e)$. If $\cQ$ is finite, $H_\cQ$ has compact resolvent (\cite[Theorem 3.1.1]{BeKu13}). So for finite $\cQ_N$, there exists an orthonormal basis $(\psi_j^{(N)})_{j\ge 1}$ of $L^2(\cG_N)$ made of eigenfunctions of $H_{\cQ_N}$. We denote by $\big(\lambda_j^{(N)}\big)_{j\ge 1}$ the corresponding eigenvalues. Recall that we write
\[
\mathbf{N}_N(I):= \# \left\{ j\geq 1 \mid \lambda_j^{(N)}\in I\right\}.
\]

 Since our Schr\"odinger operators have real coefficients, we may---and will---restrict our attention to real valued eigenfunctions. This restriction was already present in the work \cite{AS2} on discrete graphs. It is only necessary in Appendix~\ref{sec:nonbared}. More precisely, what we need to assume is that
\[
\overline{\psi_j^{(N)}(o_b)}\psi_j^{(N)}(t_b)\in \R
\]
for any $N$, $j$ and $b\in B_N$.

\subsection{Eigenfunctions on the edges}\label{e:sols}
Let $\cQ$ be a quantum graph.
Given an oriented edge $b\in B$ and $\gamma\in \C$, let $C_{\gamma, b}(x_b)$ and $S_{\gamma,b}(x_b)$ be a basis of solutions of the equation
\begin{equation}\label{e:EigenProblem}
-\psi''(x_b) + W_b(x_b) \psi(x_b) = \gamma \psi(x_b)
\end{equation}
satisfying
\[
\begin{pmatrix} C_{\gamma,b}(0) & S_{\gamma,b}(0) \\ C_{\gamma,b}'(0) & S_{\gamma,b}'(0) \end{pmatrix} = \begin{pmatrix} 1 & 0 \\ 0 & 1 \end{pmatrix} .
\]
If $W_b\equiv 0$, these are the familiar cosine and sine functions, hence 
our notation.
Note that any solution $\psi$ of \eqref{e:EigenProblem} satisfies
\[
\begin{pmatrix} \psi(L_b)\\ \psi'(L_b) \end{pmatrix} = M_{\gamma}(b) \begin{pmatrix} \psi(0) \\ \psi'(0) \end{pmatrix} \qquad \text{where } M_{\gamma}(b) = \begin{pmatrix} C_{\gamma,b}(L_b) & S_{\gamma,b}(L_b) \\ C'_{\gamma,b}(L_b) & S'_{\gamma,b}(L_b) \end{pmatrix}.
\]

For all $b\in B$ and all $z\in \C$, we have
\begin{equation}\label{eq:SBar}
S_{\overline{\gamma},b} = \overline{S_{\gamma,b}},
\end{equation}
since both functions satisfy the same differential equation with the same boundary conditions.

For any $x_b\in [0,L_b]$, the maps $\gamma \mapsto C_{\gamma,b}(x_b)$ and $\gamma \mapsto S_{\gamma,b}(x_b)$ are holomorphic functions. A proof of this fact can be found, for instance, in \cite[Chapter 1]{PT87}.

Similarly, using the symmetry of the potential, we note that for any $\gamma \in \C$, we have the ``trigonometric'' relations (similar to the usual ones satisfied by $\cos$ and $\sin$)
\begin{equation}\label{e:simplif}
\begin{aligned}
S_{\gamma,b}(L_b)C_{\gamma,b}(x_b)-C_{\gamma,b}(L_b) S_{\gamma,b}(x_b) &= S_{\gamma,b}(L_b-x_b)\,,\\
S_{\gamma,b}'(L_b)C_{\gamma,b}(x_b)-C_{\gamma,b}'(L_b) S_{\gamma,b}(x_b) &= C_{\gamma,b}(L_b-x_b)\,,
\end{aligned}
\end{equation}
since the two pairs of functions satisfy the same equation with the same boundary conditions (see \cite[\S 2.1]{QEQGEQ} for more details). In the sequel, we will write $S_{\gamma}(x_b), C_{\gamma}(x_b)$ for $S_{\gamma, b}(x_b), C_{\gamma, b}(x_b)$ to lighten notations.

\subsection{Main result}
Let $\cQ_N=(V_N,E_N,L_N,W_N,\alpha_N)$ be a sequence of quantum graphs, each with $N$ vertices. We denote by $G_N=(V_N, E_N)$ the underlying discrete graphs.

Our first two assumptions concern the geometry of $G_N=(V_N, E_N)$:

\medskip

\raisebox{1.2\ht\strutbox}{\hypertarget{lab:exp}{}}\textbf{(EXP)}
The sequence $(G_N)$ forms an expander family. That is, if $\cA_{G_N}$ is the adjacency matrix\footnote{Recall that the adjacency matrix acts on $f\in \ell^2(V)$ by $\big{(}\mathcal{A}f\big{)}(v) = \sum_{w\sim v} f(w)$.} on $G_N$, then the operator $P_N = \frac{1}{d_N} \cA_{G_N}$, has a uniform spectral gap in $\ell^2(V_N)$. More precisely, the eigenvalue $1$ of $P_N$ is simple, and the spectrum of $P_N$ is contained in $[-1+\beta, 1-\beta]\cup \{1\}$, where $\beta>0$ is independent of $N$.

\medskip

\raisebox{1.2\ht\strutbox}{\hypertarget{lab:bst}{}}\textbf{(BST)}
For all $r>0$,
\[
\lim_{N \to \infty} \frac{|\{x \in V_N : \rho_{G_N}(x)<r\}|}{N} = 0 \, ,
\]
where $\rho_{G_N}(x)$ is the \emph{injectivity radius} at $x$, i.e.\ the largest $\rho$ such that the ball $B_{G_N}(x,\rho)$ is a tree. Here, because of our uniformity assumption \Data{},
it does not matter to choose a discrete distance or a continuous one for this property.

We will in addition suppose that $\cQ_N$ converges in the sense of Benjamini-Schramm to some probability measure $\mathbb{P}$ on the set of rooted graphs satisfying \Data{} --- as we may always do, up to extracting a subsequence \cite[Corollary 3.6]{BSQG}. Assumption \BST{} becomes equivalent to asking that $\mathbb{P}$ is supported on the set of quantum trees, i.e.\ quantum graphs without cycles. In a more probabilistic language, $\cQ_N$ converges to a random rooted quantum tree.

\medskip

The next two hypotheses can be seen as a condition of \emph{spectral delocalization}. Indeed, they imply that $\mathbb{P}$-almost all quantum trees have purely absolutely continuous spectrum in $I$, see \cite[Theorem A.6]{AC}. 

\medskip

\raisebox{1.2\ht\strutbox}{\hypertarget{lab:green}{}}\textbf{(Green)} 
There exists a bounded open interval $I_1\subset \R$ such that for all $s>0$,
\[
\sup_{\lambda\in I_1,\eta\in (0,1)} \expect_{\prob}\left(\left|\Im \hat{R}_{\lambda+\ii\eta}^{+}(o_{b})\right|^{-s} + \left|\Im \hat{R}_{\lambda+\ii\eta}^{+}(o_{\hat{b}})\right|^{-s} \right)<\infty \,,
\]
where $\hat{R}^+$ is the \emph{Weyl-Titchmarsh function} (defined in Section \ref{subsecGreen} below).

This assumption implies in particular that the Green's functions of the Schr\"odinger operator on the infinite tree has finite moments; see \S~\ref{sec:greencons}.

As the proof will show, we actually only need \Green{} to hold for all $0<s<s_0$ for some finite $s_0$ which can in principle be made explicit and is not too big; we chose the above formulation for comfort.

\medskip

Our last assumption is that $I_1$ avoids the ``Dirichlet spectrum''.

\medskip

\raisebox{1.2\ht\strutbox}{\hypertarget{lab:non-dirichlet}{}}%
\textbf{(Non-Dirichlet)}
Let
\[
\mathscr{D} = \bigcup_N\bigcup_{b\in B_N} \{\lambda\in \R: S_{\lambda}(L_b)=0\}\,.
\]

Then we assume
\[
I_1\cap \mathscr{D} = \emptyset\,,
\]
so any compact $\overline{I}\subset I_1$ is isolated from all Dirichlet values. In the applications we have in mind, $\overline{\mathscr{D}}$ is a discrete subset of $\R$, or an $\epsilon$-neighborhood  of a discrete subset.

It follows that there exists $C_{\mathrm{Dir}}>0$ such that for all $\lambda\in I$,
we have 
\[
|S_\lambda(L_b)|\geq C_{\mathrm{Dir}}.
\] 
By continuity, this implies the existence of $C_{\mathrm{Dir}}'>0$ and $0<\eta_{\mathrm{Dir}}\leq 1$ such that for all $\lambda\in I$, $\eta\in [0,\eta_{\mathrm{Dir}}]$,
\begin{equation}\label{eq:lowerDir}
|\Re S_{\lambda+\ii\eta}(L_b)|\geq C'_{\mathrm{Dir}}.
\end{equation}

We also note that $S_{\gamma}(L_b)\neq 0$ for any $\gamma\in \C^+=\{z:\Im z>0\}$, since otherwise the self-adjoint operator $Df=-f''+W_bf$ on $[0,L_b]$ with Dirichlet boundary conditions would have a complex eigenvalue $\gamma$ corresponding to $S_{\gamma}(x)$.

\begin{thm}\label{thm:qeqg}
Let $\cQ_N$ be a sequence of quantum graphs satisfying \emph{\Data{}} for each $N$, such that \emph{\EXP{}}, \emph{\BST{}}, \emph{\Green{}} and 
\emph{\NonDirichlet{}} hold true on the interval $I_1$. Fix an interval $I$ such that $\overline{I}\subset I_1$.

Let $(\psi_j^{(N)})_{j\in\N}$ be an orthonormal basis of eigenfunctions of $H_{\cQ_N}$. Then for any sequence of functions $f_N\in L^\infty(\mathcal{G}_N)$ satisfying $\|f_N\|_{\infty}  \le 1$, we have
\begin{equation}\label{e:main'}
\lim_{\eta \downarrow 0} \lim_{N\to\infty} \frac{1}{\mathbf{N}_N(I)} \sum_{\lambda_j^{(N)}\in I} \left|\langle \psi_j^{(N)}, f_N\psi_j^{(N)}\rangle_{L^2(\cG_N)} -\langle f_N\rangle_{\gamma_j^{(N)}} \right| =0 \,,
\end{equation}
where 
$\langle \psi_j, f_N\psi_j\rangle_{L^2(\cG_N)} = \int_{\mathcal{G}_N} f_N(x) |\psi_j(x)|^2\,\dd x$, $\gamma_j^{(N)} = \lambda_j^{(N)}+\ii\eta$ and
\begin{equation}\label{eq:DefBracket}
\langle f_N\rangle_{\gamma_j} = \frac{\int_{\mathcal{G}_N}f_N(x)\Im\tilg_N^{\gamma_j}(x,x)\,\dd x}{\int_{\mathcal{G}_N}\Im \tilg_N^{\gamma_j}(x,x)\,\dd x}.
\end{equation}
\end{thm}

Here, $\tilg_N^{\gamma}(x,y)$ is the Green kernel of the operator $H_{\widetilde{\cQ}_N}$, defined on the universal cover of $\cQ_N$. It will be defined precisely in Section~\ref{subsec:Univ}. We will see in Section~\ref{subsecGreen} that its imaginary part is always positive.  In the special case of regular equilateral quantum graphs treated in \cite{QEQGEQ}, the universal covering quantum tree is the regular equilateral quantum tree $\mathbf{T}_q$, independent of $N$, so the Green function $\tilg^{\gamma}_N(v,w) = G^{\gamma}_{\mathbf{T}_q}(v,w)$ itself is independent of $N$.

\begin{rem}
By the Cauchy-Schwarz inequality, if $f_N\geq 0$, we have
\[
\left( \int_{\mathcal{G}_N} (f_N(x))^{1/2} \dd x\right)^{2} \leq \left(\int_{\mathcal{G}_N}f_N(x)\Im\tilg_N^{\gamma_j}(x,x)\,\dd x\right) \left(\int_{\mathcal{G}_N}\left(\Im\tilg_N^{\gamma_j}(x,x)\right)^{-1} \,\dd x\right).
\]

By  Corollary~\ref{cor:ConfinedTilMay11}, there exists $C>0$ independent of $N$ such that
\[
\sup \limits_{\lambda\in I_1, \eta \in (0,\eta_{\mathrm{Dir}})} \limsup\limits_{N\longrightarrow \infty} \left|\frac{1}{N}\int_{\mathcal{G}_N}\left(\Im\tilg_N^{\gamma}(x,x)\right)^{\pm 1} \,\dd x \right|\leq C.
\]
 This implies that, when $f_N$ is non-negative,
\[
\liminf_{N\To\infty} \left\langle f_N\right\rangle_{\gamma_j} \geq \frac{1}{C^2}\left(\frac{1}{N} \int_{\mathcal{G}_N} (f_N(x))^{1/2} \dd x\right)^{2}.
\]
For example, if $f_N = \chi_{\Lambda_N}$ with $|\Lambda_N|=\alpha N$, this gives a lower bound $\frac{\alpha^2}{C^2}>0$. As $\Lambda_N$ is arbitrary (for fixed cardinality), Theorem \ref{thm:qeqg} is really a delocalization result: for most $\psi_j$, we cannot have $|\psi_j(x)|^2$ concentrating on a portion of $\cG_N$ of cardinality $o(N)$.
\end{rem}

\subsubsection*{Quantum ergodicity for integral operators}

In a weak sense, the previous theorem asserts that $|\psi_j^{(N)}(x)|^2$ behaves asymptotically like $\frac{\Im \tilde{g}_N^{\lambda_j}(x,x)}{\int_{\cG_N} \Im\tilde{g}_N^{\lambda_j}(x,x)\,\dd x}$. More generally, we have a quantum ergodicity result involving integral operators. The aim here is to study the eigenfunction correlator $\overline{\psi_j^{(N)}(x)}\psi_j^{(N)}(y)$.

In general, an integral operator $\mathcal{K}$ on $L^2(\cG)$ takes the form
\[
(\mathcal{K}\psi)_b(x_b) = \sum_{b'\in B} \int_0^{L_{b'}} \mathcal{K}_{b,b'}(x_b,y_{b'})\psi_{b'}(y_{b'})\,\dd y_{b'} \,,
\]
with the condition that for all $b,b'\in B$ and almost all $x_b\in [0, L_b]$, $y_{b'}\in [0,L_{b'}]$, we have
\[
 \mathcal{K}_{\hat{b},b'}(L_b-x_b,y_{b'}) = \mathcal{K}_{b,b'}(x_b,y_{b'})
\]
 and similarly for the second argument.
 In the sequel, we will denote by $\mathrm{B}_k$ the set of non-backtracking paths of length $k$ (see \eqref{eq:defNonBackPath} below for a precise definition), and work with the spaces of operators (indexed by $k\in \N$)
 
\[
 \begin{aligned}
\mathscr{K}_k := \Big\{& \mathcal{K} : L^2(\cG) \longrightarrow L^2(\cG) \mid \forall b,b'\in B,  \mathcal{K}_{b,b'} \equiv 0\\
& \text{ unless there exists } (b_1,\dots,b_k)\in \mathrm{B}_k \text{ with } b_1=b, b_k=b'  \Big\}.
\end{aligned}
\]

Thus, $\mathscr{K}_k$ is the space of operators with kernel supported on edges connected by a non-backtracking path of length $k$. Note that any integral operator with a bounded compactly supported kernel can be written as a finite linear combination of operators in $\mathscr{K}_k$ for various $k$. When we want to insist that these operators live on the graph $\cQ_N$ indexed by $N$, we will denote this space by $\mathscr{K}_k^{(N)}$.

\begin{thm}\label{thm:integralver}
Let $k\geq 0$. Under the assumptions of Theorem~\ref{thm:qeqg}, let $(\mathcal{K}_N)_{N\in \N}$ be a sequence of operators with $\mathcal{K}_N \in \mathscr{K}_k^{(N)}$, such that $|K_{N, b,b'}(x,y)|\leq 1$ for every $N\in \N$. Then
\begin{equation}\label{eq:MainStatement}
\lim_{\eta \downarrow 0} \lim_{N\to \infty} \frac{1}{\mathbf{N}_N(I)} \sum_{\lambda_j^{(N)}\in I} \left| \langle \psi_j^{(N)}, \mathcal{K}_N \psi_j^{(N)}\rangle_{L^2(\cG_N)} - \langle \mathcal{K}_N\rangle_{\gamma_j^{(N)}}\right| = 0\,,
\end{equation}
where $\langle \psi_j,\cK\psi_j \rangle = \sum\limits_{(b_1;b_k)\in \mathrm{B}_k}\int_0^{L_{b_1}}\int_0^{L_{b_k}}\cK(x_{b_1},y_{b_k})\overline{\psi_j(x_{b_1})}\psi_j(y_{b_k})\,\dd x_{b_1}\dd y_{b_k}$,  $\gamma_j^{(N)}=\lambda_j^{(N)}+\ii\eta$, and 
\[
\langle \mathcal{K}\rangle_{\gamma} =\frac{\sum_{(b_1;b_k)}\int\int \cK(x_{b_1},y_{b_k})\Im \tilg^{\gamma}(x_{b_1},y_{b_k})\,\dd x_{b_1}\dd y_{b_k}}{\int_{\mathcal{G}_N}\Im \tilg^{\gamma}(x,x)\,\dd x}.
\]
\end{thm}

This says that in a weak sense, when $N$ gets large, the eigenfunction correlator $\overline{\psi_j(x)}\psi_j(y)$ looks like the quotient of spectral densities $\frac{\Im \tilg^{\lambda_j}_N(x,y)}{\int_{\mathcal{G}_N}\Im \tilg^{\lambda_j}_N(x,x)\,\dd x}$ on the universal cover.

\subsection{Examples}\label{sec:examp}

\subsubsection{N-lifts}
An important example is when $\cQ_N$ is some (connected) $N$-lift of a compact quantum graph $\cQ_1$. In other words, the underlying graph\footnote{We refer to \cite[Chapter 5]{Sunada} for more background on coverings of finite graphs.} $G_N$ is an $N$-fold covering over $G_1$ and the data is lifted naturally $L_{(v,w)} = L_{(\pi_N v,\pi_Nw)}$, $W_{(v,w)}=W_{(\pi_Nv,\pi_Nw)}$, $\alpha_v = \alpha_{\pi_N v}$, where $\pi_N:G_N\To G_1$ is the covering projection.

It is known that $N$-lifts -- when picked randomly -- are typically connected and most of their points have a large injectivity radius -- see \cite[Lemma 24]{Bornew}, \cite[Lemma 9]{BDGHT}. More precisely, condition \BST{} holds generically. It is also known that they are typically expanders; see \cite{Fri03,Pu}. Thus, our assumptions are generic.

It is known that such $(\cQ_N)$ converge in the Benjamini-Schramm sense to a deterministic limit, namely the universal covering tree $\cT = \widetilde{\cQ}_1$ with a random root (see \cite{BSQG}). More precisely $\cQ_N$ converges to the random rooted quantum tree defined by the measure

\[
\mathbb{P} = \frac{1}{\sum_{b\in B_1} L_b}\sum_{b\in B_1} \int_0^{L_b} \delta_{[\mathbb{T},\tilde{L}^1,\tilde{W}^1,\tilde{\alpha}^1,(\tilde{b},x_b)]}\,\dd x_b\,,
\]
where $(L^1,W^1,\alpha^1)$ is the data on the base graph $\cG_1$ and $\T=\widetilde{G}_1$ is the combinatorial tree underlying $\cT$. In particular
\[
\expect_{\prob}(|\Im \hat{R}_{\lambda+\ii\eta}^{\pm}(o_{b})|^{-s} + |\Im \hat{R}_{\lambda+\ii\eta}^{\pm}(o_{\hat{b}})|^{-s} ) = \frac{2}{\sum_{b\in B_1}L_b}\sum_{b\in B_1}L_b\,|\Im \hat{R}_{\lambda+\ii\eta}^{\pm}(o_{b})|^{-s} .
\]
Also note that in this example, $\mathscr{D} = \cup_{b\in B_1}\{\lambda\in \R:S_{\lambda}(L_b)=0\}$.

We showed in \cite{AC} that the spectrum of $H_{\cT}$ consists of bands of pure AC spectra along with a possible set of discrete eigenvalues (outside the bands). We also showed that within the bands, the limits $\hat{R}_{\lambda+\ii0}^{\pm}(o_b)$ exist, are finite and satisfy $\Im \hat{R}_{\lambda+\ii 0}^{\pm}(o_b)>0$. It follows that \Green{} is satisfied on any compact $\overline{I}\subset I_1$, where $I_1$ is some AC band.

Theorem~\ref{thm:integralver} thus tells us that such $(\cQ_N)$ are quantum-ergodic. This result can be regarded as a non-regular, non-equilateral generalization of \cite{QEQGEQ}.

\subsubsection{Random quantum graphs}

We may also consider weak random perturbations $\cQ_{\omega_N}$ of the previous example. We leave the precise definition to \cite[Section 8.4]{BSQG}, but essentially one endows the graphs with random independent, identically-distributed (i.i.d.) lengths $L_e^{\omega}$ and i.i.d.\ coupling constants $\alpha^{\omega}_v$. Note that here only the combinatorial graph $G_N$ covers $G_1$, the data on each $\cQ_{\omega_N}$ is entirely random, so each has its own universal cover. Assuming \BST{} holds (which is true generically as previously mentioned), we calculated the limit measure in \cite{BSQG}. In particular, we get
\[
\expect_{\prob}(|\Im \hat{R}_{\lambda+\ii\eta}^{\pm}(o_{b})|^{-s} + |\Im \hat{R}_{\lambda+\ii\eta}^{\pm}(o_{\hat{b}})|^{-s} ) = \frac{2}{|B_1|\mathbf{E}(L_b^{\omega})}\sum_{b\in B_1}\mathbf{E}(L_b^{\omega}|\Im \hat{R}_{\lambda+\ii\eta}^{\pm}(o_{b})|^{-s} ),
\]
where $\mathbf{E}$ is the expectation with respect to the random data $(L_b^{\omega},\alpha^{\omega})$ and the distributions are assumed to be the same for each $b\in B_1$. In other words, this is a random perturbation of an equilateral model on $\cQ_N$ (more general situations can be considered).

We showed in \cite{AC} that if the perturbation is weak enough, then the bands of AC spectra remain stable, and \Green{} holds in such bands. Note that here we assume there is no edge potential $W$. So $\mathscr{D} = \cup_{n\ge 0} [\frac{\pi^2n^2}{(L-\eps)^2},\frac{\pi^2n^2}{(L+\eps)^2}]$, where $\eps$ is the small disorder window around the unperturbed length $L$. We technically need the coupling constants to be nonnegative with a H\"older distribution. This can be seen as a result of Anderson (spectral) delocalization, strengthening earlier results in \cite{ASW06}.

Theorem~\ref{thm:integralver} implies that almost surely, the eigenfunctions of $\cQ_{\omega_N}$ are quantum ergodic. In other words, Anderson (spatial) delocalization holds, in addition to spectral delocalization.

\section{Preliminary constructions and notation}\label{sec:Notation}

\subsection{The Green function on a quantum tree}\label{subsecGreen}
Let $\mathbf{T}=(V,E,L,W,\alpha)$ be a quantum tree, i.e.\ a quantum graph such that $(V,E)$ is a tree, with underlying metric graph $\mathcal{T}$. We describe here the functional equations satisfied by the Green function on $\mathbf{T}$, due to the topological fact that trees are disconnected by removing a point. While such relations are well-known for discrete laplacians on trees, they have been less exploited for quantum trees. This paragraph builds on the work of Aizenman-Sims-Warzel \cite{ASW06}.

If $b\in B(\T)$ we denote $\T_b^{\pm}$ the two subtrees obtained by removing $b$, more precisely $t_b\in \T_b^+$ while $o_b\in \T_b^-$. Let $\mathbf{T}_b^{\pm}$ be the induced quantum trees and $x_b\in [0, L_b]$. If $x=(b, x_b)$ is the corresponding point in $\T$, we define $\mathbf{T}_{x}^+$ as the quantum tree $[x_b,t_b]\cup \mathbf{T}_b^+$ see \cite{AC} for a more precise definition. The quantum tree $\mathbf{T}_x^-$ is defined in a similar fashion.

Let us define $H_{\mathbf{T}_x^{\pm}}^{\max}$ on $\mathcal{T}_x^{\pm}$ to be the Schr\"odinger operator $-\Delta+W$ with domain $D(H_{\mathbf{T}_x^{\pm}}^{\max})$, the set of $\psi\in H^2(\mathbf{T}_x^{\pm})$ satisfying $\delta$-conditions on inner vertices of $\mathbf{T}_x^{\pm}$.  Note that $H_{\mathbf{T}_x^\pm}^{\max}$ is not self-adjoint, due to the absence of boundary condition at the root $x$. By \cite[Theorem 2.1]{ASW06},  for any $\gamma\in\C^+ =\{z\in \C; \Im z>0\}$, there are unique eigenfunctions $V_{\gamma;x}^+\in D(H_{\mathbf{T}_x^+}^{\max})$, $U_{\gamma;x}^-\in D(H_{\mathbf{T}_x^-}^{\max})$ of $H_{\mathbf{T}_x^{\pm}}^{\max}$, for the eigenvalue $\gamma$, satisfying $U_{\gamma;x}^-(x)=V_{\gamma;x}^+(x)=1$.

One can use the functions $U^-_{\gamma},V^+_{\gamma}$ to construct the Green's function $G^\gamma$ of $H_{\mathbf{T}}$, see \cite[Lemma 2.1]{AC}. For our purposes, we use them to define the \emph{Weyl-Titchmarch} functions \cite{ASW06} as follows: if $x\in \mathbf{T}_o^+\cap\mathbf{T}_v^-$,
\begin{equation}\label{e:WT}
R^+_{\gamma}(x) = \frac{(V_{\gamma;o}^+)'(x)}{V_{\gamma;o}^+(x)} \quad \text{and} \quad R^-_{\gamma}(x) =  \frac{-(U_{\gamma;v}^-)'(x)}{U_{\gamma;v}^-(x)} \,.
\end{equation}
This does not depend on $o,v$. Given an oriented edge $b=(o_b,t_b)$, we define
\begin{equation}\label{e:zetadef}
\zeta^{\gamma}(b) = \frac{G^{\gamma}(o_b,t_b)}{G^{\gamma}(o_b,o_b)} \,.
\end{equation}

This quotient of Green kernels will appear in the definition of the non-backtracking eigenfunctions. See \cite[\S~2]{AC} for more comments on this quantity.

Given an oriented edge $b$, let $\cN_b^+$ be the set of outgoing bonds from $b$, and let  $\cN_b^-$ be the set of incoming bonds to $b$ i.e.
\begin{equation}\label{eq:DefOutInBond}
\begin{aligned}
\cN_b^+ &:= \left\{b'\in B ; o_{b'} = t_b, b' \neq \hat{b}\right\}\\
\cN_b^- &:= \left\{b'\in B ; t_{b'} = o_b, b' \neq \hat{b}\right\} .
\end{aligned}
\end{equation}
(Later these definitions will apply to more general graphs than trees.)

The following lemma gives a quantum graph analog for the classical recursive identities of Green's functions on discrete trees. It tells us that the functions $\zeta^\gamma, C_\gamma$ and $S_\gamma$ can be used as building blocks to understand the function $G^\gamma$. In particular, \eqref{e:greenmul} is the well-known multiplicative property of the Green function on a tree. See \cite[Section 2]{AC} for a proof, and Appendix \ref{sec:greenquan} for a complement.

\begin{lem}\label{lem:IdentitiesGreen}
Let $\gamma\in \C^+$. We have the following relations between $\zeta^{\gamma}$ and the WT functions $R_{\gamma}^{\pm}$:
\begin{equation}\label{e:zetawt}
\zeta^{\gamma}(b) = C_{\gamma}(L_b) + R_{\gamma}^+(o_b)S_{\gamma}(L_b)\,, \qquad \zeta^{\gamma}(\hat{b}) = S_{\gamma}'(L_b) + R_{\gamma}^-(t_b) S_{\gamma}(L_b) \,,
\end{equation}
\begin{equation}\label{e:r+-id}
R^+_{\gamma}(t_b) = \frac{S_{\gamma}'(L_b)}{S_{\gamma}(L_b)} - \frac{1}{S_{\gamma}(L_b)\zeta^{\gamma}(b)}\,, \qquad R_{\gamma}^-(o_b) = \frac{C_{\gamma}(L_b)}{S_{\gamma}(L_b)} - \frac{1}{S_{\gamma}(L_b) \zeta^{\gamma}(\hat{b})} \,.
\end{equation}
Moreover,
\begin{equation}\label{e:1}
 \frac{1}{\zeta^{\gamma}(b) S_{\gamma}(L_b)} + \sum_{b^+\in \cN_b^+} \frac{\zeta^{\gamma}(b^+)}{S_{\gamma}(L_{b^+})} = \sum_{b^+\in \cN_b^+} \frac{C_{\gamma}(L_{b^+})}{S_{\gamma}(L_{b^+})} + \frac{S_{\gamma}'(L_b)}{S_{\gamma}(L_b)} + \alpha_{t_b} \,,
\end{equation}
\begin{equation}\label{e:zetainv}
\frac{1}{\zeta^{\gamma}(b)} - \zeta^{\gamma}(\hat{b}) = \frac{S_{\gamma}(L_b)}{G^{\gamma}(t_b,t_b)}\,, \qquad \frac{\zeta^{\gamma}(\hat{b})}{\zeta^{\gamma}(b)} = \frac{G^{\gamma}(o_b,o_b)}{G^{\gamma}(t_b,t_b)} \,,
\end{equation}
\begin{equation}\label{e:midgreen}
\frac{-1}{G^{\gamma}(o_b,o_b)} = R_{\gamma}^+(o_b) + R_{\gamma}^-(o_b)
\end{equation}
and
\begin{equation}\label{e:2}
 \sum_{b^+\in \cN_b^+} \frac{C_{\gamma}(L_{b^+})}{S_{\gamma}(L_{b^+})} + \frac{S_{\gamma}'(L_b)}{S_{\gamma}(L_b)} + \alpha_{t_b} = \sum_{t_{b'}\sim t_b} \frac{\zeta^{\gamma}(b')}{S_{\gamma}(L_{b'})} + \frac{1}{G^{\gamma}(t_b,t_b)} \,,
\end{equation}
where $b'=(t_b,t_{b'})$. Given a non-backtracking path $(v_0;v_k)\in \T$, if $b_j = (v_{j-1},v_j)$, then
\begin{equation}\label{e:greenmul}
G^{\gamma}(v_0,v_k) = G^{\gamma}(v_0,v_0)\zeta^{\gamma}(b_1)\cdots \zeta^{\gamma}(b_k) = G^{\gamma}(v_k,v_k) \zeta^{\gamma}(\hat{b}_1) \cdots \zeta^{\gamma}(\hat{b}_k)\,.
\end{equation}
Finally, for any path $(v_0;v_k)\in \T$,
\begin{equation}\label{e:sym}
G^{\gamma}(v_0,v_k) = G^{\gamma}(v_k,v_0)\,.
\end{equation}
\end{lem}

The following lemma is an important result on the properties of the Weyl-Titchmarsh functions \eqref{e:WT} and the fact that they are involved in ``currents'' passing through the edges from some fixed arbitrary source $\ast$ (the ``current'' is $I_{\ast}^{\lambda}(b)=|G^{\lambda+\ii0}(\ast,o_b)|^2\Im R_{\lambda+\ii0}^+(o_b)$).

\begin{lem}\label{lem:ASW}
The functions $F(\gamma) = R_{\gamma}^+(o_b)$, $R_{\gamma}^-(t_b)$ and $G^{\gamma}(v,v)$ are Herglotz functions: $\Im F(\gamma) \ge 0$ for $\gamma\in\C^+$. Moreover, we have the following ``current'' relations:
\begin{equation}\label{e:ASW}
\sum_{b^+\in \cN_b^+} \Im R_{\gamma}^+(o_{b^+}) \le \frac{\Im R_{\gamma}^+(o_b)}{|\zeta^{\gamma}(b)|^2} \qquad \text{and} \qquad \sum_{b^-\in \cN_b^-} \Im R_{\gamma}^-(t_{b^-}) \le \frac{\Im R_{\gamma}^-(t_b)}{|\zeta^{\gamma}(\hat{b})|^2} \,.
\end{equation}
Equality holds in both cases if $\Im \gamma=0$, whenever defined. 

More precisely, we have
\begin{equation}\label{e:cur1}
\sum_{b^+\in \cN_b^+} \Im R_{\gamma}^+(o_{b^+}) = \frac{\Im R^+_{\gamma}(o_b)}{|\zeta^{\gamma}(b)|^2} - \frac{\Im \gamma}{|\zeta^{\gamma}(b)|^2} \int_{0}^{L_b} |\xi_{+}^{\gamma}(x_b)|^2\,\dd x_b \,,
\end{equation}
\begin{equation}\label{e:cur2}
\sum_{b^-\in \cN_b^-} \Im R_{\gamma}^-(t_{b^-}) = \frac{\Im R_{\gamma}^-(t_b)}{|\zeta^{\gamma}(\hat{b})|^2} - \frac{\Im \gamma}{|\zeta^{\gamma}(\hat{b})|^2} \int_{0}^{L_b} |\xi_{-}^{\gamma}(x_b)|^2\,\dd x_b \,,
\end{equation}
where
\begin{equation}\label{eq:DefXi}
\begin{aligned}
\xi_{+}^{\gamma}(x_b) &= \frac{V_{\gamma;o}^+(x_b)}{V_{\gamma;o}^+(o_b)} = C_{\gamma}(x_b) + R_{\gamma}^+(o_b)S_{\gamma}(x_b) \\
\xi_{-}^{\gamma}(x_b) &= \frac{U_{\gamma;v}^-(x_b)}{U_{\gamma;v}^-(t_b)} = C_{\gamma}(L_b-x_b) + R_{\gamma}^-(t_b)S_{\gamma}(L_b-x_b) \,.
\end{aligned}
\end{equation}
\end{lem}

See Appendix~\ref{sec:greenquan} for a proof.

\subsection{Operators on graphs}\label{subsec:OpOnGr}
Let now $\cQ$ be a finite quantum graph (typically one in our family $\cQ_N$) with vertex set $V$ and bond set $B$. When $\cQ=\cQ_N$ the corresponding notions will be indexed by $N$ (e.g. $B_N$, $V_N$), but we will often tend to drop the index $N$ from the notation.

In the study of quantum ergodicity, we need to go back and forth between ``classical observables'' (i.e.\ functions on a classical phase space) and ``quantum observables'' (i.e.\ operators on a Hilbert space). Here this will be done in a simple-minded way: starting from a function on the set of non-backtracking paths, we explain how to build an operator on $\ell^2(B)$ or $\ell^2(V)$.

If $k\geq 1$, we denote by $\mathrm{B}_k$ the set of non-backtracking paths in $B^k$ of length $k$: 
\begin{equation}\label{eq:defNonBackPath}
\mathrm{B}_k:= \big{\{}(b_1,\ldots,b_k)\in B^k ~ ;~  \forall i=1,\ldots, k-1, t_{b_i} = o_{b_{i+1}} \text{ and } t_{b_{i+1}} \neq o_{b_i}\big{\}}.
\end{equation}

If $b_0\in B$, we will also write
\begin{align*}
\mathrm{B}^{b_0}_k:= \big{\{}(b_{1},\ldots,b_{k})\in B^k ~ \text{such that} ~ (b_0,\ldots,b_{k})\in \mathrm{B}_{k+1} \big{\}}\\
\mathrm{B}_{k,b_0}:= \big{\{}(b_{-k},\ldots,b_{-1})\in B^k ~ \text{such that} ~ (b_{-k},\ldots,b_{0})\in \mathrm{B}_{k+1} \big{\}}.
\end{align*}

Note that $\mathrm{B}_1^{b_0} = \cN_{b_0}^+$ and $\mathrm{B}_{1,b_0} = \cN_{b_0}^-$.

In the sequel, we will often write $(b_1;b_k)$ instead of $(b_1,\ldots,b_k)$ to lighten the notation.

For all $k\geq 1$, we also define
\[
\mathscr{H}_k := \C^{\mathrm{B}_k} .
\]
If $K\in \mathscr{H}_k$, then $K$ is a map from $\mathrm{B}_k$ to $\C$, and we extend it to a map $B^k\to \C$ by zero on $B^k\setminus \mathrm{B}_k$. This extension will still be denoted by $K$.

If $K\in \mathscr{H}_k$, we define an operator $K_B : \ell^2(B)\To \ell^2(B)$ by
\begin{equation}\label{eq:DefKB}
\forall f\in \ell^2(B),\forall b_1\in B, \quad (K_B f)(b_1) = \sum_{(b_2 ; b_k)\in \mathrm{B}^{b_1}_{k-1}} K(b_1;b_k) f(b_k).
\end{equation}
In particular, if $K\in \mathscr{H}_1$, then $K_B$ is a diagonal operator.

For $k=0$, we will write $\mathrm{B}_0:= V$, and $\mathscr{H}_0 := \C^V$.

For all $k\geq 0$, we also define $K_G : \ell^2(V) \longrightarrow \ell^2(V)$ by 
\begin{align*}
\forall h\in \ell^2(V),\forall v\in V, \quad (K_G h)(v) &:= \sum_{\substack{(b_1 ; b_k)\in \mathrm{B}_{k}\\o_{b_1}=v}} K(b_1 ; b_k) h(t_{b_k}) \qquad \text{ if } k\geq 1\\
(K_G h)(v) &:=  (Kh)(v) = K(v) h(v)  \qquad\quad \text{ if } k=0.
\end{align*}

\subsection{Green functions notation}\label{subsec:Univ}
In the paper we consider a variety of quantum graphs, and we need to adopt a notation for the Green function of each of them: the sequence $\cQ_N$, their universal covers $\widetilde\cQ_N$, the limiting random quantum tree $\cT$.

Let us first define notations pertaining to universal covers. Let $\cQ$ be a quantum graph. Let $\tilde{G}= (\tilde{V}, \tilde{E})$ be the universal cover of the combinatorial graph $G$. We endow $\tilde{G}$ with the lifted data $\tilde{L}_b := L_{\pi b}$, $\tilde{W}_b := W_{\pi b}$ and $\tilde{\alpha}_v = \alpha_{\pi v}$, where $\pi:\tilde{G} \to G$ is the covering map. This yields a quantum tree $\widetilde{\cQ}:= (\tilde{V},\tilde{E},\tilde{L},\tilde{W},\tilde{\alpha})$, which is called the \emph{universal cover} of $\cQ$. The underlying metric graph of $\widetilde{\cQ}$ will be denoted by $\widetilde{\mathcal{G}}$. We then have a natural projection $\pi : \widetilde{\mathcal{G}} \longrightarrow \mathcal{G}$.

Throughout the paper, if $v,w\in V_N$ and $z\in \C\setminus \R$, we will write
\[
g_N^\gamma(v,w) := \left( H_{\cQ_N}-\gamma\right)^{-1}(v,w)
\]
for the Green function of the compact quantum graph $\cQ_N$.

Let $\tilg_N^\gamma$ be the Green's functions of $H_{\widetilde{\cQ}_N}$. 
Throughout the paper we will encounter quantities of the form $\tilg_N^\gamma(o_{b_i},t_{b_j})$ where $(b_1,\dots,b_k)$ is a fixed non-backtracking path in $G_N$ and $i,j\le k$. We define this as follows.

Given $(b_1;b_k)\in \mathrm{B}_k(\cQ_N)$, choose any lift $\tilde{b}_1 \in B(\widetilde{\cQ}_N)$ and let $(\tilde{b}_2;\tilde{b}_k)\in \mathrm{B}_{k-1}^{\tilde{b}_0}(\widetilde{\cQ})$ be the path such that $\pi(\tilde{b}_1,\dots,\tilde{b}_k)=(b_1,\dots,b_k)$. Then we define
\[
\tilg_N^\gamma(o_{b_1},t_{b_k}) := (H_{\widetilde{\cQ_N}}-\gamma)^{-1}(o_{\tilde{b}_1},t_{\tilde{b}_k})\,.
\]

This depends on the full path $(b_1,\dots,b_k)$ (although not apparent in our notation), however it does not depend on the choice of the lift $(\tilde{b}_1,\dots,\tilde{b}_k)$. We define $\tilg_N^\gamma(o_{b_i},t_{b_j}):=(H_{\widetilde{\cQ_N}}-\gamma)^{-1}(o_{\tilde{b}_i},t_{\tilde{b}_j})$ for $i,j\le k$, where $(\tilde{b}_1,\dots,\tilde{b}_k)$ is the lift we fixed. The definition extends naturally to $\tilg_N^\gamma(x,y)$ with $x\in b_i$ and $y\in b_j$.

Throughout the paper, we always let for $b\in B_N$,
\begin{equation}\label{e:reczet}
\zeta^\gamma(b) := \frac{\tilg^\gamma_N(o_b,t_b)}{\tilg^\gamma_N(o_b,o_b)}
\end{equation}
thus suppressing the index $N$, which should cause no confusion.

Similarly, the Weyl-Titchmarsh functions \eqref{e:WT} denoted by $R^{\pm}_{\gamma}(x)$ will stand (without index $N$) for the WT-functions of the universal covering tree $\widetilde{Q}_N$.

For the Benjamini-Schramm limiting random tree $\cT$, we use the notation
\[
G^\gamma(x,y) = (H_{\cT}-\gamma)^{-1}(x,y)
\]
for $x,y\in \cT$. For $\zeta$ and the WT-functions, we simply add a hat. More precisely, we let
\[
\hat{\zeta}^\gamma(b) := \frac{G^\gamma(o_b,t_b)}{G^\gamma(o_b,o_b)} 
\]
for $b\in B(\cT)$, and similarly denote the Weyl-Titchmarsh functions of $\cT$ by $\hat{R}_{\gamma}^\pm(x)$.

\subsection{A scalar product expression for boundary values of eigenfunctions}\label{sec:ScalarProd}
For each $\gamma\in \C$ and $b\in B_N$, let us define
\[
\begin{aligned}
\Sigma_1(\gamma;b)&:= \int_0^{L_b} |S_{\gamma}(x_b)|^2 \mathrm{d}x_b\\
\Sigma_2(\gamma;b)&:= \int_0^{L_b} S_{\gamma}(L_b-x_b) \overline{S_{\gamma}(x_b)} \mathrm{d}x_b.
\end{aligned}
\]

Note that, by the Cauchy-Schwarz inequality, $\Sigma_1^2(\gamma;b) - |\Sigma_2(\gamma;b)|^2 >0$. As the lower bound only depends on $L_b,W_b,\gamma$, we have $\Sigma_1^2(\gamma;b) - |\Sigma_2^2(\gamma;b)|\geq c>0$ in the compact set $(\gamma,L,W)\in (\overline{I}+\ii[0,1])\times \mathrm{Lip}_{\mathrm{M}}[\mathrm{m},\mathrm{M}]$, where $\mathrm{Lip}_{\mathrm{M}}[\mathrm{m},\mathrm{M}]$ denotes the set of $(L, f)$, where $m\leq L\leq M$, and $f$ is a Lipschitz function on $[0,L]$ with both norm and Lipschitz constant $\le \rmM$.

For each $\gamma\in \C$ and each $b\in B_N$, we denote by $S^b_{\gamma,+}$ and $S^b_{\gamma,-}$ the functions on $\mathcal{G}_N$ defined respectively as $x_b\mapsto S_{\gamma}(x_b)$ and $x_b\mapsto S_{\gamma}(L_b-x_b)$ on the edge $b$, and vanishing on the other edges.

If $\psi_\gamma$ satisfies $H_\cQ \psi_\gamma = \gamma \psi_\gamma$, and if the potentials are symmetric, then we have
\[
\psi_{\gamma}(x_b) = \frac{S_{\gamma}(L_b-x_b)}{S_{\gamma}(L_b)} \psi_\gamma(o_b) + \frac{S_{\gamma}(x_b)}{S_{\gamma}(L_b)}\psi_\gamma(t_b),
\]
 so that
\[
\begin{aligned}
S_{\gamma}(L_b) \langle S_{\gamma,+}^b, \psi_\gamma \rangle &= \psi_\gamma(o_b) \Sigma_2(\gamma;b) + \psi_\gamma(t_b) \Sigma_1(\gamma;b)\\
S_{\gamma}(L_b) \langle  S_{\gamma,-}^b, \psi_\gamma \rangle &= \psi_\gamma(o_b) \Sigma_1(\gamma;b) + \psi_\gamma(t_b) \overline{\Sigma_2(\gamma;b)}.
\end{aligned}
\]

We therefore have
\begin{equation}\label{e:Maxime}
\begin{aligned}
\psi_\gamma(o_b) &= \frac{S_\gamma(L_b)}{|\Sigma^2_2(\gamma;b)|-\Sigma^2_1(\gamma;b)} \left\langle  \Sigma_2(\gamma;b) S_{\gamma,+}^b-\Sigma_1(\gamma;b) S_{\gamma,-}^b, \psi_\gamma \right\rangle =: \langle Y_\gamma^b, \psi_\gamma \rangle_{L^2(\mathcal{G}_N)} \\
\psi_\gamma(t_b) &= \frac{S_\gamma(L_b)}{|\Sigma^2_2(\gamma;b)|-\Sigma^2_1(\gamma;b)}\left\langle  \overline{\Sigma_2(\gamma;b)} S_{\gamma,-}^b- \Sigma_1(\gamma;b) S_{\gamma,+}^b, \psi_\gamma \right\rangle  =: \langle Z_\gamma^b, \psi_\gamma \rangle_{L^2(\mathcal{G}_N)} 
\end{aligned}
\end{equation}

Note that, thanks to hypothesis \Data{}, we have
\[
\sup_N\max_{b\in B_N}\sup_{\gamma\in I+\ii[0,1]} \max (\|Y^b_\gamma\|_{L^2(\mathcal{G}_N)}, \|Z^b_\gamma\|_{L^2(\mathcal{G}_N)})  \le C_{I,\mathrm{M}}.
\]

That $\gamma\mapsto Z_\gamma^b$ is not analytic poses a technical
obstacle.  Since we will be using holomorphic tools, we prefer to use
\[
\widehat{\Sigma_1}(\gamma;b):= \int_0^{L_b} S_{\gamma}(x_b)^2 \dd x_b ,\qquad
\widehat{\Sigma_2}(\gamma;b):= \int_0^{L_b} S_{\gamma}(L_b-x_b) S_{\gamma}(x_b) \dd x_b\,,
\]
in other words the analytic functions that coincide with $\Sigma_1, \Sigma_2$ on the real line.

By Cauchy-Schwarz and by continuity, we may find $\eta_I, c_I>0$ such that for all $N$, $b\in B_N$, $\gamma\in\Omega_I:= I+ \ii [-\eta_I, \eta_I]$, we have
\[
\Re\left(\widehat{\Sigma^2_1}(\gamma;b)-\widehat{\Sigma^2_2}(\gamma;b)\right) >c_I.
\]

Therefore, we may define, for $\gamma \in \Omega_I$, the functions
\[
\widehat{Z}_\gamma^b:=  \frac{S_\gamma(L_b)}{\widehat{\Sigma^2_2}(\gamma;b)-\widehat{\Sigma^2_1}(\gamma;b)} \left(\widehat{\Sigma_2}(\gamma;b) S_{\gamma,-}^b- \widehat{\Sigma_1}(\gamma;b) S_{\gamma,+}^b\right).
\]

Note that $\Omega_I\ni \gamma\mapsto \widehat{Z}_\gamma^b\in L^2(\mathcal{G})$ is holomorphic, coincides with $Z_\gamma^b$ when $\gamma\in I$, and

\begin{equation}\label{eq:ZandHatZ}
  \left\| \widehat{Z}_\gamma^b - Z_\gamma^b \right\|_{L^2(\mathcal{G}_N)} \le \left\| \widehat{Z}_\gamma^b - \widehat{Z}_{\Re(\gamma)}^b \right\|_{L^2(\mathcal{G}_N)} + \left\| Z_\gamma^b - Z_{\Re(\gamma)}^b \right\|_{L^2(\mathcal{G}_N)} \le C_I'\eta_0
\end{equation}
uniformly in $\gamma\in \Omega_I$, $b\in B_N$, $N\in \N$.

\subsection{Notation for the remainders}
In all the paper, we will be dealing with quantities depending on $N$, and on a complex parameter $\gamma$, and we will use the following notation.
Let $f_N : \C^+ \To \C$ be a sequence of functions. We will write that $f_N = O_{N\to +\infty, \gamma}(1)$ if
\[
\sup \limits_{\eta_0\in (0,\eta_{\mathrm{Dir}})} \limsup_{N\longrightarrow \infty}\sup_{\lambda\in I_1} |f_N(\lambda+\ii\eta_0) | < \infty,
\]
with $\eta_{\mathrm{Dir}}$ as in \eqref{eq:lowerDir}.

Similarly, we will write $f_N= O_{N\to +\infty, \gamma}(\Im \gamma)$ if $\frac{f_N(\gamma)}{\Im \gamma}=O_{N\to +\infty, \gamma}(1)$.
In fact, most of time the imaginary part of $\gamma$ will be denoted $\eta_0$, so we will write $f_N= O_{N\to +\infty, \gamma}(\eta_0)$.

If $f_N$ depends on an additional parameter $\kappa$, we write $f_N=O_{N\to +\infty, \gamma}^{(\kappa)}(1)$.

If the quantity $f_N$ we consider depends on $N$ and $\eta_0$, but not on $\lambda$, we will write $f_N=O_{N\to +\infty, \eta_0}(1)$ or $f_N=O_{N\to +\infty, \eta_0}(\eta_0)$, with the same definition.

\begin{rem}
Our main statements, Theorems \ref{thm:qeqg} and \ref{thm:integralver} say that some quantity, divided by $\mathbf{N}_N(I)$ goes to zero as $N\longrightarrow + \infty$ followed by $\eta \downarrow 0$. We will recall in Appendix \ref{sec:BS}, and more precisely in \eqref{eq:UpperNumberEigen} and \eqref{eq:LowerNumerEigen} that, under the assumptions we make, there exist constants $C_1, C_2>0$ such that for all $N$ large enough,
\[
C_1 N \leq \mathbf{N}_N(I) \leq C_2 N.
\]
Therefore, in the course of the proof, when trying to show that a quantity divided by $\mathbf{N}_N(I)$ goes to zero, we will sometimes replace $\mathbf{N}_N(I)$ by $N$.
\end{rem}

\section{Non-backtracking eigenfunctions}\label{sec:nonback}

The quantum variance \eqref{e:main'} involves functions living on the \emph{metric graph} $\cG_N$. Through the main part of the paper, we shall prefer to work with quantum variances defined on the \emph{combinatorial graph} $G_N$. It is shown in Section~\ref{Sec:Reduction} how to pass from one to the other. Such discretization is generally nontrivial for non-equilateral quantum graphs. We will show however that we can construct functions on the directed edges $B_N$ which, quite miraculously, are eigenfunctions of a simple non-backtracking operator denoted below by $\zeta^{\gamma}\cB$. This reduction from continuous to discrete will use the quantum Green's functions identities derived in Section~\ref{subsecGreen}, and may be relevant to other problems on quantum graphs.

Let $\cQ$ be a quantum graph, whose set of oriented edges is denoted by $B$ (later, the following construction will be applied to $\cQ=\cQ_N$, and all the objects depend on $N$). From now on, we fix $\eta_0>0$ which will go to zero in the sequel. Recall that $\zeta^\gamma$ was defined in Section~\ref{subsec:Univ} using the universal cover of $\cQ$.

Let $\psi_j$ be an eigenfunction of $H_{\cQ}$ with eigenvalue $\lambda_j$. We define $f_j,f_j^{\ast}\in \C^{B}$ by
\begin{equation}\label{e:fj}
f_j(b) = \frac{\psi_j(t_b)}{S_{\lambda_j}(L_b)} - \frac{\zeta^{\gamma_j}(b) \psi_j(o_b)}{S_{\lambda_j}(L_b)} , \qquad f_j^{\ast}(b) = \frac{\psi_j(o_b)}{S_{\lambda_j}(L_b)} - \frac{\zeta^{\gamma_j}(\widehat{b})\psi_j(t_b)}{S_{\lambda_j}(L_b)}
\end{equation}
where $\gamma_j = \lambda_j + \ii\eta_0$ and $\zeta^{\gamma_j}(b)$ is as in \eqref{e:reczet}.

Note that $\psi_j(o_b)$ and $\psi_j(t_b)$ can be recovered from $f_j(b)$ and $f_j^*(b)$ as follows: $\psi_j(o_b) = \frac{S_{\lambda_j}(L_b)}{{1- \zeta^{\gamma_j}(b)\zeta^{\gamma_j}(\widehat{b})}} \big(f^*_j(b) + \zeta^{\gamma_j}(\widehat{b}) f_j(b)\big)$ and $\psi_j(t_b) = \frac{S_{\lambda_j}(L_b)}{1- \zeta^{\gamma_j}(b)\zeta^{\gamma_j}(\widehat{b})} \big(f_j(b) + \zeta^{\gamma_j}(b) f_j^*(b)\big)$.   These expressions are well defined, since by \eqref{e:zetainv} we
have
\[
  1 - \zeta^{\gamma}(b)\zeta^{\gamma}(\hat{b}) =   
\frac{S_{\gamma}(L_b) \zeta^{\gamma}(b)}{G^{\gamma}(t_b,t_b)} =
\frac{S_{\gamma}(L_b) \zeta^{\gamma}(\hat{b})}{G^{\gamma}(o_b,o_b)},
\]
which does not vanish using \eqref{e:1}.

Recall that $\cN_b^+$ was introduced in \eqref{eq:DefOutInBond}. We define the non-backtracking operator $\cB:\C^B\to \C^B$ by 
\[
(\cB f)(b) = \sum_{b^+\in \cN_b^+} f(b^+) \,.
\]
We observe that $\psi_j(t_b) = \psi_j(o_b) C_{\lambda_j}(L_b) + \psi_j'(o_b)S_{\lambda_j}(L_b)$, so
\begin{align*}
\sum_{b^+\in \cN_b^+} \frac{\psi_j(t_{b^+})}{S_{\lambda_j}(L_{b^+})} & = \psi_j(t_b)\sum_{b^+\in \cN_b^+}\frac{C_{\lambda_j}(L_{b^+})}{S_{\lambda_j}(L_{b^+})} + \sum_{b^+\in \cN_b^+} \psi_j'(o_{b^+}) \\
& = \psi_j(t_b)\sum_{b^+\in \cN_b^+}\frac{C_{\lambda_j}(L_{b^+})}{S_{\lambda_j}(L_{b^+})} + \psi_j'(t_b) + \alpha_{t_b}\psi_j(t_b) \\
& = \psi_j(t_b)\sum_{b^+\in \cN_b^+}\frac{C_{\lambda_j}(L_{b^+})}{S_{\lambda_j}(L_{b^+})} + \frac{S_{\lambda_j}'(L_b)\psi_j(t_b)-\psi_j(o_b)}{S_{\lambda_j}(L_b)} + \alpha_{t_b}\psi_j(t_b)\,,
\end{align*}
where we used the $\delta$-conditions and the fact that $\begin{pmatrix} \psi_{\lambda}(o_b) \\ \psi'_{\lambda}(o_b) \end{pmatrix} = M_{\lambda}(b)^{-1} \begin{pmatrix} \psi_{\lambda}(t_b) \\ \psi'_{\lambda}(t_b) \end{pmatrix}$.

We thus get
\begin{equation}\label{eq:InvNonBack}
\begin{aligned}
(\cB f_j)(b) & = \psi_j(t_b) \bigg(\sum_{b^+\in \cN_b^+} \left[\frac{C_{\lambda_j}(L_{b^+})}{S_{\lambda_j}(L_{b^+})} - \frac{\zeta^{\gamma_j}(b^+)}{S_{\lambda_j}(L_{b^+})}\right] + \frac{S_{\lambda_j}'(L_b)}{S_{\lambda_j}(L_b)} +\alpha_{t_b} \bigg)- \frac{\psi_j(o_b)}{S_{\lambda_j}(L_b)} \\
& = \frac{1}{\zeta^{\gamma_j}(b)} f_j(b) + O_{\psi_j, \eta_0}(b) \,,
\end{aligned}
\end{equation}
with $O_{\psi_j, \eta_0}(b) = \psi_j(t_b)O_{\gamma_j}(b)$ 
where, by \eqref{e:1}, we have
\begin{align*}
O_{\gamma_j}(b) = &  \sum_{b^+\in \cN_b^+} \left[ \frac{C_{\lambda_j}(L_{b^+}) - \zeta^{\gamma_j}(b^+)}{S_{\lambda_j}(L_{b^+})}- \frac{C_{\gamma_j}(L_{b^+})- \zeta^{\gamma_j}(b^+)}{S_{\gamma_j}(L_{b^+})} \right]\\
& + \frac{S_{\lambda_j}'(L_b)}{S_{\lambda_j}(L_b)} - \frac{S_{\gamma_j}'(L_b)}{S_{\gamma_j}(L_b)} +\frac{1}{\zeta^{\gamma_j}(b)S_{\gamma_j}(L_b)} - \frac{1}{\zeta^{\gamma_j}(b)S_{\lambda_j}(L_b)} .
\end{align*}

Similarly, since $f_j^{\ast} = \iota f_j$, where $\iota$ is the edge-reversal, and since $\cB^{\ast} = \iota \cB \iota$, we get $\cB^{\ast} f_j^{\ast} = \frac{1}{\iota \zeta^{\gamma_j}} f_j^{\ast} + \iota O_{\psi_j, \eta_0}(b)$,
with $\iota O_{\psi_j, \eta_0}(b) =  \psi_j(o_b)\iota O_{\gamma_j}(b)$.

Note that, by Corollary~\ref{cor:ConfinedTilMay11} and analyticity of $C_z(L_b),S_z(L_b)$,  we have for all $s>0$
\begin{equation}\label{eq:OSmall}
\frac{1}{N} \|O_{\gamma_j}\|^s_{\ell^s(B_N)}=O^{(s)}_{N\to +\infty, \gamma}(\eta_0).
\end{equation}

If $K=K^\gamma\in \mathscr{H}_k$ for some $k\geq 1$, is some operator, possibly depending on $\gamma\in \C^+$, we define its \emph{non-backtracking quantum variance} by
\begin{align}\label{e:nbtqv}
\mathrm{Var}^I_{\mathrm{nb}, \eta_0}(K) &:= \frac{1}{\mathbf{N}_N(I)} \sum_{\lambda_j\in I} \left| \langle f_j^{\ast}, K^{\lambda_j+\ii \eta_0}_B f_j \rangle \right| \\
&=  \frac{1}{\mathbf{N}_N(I)} \sum_{\lambda_j\in I} \left| \sum_{(b_1,\dots, b_k)\in \mathrm{B}_k} \overline{f_j^{\ast}(b_1)} K^{\lambda_j+\ii \eta_0} (b_1,\dots, b_k) f_j(b_k) \right| .
\end{align}

The advantage of this non-backtracking variance is that it is invariant under a spectral averaging operator, whose kernel is relatively simple due to non-backtracking.

More precisely, we have by \eqref{eq:InvNonBack} and its analog for $f_j^{\ast}$ that
\begin{align}\label{e:theerror}
\langle f_j^{\ast}, K_B f_j\rangle = \langle (\iota \zeta^{\gamma_j}\cB^{\ast})^{n-r} f_j^{\ast}, K_B (\zeta^{\gamma_j}\cB)^r f_j \rangle + \mbox{ error}
\end{align}
for any $n\geq 1$ and $0\leq r\leq n$, where the error depends on $O_{\psi_j,\eta_0}$ and will be developed in \eqref{eq:RInv2} and \eqref{eq:RInv4}. But before starting the calculation, we would like to express the ``sandwich'' $(\cB\overline{\iota \zeta^{\gamma_j}})^{n-r}K_B(\zeta^{\gamma_j}\cB)^r$ as a new observable in phase space\footnote{Very roughly, this heuristic can be likened to the classical proof of quantum ergodicity on manifolds \cite{CdV,Zel} where we replace the sandwich $\ee^{-\ii t\Delta}a\ee^{\ii t\Delta}$ by $\mathrm{Op}(a\circ G^t)$, with $G^t$ the geodesic flow, using Egorov's theorem. We see that even if $a(x)$ is a function, one must consider the phase space observable $a\circ G^t(x,\xi)$. In our case, even if $K(f_N)\in\mathscr{H}_0$ is a function, the new observable $\cR_{n,r}^{\gamma}K$ will lie in $\mathscr{H}_n$.}. 

First note that 
\begin{align*}[(\zeta^{\gamma}\cB)^r f](b_1) = \sum_{(b_2,\dots,b_{r+1})\in B^{b_1}_r} \zeta^{\gamma}(b_1)\cdots \zeta^{\gamma}(b_r) f(b_{r+1})\\
[(\iota \zeta^{\gamma}\cB^{\ast} )^k f](b_1) = \sum_{(b_{-k+1},\dots,b_0)\in B_{r,b_1}} \zeta^{\gamma}(\hat{b}_1) \cdots \zeta^{\gamma}(\hat{b}_{-k+2})f(b_{-k+1}).
\end{align*}
Recalling \eqref{eq:DefKB} this yields:
\begin{prp}If we define $\cR_{n,r}^{\gamma_j}: \mathscr{H}_k\longrightarrow \mathscr{H}_{n+k}$ by
\begin{multline}\label{e:rnr}
(\cR_{n,r}^{\gamma_j}K)(b_1 ; b_{n+k}) \\
= \overline{\zeta^{\gamma_j}(\hat{b}_2)\cdots \zeta^{\gamma_j}(\hat{b}_{n-r+1})} K(b_{n-r+1}; b_{n-r+k})\zeta^{\gamma_j}(b_{n-r+k})\cdots \zeta^{\gamma_j}(b_{n+k-1}) \,,
\end{multline}
we have
\[
\langle (\iota \zeta^{\gamma_j}\cB^{\ast})^{n-r} f_j^{\ast}, K_B (\zeta^{\gamma_j}\cB)^r f_j \rangle = \langle f_j^{\ast}, (\cR_{n,r}^{\gamma_j}K)_Bf_j\rangle .
\]
\end{prp}
Thus, $\cR_{n,r}^{\gamma_j}K$ is the ``observable'' we seek.

We now derive the expression of the error in \eqref{e:theerror} more precisely. We have
\begin{equation}\label{eq:RInv1}
\begin{aligned}
&\sum_{(b_1 ; b_{n+k})\in \mathrm{B}_{n+k}} \overline{f_j^{\ast}(b_1)} (\cR_{n,r}^{\gamma_j}K)(b_1 ; b_{n+k}) f_j(b_{n+k}) \\
&=\sum_{(b_1 ; b_{n+k-1})\in \mathrm{B}_{n+k-1}}\Big{[} \overline{f_j^*(b_1) \zeta^{\gamma_j}(\hat{b}_2)\cdots \zeta^{\gamma_j}(\hat{b}_{n-r+1})} K(b_{n-r+1} ; b_{n-r+k})\\
& \qquad\qquad\qquad\qquad\qquad\qquad \times \zeta^{\gamma_j}(b_{n-r+k})\cdots \zeta^{\gamma_j}(b_{n+k-2}) \cdot [\zeta^{\gamma_j}\cB f_j](b_{n+k-1})\Big{]}\\
&= \sum_{(b_1 ; b_{n+k-1})\in \mathrm{B}_{n+k-1}} \overline{f_j^{\ast}(b_1)} (\cR_{n-1,r-1}^{\gamma_j}K)(b_1 ; b_{n-1+k})  [f_j+\zeta^{\gamma_j} O_{\psi_j,\eta_0}](b_{n+k-1}) \\
&= \langle f_j^{\ast},(\cR_{n-1,r-1}^{\gamma_j}K)_B [f_j+\zeta^{\gamma_j}O_{\psi_j,\eta_0}]\rangle.
\end{aligned}
\end{equation}
where we used \eqref{eq:InvNonBack}. Iterating this equation $r$ times, we obtain
\begin{equation}\label{eq:RInv2}
\langle f_j^{\ast},(\cR_{n,r}^{\gamma_j}K)_Bf_j\rangle = \langle f_j^{\ast},(\cR_{n-r,0}^{\gamma_j}K)_Bf_j\rangle + \sum_{\ell=1}^r\langle f_j^{\ast},(\cR_{n-\ell,r-\ell}^{\gamma_j}K)_B\zeta^{\gamma_j}O_{\psi_j,\eta_0}\rangle .
\end{equation}

By a computation similar to \eqref{eq:RInv1}, we obtain that for all $m>1$,
\begin{equation}\label{eq:RInv3}
\begin{aligned}
&\sum_{(b_1 ; b_{m+k})\in \mathrm{B}_{m+k}} \overline{f_j^{\ast}(b_1)} (\cR_{m,s}^{\gamma_j}K)(b_1 ; b_{m+k}) f_j(b_{m+k}) \\
&= \sum_{(b_2,\ldots,b_{m+k})\in \mathrm{B}_{m+k-1}} \big{(}\overline{f_j^{\ast}(b_2) + \iota \zeta^{\gamma_j}\iota O_{\psi_j, \eta_0}(b_2)}\big{)} (\cR_{m-1,s}^{\gamma_j}K)(b_2 ; b_{m+k}) f_j(b_{m+k})\\
& = \langle f_j^{\ast}+\iota\zeta^{\gamma_j}\iota O_{\psi_j,\eta_0},(\cR_{m-1,s}^{\gamma_j}K)_B f_j\rangle.
\end{aligned}
\end{equation}

Applying \eqref{eq:RInv3} $n-r$ times to the first term in the left-hand side of \eqref{eq:RInv2}, we obtain
\begin{equation}\label{eq:RInv4}
\langle f_j^{\ast},(\cR_{n-r,0}^{\gamma_j} K)_Bf_j\rangle  = \langle f_j^{\ast}, K_Bf_j\rangle  + \sum_{\ell'=1}^{n-r} \langle \iota\zeta^{\gamma_j}\iota O_{\psi_j,\eta_0},(\cR^{\gamma_j}_{n-r-\ell',0}K)_Bf_j\rangle.
\end{equation}

As \eqref{eq:RInv2} and \eqref{eq:RInv4} hold for any $r\le n$, we thus get the desired (approximate) invariance of the quantum variance under the averaging operator $\cR_{n,r}^{\gamma}$:

\begin{prp}\label{p:VarInv}
\begin{equation}\label{eq:VarInv}
\mathrm{Var}^I_{\mathrm{nb}, \eta_0}(K) \le \mathrm{Var}^I_{\mathrm{nb}, \eta_0}\bigg(\frac{1}{n}\sum_{r=1}^n\cR_{n,r}^{\gamma}K\bigg) + \cE_{n ,\eta_0}(K),
\end{equation}
where
\begin{multline}\label{eq:RemainderVarInv}
\cE_{n ,\eta_0}(K) = \frac{1}{n} \sum_{r=1}^n \frac{1}{\mathbf{N}_N(I)}\sum_{\lambda_j\in I}\bigg| \sum_{\ell=1}^r\langle f_j^{\ast},(\cR_{n-\ell,r-\ell}^{\gamma_j}K)_B\zeta^{\gamma_j}O_{\psi_j,\eta_0}\rangle\\
+ \sum_{\ell'=1}^{n-r} \langle \iota\zeta^{\gamma_j}\iota O_{\psi_j,\eta_0},(\cR^{\gamma_j}_{n-r-\ell',0}K)_Bf_j\rangle \bigg|\,.
\end{multline}
\end{prp}
The quantity $\cE_{n ,\eta_0}(K)$ is negligible thanks to the following lemma proven in Section~\ref{subsec:ApplicHolom}.
\begin{lem}\label{lem:SmallO}
For any $n\in \N$, $k\geq 0$ and any $K\in \mathscr{H}_k$ possibly depending on $\gamma$, but satisfying \eqref{eq:APrioriBoundOperators}, we have
\[
\cE_{n ,\eta_0}(K) =O_{N\to +\infty, \eta_0}^{(n)}(\eta_0).
\]
\end{lem}

\section{Reduction to non-backtracking variances}\label{Sec:Reduction}
From now on, all the constructions take place on the graph $\cQ_N$, and $N$ varies and goes to $+\infty$. So it should be understood that all objects depend on $N$, although it is not always apparent in the notation.

In this section, we will show that the statement of Theorem \ref{thm:integralver} can be reduced to an estimate on the non-backtracking variances \eqref{e:nbtqv}. To this end, we will start from an integral operator $\mathcal{K}\in \mathscr{K}_k$, and build several new operators $J_{\mathcal{K}}^\gamma\in \mathscr{H}_{k'}$ from it. These $J_{\cK}^\gamma$ will depend on the parameter $\gamma\in \C^+$ (and $N$), but this dependence will always satisfy the following hypothesis.

\begin{defa}\label{def:Hol}
Let $k\ge 0$ and let $\C^+\ni \gamma\mapsto K^\gamma = K_N^\gamma\in \mathscr{H}_k$. We say that $K^\gamma$ satisfies hypothesis %
\raisebox{1.2\ht\strutbox}{\hypertarget{lab:hol}{}}\textbf{(Hol)} if 
\begin{itemize}
\item For all $0<\eta_0< \eta_{\mathrm{Dir}}$ and $(b_1,\dots,b_k)\in \mathrm{B}_k$, the map $\R\ni \lambda \mapsto K^{\lambda +\ii\eta_0}(b_1,\dots,b_k)$ has an analytic extension $K_{\eta_0}^z(b_1,\dots,b_k)$ to the strip $\{z : |\Im z|<\eta_0/2\}$.
\item For all $s>0$, we have
\begin{equation}\label{eq:APrioriBoundOperators}
\sup\limits_{\eta_0\in (0,\eta_{\mathrm{Dir}})} \limsup\limits_{N\to +\infty}    \sup\limits_{\eta_1\in (-\frac{\eta_0}{2},\frac{\eta_0}{2})}\sup\limits_{\lambda\in I_1} \frac{1}{N}\sum_{(b_1,\dots,b_k)\in \mathrm{B}_k} | K_{\eta_0}^{\lambda+\ii\eta_1}(b_0,\dots,b_k)|^s < +\infty.
\end{equation}
\item For almost all $\lambda\in I_1$, for all $s>0$ and all $t\in (-1/2,1/2)$, we have 
\begin{equation}\label{eq:APrioriBoundOperators3}
 \lim\limits_{\eta_0\downarrow 0} \limsup\limits_{N\to +\infty}  \frac{1}{N} \sum_{(b_1,\dots,b_k)\in \mathrm{B}_k} | K_{\eta_0}^{\lambda+t\ii\eta_0}(b_0,\dots,b_k) - K^{\lambda+\ii\eta_0}(b_0,\dots,b_k) |^s =0.
\end{equation}
\end{itemize}
\end{defa}

Beware that, unless $\gamma \mapsto K^{\gamma}(b_1,\dots,b_k)$ is analytic, the quantities $K^{\lambda + \ii \eta_0 +\ii \eta_1}(b_1,\dots,b_k)$ and $K_{\eta_0}^{\lambda +\ii \eta_1}(b_1,\dots,b_k)$ are in general different for $\eta_1\neq 0$.

\begin{rem}\label{rem:HolStable}
Note that if $K^\gamma$ and $J^\gamma$ satisfy \Hol{}, so does their sum and product. If $K^\gamma$ satisfies \Hol{}, then $\overline{K^\gamma}$ also satisfies \Hol{}. Indeed, its restriction to a horizontal line is real-analytic, so it can be extended holomorphically to a strip by an extension satisfying the desired properties.

Last but not least, thanks to Corollary \ref{cor:ConfinedTilMay11} and Proposition \ref{prop:GreenisCool3}, the functions in the class $\mathscr{L}_k^\gamma$ introduced in Definition \ref{def:GeneralContinuousOp} all satisfy \Hol{}.
\end{rem}

Given a linear operator $F^\gamma : \ell^2(V)\longrightarrow \ell^2(V)$ possibly depending on $\gamma\in \C^+$, we define 
\[
\varie(F^\gamma):=  \frac{1}{\mathbf{N}_N(I)}\sum_{\lambda_j\in I} |\langle \mathring{\psi}_j, F^{\lambda_j+\ii \eta_0} \mathring{\psi}_j \rangle_{\ell^2(V)}| \, ,
\]
where $\mathring{\psi}_j$ is the restriction of $\psi_j$ to the vertices, which  is well-defined thanks to \eqref{eq:Continuity}
\[
\mathring{\psi}_j(v) = \psi_j(v).
\]

Let us write $\Psi_{\gamma,v}(w) := \Im \tilg_N^{\gamma}(v,w)$ and define, for  $K\in \mathscr{H}_k$ (possibly depending on $\gamma\in \C^+$):
\begin{equation}\label{eq:DefBracket2}
\begin{aligned}
\langle K \rangle^{\gamma} &:= \frac{1}{\sum_{v\in V}\Psi_{\gamma,v}(v)} \sum_{(b_1,\ldots, b_k)\in \mathrm{B}_k} K(b_1,\ldots, b_k)\Psi_{\gamma,o_{b_1}}(t_{b_k})  &\text{ if } k\geq 1 \\
\langle K \rangle^{\gamma} &:= \frac{1}{\sum_{v\in V}\Psi_{\gamma,v}(v)} \sum_{v\in V} K(v)\Psi_{\gamma,v}(v)  &\text{ if } k=0.
\end{aligned}
\end{equation}

Finally, we denote by $\mathbf{1} = \mathbf{1}_0 \in \mathscr{H}_0$ the constant function equal to one on every vertex.
In the sequel, we will also denote by $\hat{\mathbf{1}}$ the function on $\mathcal{G}_N$ which is constant equal to one.

As a first step towards the reduction to non-backtracking variances, we control the quantities appearing in Theorems~\ref{thm:qeqg} and \ref{thm:integralver} by some discrete variances of the form $\varie$.

\begin{prp}\phantomsection \label{prop:Reduc1}
\begin{enumerate}[\rm (1)]
\item Let $f_N\in L^\infty(\mathcal{G}_N)$ be a sequence of functions with $\|f_N\|_{\infty}  \le 1$. We may build $(J^\gamma_{f,p})_{p=1,\dots,8}$, each belonging to $\mathscr{H}_{0}$ or $\mathscr{H}_{1}$  and satisfying \emph{\Hol{}}, such that
\[
\lim_{\eta_0\downarrow 0}\limsup_{N\to\infty} \frac{1}{\mathbf{N}_N(I)}\sum_{\lambda_j\in I}|\langle \psi_j, f \psi_j\rangle-\langle  f \rangle_{\gamma_j}|
\le C \sum_{p=1}^8\lim_{\eta_0\downarrow 0}\limsup_{N\to\infty} \varie\left( (J_{f,p}^\gamma)_G -\langle J_{f,p}^\gamma\rangle^{\gamma} \mathbf{1}\right).
\]
\item Let $k\geq 1$ and let $\mathcal{K}_N\in \mathscr{K}_k$ be a sequence of non-backtracking integral operators satisfying $|\mathcal{K}_N(x,y)|\leq 1$ for all $x,y\in \mathcal{G}_N$.

We may build $(J^\gamma_{\mathcal{K},p})_{p=1,\ldots,6}$, each belonging to $\mathscr{H}_{k_p}$ for $k-2\le k_p\le k$, and satisfying \emph{\Hol{}}, such that
\[
\lim_{\eta_0\downarrow 0}\limsup_{N\to\infty} \frac{1}{\mathbf{N}_N(I)}\sum_{\lambda_j\in I}|\langle \psi_j, \mathcal{K} \psi_j\rangle-\langle  \mathcal{K} \rangle_{\gamma_j}|
\le C \sum_{p=1}^{6}\lim_{\eta_0\downarrow 0}\limsup_{N\to\infty} \varie\left( (J_{\mathcal{K},p}^\gamma)_G -\langle J_{\mathcal{K},p}^\gamma\rangle^{\gamma} \mathbf{1}\right).
\]
\end{enumerate}
\end{prp}

The idea of the proof is quite simple: expand the eigenfunction $\psi_j$ in the basis of solutions $S_{\lambda},C_{\lambda}$ given in \S~\ref{e:sols}, then split the average $\langle f\rangle_{\gamma_j}$ into convenient pieces, so that each operator in the discrete variance has zero mean. The scheme is quite similar to the equilateral case \cite{QEQGEQ}. We give the technical details in Appendix~\ref{sec:disred} 

\medskip

Our next step is to show that the discrete variance of $J-\langle J \rangle^{\gamma}\mathbf{1}$ can be bounded by some non-backtracking variances, plus some terms which will be negligible. Before stating the proposition, we need to introduce some notation. In the sequel, we will often write $\langle J \rangle^{\gamma}$ instead of $\langle J \rangle^{\gamma}\mathbf{1}$ in the computations of the variances.

If $d(x)$ is the degree of $x\in V$, we denote $P = \frac{1}{d} \cA_G$. Let
\[
N_{\gamma}(x) = \Im \tilg^{\gamma}(\tilde{x},\tilde{x}) \, ,\qquad P_{\gamma} = \frac{d}{N_{\gamma}} P \frac{N_{\gamma}}{d} \, .
\]
We define $\mathcal{S}_{T,\gamma}:\ell^2(V)\to \ell^2(V)$ and $\widetilde{\mathcal{S}}_{T,\gamma}:\ell^2(V)\to\ell^2(V)$ by
\[
\mathcal{S}_{T,\gamma}J = \frac{1}{T} \sum_{s=0}^{T-1} (T-s) P_{\gamma}^s J \quad \text{and} \quad \widetilde{\mathcal{S}}_{T,\gamma}J = \frac{1}{T} \sum_{s=1}^T P_{\gamma}^s J \, .
\]

We also define $\mathcal{L}^{\gamma} : \ell^2(V) \to \ell^2(B)$ and $\cE_{\gamma}:\ell^2(V)\to\ell^2(V)$ by
\[
(\cL^\gamma J)(b) = \frac{S_{\Re \gamma}^2(L_b)|\tilg^{\gamma}(t_b,t_b)|^2(1+\overline{\zeta^\gamma(b)}\zeta^\gamma(\widehat{b}))}{\Re S_{\gamma}(L_b)|\zeta^\gamma(b)|^2} \bigg(\frac{J(t_b)}{N_{\gamma}(o_b)} - \frac{J(o_b)}{N_{\gamma}(t_b)}\bigg),
\]
\[
(\cE_{\gamma}J)(o_b) = \sum_{t_b\sim o_b }\left[\frac{\Im S_{\gamma}(L_b)\Re \tilg^{\gamma}(t_b,t_b)}{\Re S_{\gamma}(L_b)}\Big(\frac{J(t_b)}{N_{\gamma}(o_b)}-\frac{J(o_b)}{N_{\gamma}(t_b)}\Big)\right].
\]

\begin{prp}\phantomsection \label{prop:Reduc2}
\begin{enumerate}[\rm (1)]
\item Let $J\in \mathscr{H}_0$. For any $T\in \N$, we have
\begin{equation}\label{eq:VarToNB}
\begin{aligned}
\mathrm{Var}^I_{\eta_0}(J-\langle J \rangle^{\gamma} )\le& \mathrm{Var}^I_{\mathrm{nb},\eta_0}[\mathcal{L}^{\gamma}d^{-1}\mathcal{S}_{T,\gamma}(J-\langle J\rangle^{\gamma})] \\
&	+\mathrm{Var}^I_{\eta_0} [\cE_{\gamma}d^{-1}\cS_{T,\gamma}(J-\langle J\rangle^{\gamma})] + \mathrm{Var}^I_{\eta_0}(\widetilde{\mathcal{S}}_{T,\gamma}(J - \langle J \rangle^{\gamma})) \, .
\end{aligned}
\end{equation}
\item Let $J\in \mathscr{H}_k$ satisfy \emph{\Hol{}}. There exists $P_1, P_2\in \N$ depending only on $k$, and operators $(L_p^\gamma J)_{1\leq p\leq P_1}$ and $(F_p^\gamma J)_{1\leq p\leq P_2}$ with $L_p^\gamma J\in \mathscr{H}_0$ and $F_p^\gamma J\in \mathscr{H}_{k_p}$ for some $1\leq k_p\leq k$, all satisfying \emph{\Hol{}}, such that
\[
\mathrm{Var}^I_{\eta_0}\left(J_G - \langle J \rangle^\gamma \right)\leq \sum_{p=1}^{P_1} \mathrm{Var}^I_{\eta_0}(L_p^\gamma J-\langle L_p^\gamma J \rangle^{\gamma}) + \sum_{p=1}^{P_2}  \mathrm{Var}^I_{\mathrm{nb},\eta_0}\left( F_p^\gamma J\right) + O_{N\to +\infty, \gamma} (\eta_0).
\]
\end{enumerate}
\end{prp}

This proposition will be applied to the $J=J_{f,p}^\gamma,J_{\cK,p}^\gamma$ appearing in Proposition~\ref{prop:Reduc1}.

The proof is given in Appendix~\ref{sec:nonbared}.

\medskip

To finish this section, we will show that the error variances $\mathrm{Var}^I_{\eta_0} [\cE_{\gamma} d^{-1} \cS_{T,\gamma}(J-\langle J\rangle^{\gamma})]$ and $\mathrm{Var}^I_{\eta_0}(\widetilde{\mathcal{S}}_{T,\gamma}(J - \langle J \rangle^{\gamma}))$ from \eqref{eq:VarToNB} vanish asymptotically. Note that if we apply part (1) of Proposition~\ref{prop:Reduc2} to the first sum of variances in part (2), then all that remains will be to control non-backtracking variances.

If $K\in \C^V\cong \mathscr{H}_0$, we define the ``weighted Hilbert-Schmidt norm''
\begin{equation}\label{eq:NormeGamma0}
\|K\|^2_{\gamma}:= \frac{1}{|V|} \sum_{v\in V} |K(v)|^2 N_\gamma(x)^2.
\end{equation}

Then we have the following bound proven in Section~\ref{subsec:ApplicHolom}. An analogous bound for non-backtracking variances will be given in Proposition~\ref{thm:upvar}. These two propositions tell us the quantum variance(s) are dominated by weighted Hilbert-Schmidt norms.

\begin{lem}\label{lem:EasyVarBound}
Let $K^\gamma\in \mathscr{H}_0$ satisfy hypothesis \emph{\Hol{}} from Definition \ref{def:Hol}.
We have
\[
\lim_{\eta_0 \downarrow 0}\limsup_{N\to\infty} \mathrm{Var}^I_{\eta_0}(K^{\gamma})^2 \lesssim \lim_{\eta_0\downarrow 0}\limsup_{N\to\infty} \int_{I} \|K^{\lambda+\ii\eta_0}\|_{\lambda+\ii\eta_0}^2\,\dd\lambda.
\]
\end{lem}

This lemma makes it easy to deal with $\mathrm{Var}^I_{\eta_0} [\cE_{\gamma}d^{-1}\cS_{T,\gamma}(J-\langle J\rangle^{\gamma})]$. Indeed, for any fixed $T\in \N$, $\cE_\gamma d^{-1}S_{T,\gamma} (K^\gamma - \langle K^\gamma \rangle^\gamma)\in \mathscr{H}_0$ satisfies \Hol{} by Remark~\ref{rem:HolStable}. Therefore, from the expression of $\cE_\gamma$, which involves $\Im S_\gamma(L_b) = O(\eta_0)$, we see using \eqref{eq:APrioriBoundOperators} that
\begin{equation}\label{eq:BoundEVar}
\lim_{\eta_0\downarrow 0}\lim_{N\to \infty} \mathrm{Var}^I_{\eta_0} \left( \cE_{\gamma}d^{-1}S_{T,\gamma}(K^{\gamma}-\langle K^{\gamma}\rangle^{\gamma}\right) = 0.
\end{equation}

The error $\mathrm{Var}^I_{\eta_0}(\widetilde{\mathcal{S}}_{T,\gamma}(J - \langle J \rangle^{\gamma}))$ is dealt with thanks to the following lemma.
\begin{lem}\label{lem:EXPisCool}
Let $K\in \mathscr{H}_0$ satisfy \emph{\Hol{}}. We then have
\[
\lim_{T\to\infty}\sup_{\eta_0\in (0,\eta_{\mathrm{Dir}})}\limsup_{N\to\infty}\sup_{\lambda\in I}\big\|\widetilde{S}_{T,\gamma} \left( K^{\gamma}-\langle K^{\gamma}\rangle^{\gamma}\right)\big\|_{\gamma}=0.
\]
\end{lem}
\begin{proof}
Let us write $(Y_\gamma K)(v)= \frac{d(v)}{N_\gamma(v)} \frac{\sum_{w\in V} N_\gamma(w) K(w)}{\sum_{w\in V}d(w)}$, so that $Y_\gamma K = \frac{ \left\langle N_\gamma K\right\rangle_U}{ \left\langle d\right\rangle_U} \frac{d}{N_\gamma}$, where $\langle J \rangle_U := \frac{1}{N}\sum_{v\in V} J(v)$ is the average with respect to the uniform scalar product.

Noting that for any $s\in \N$, we have  $P^s_{\gamma} = \frac{d}{N_{\gamma}} P^s \frac{N_{\gamma}}{d}$, we get

\begin{align*}
\big{\|}\widetilde{S}_{T,\gamma} K^{\gamma}-Y_\gamma K^\gamma \big{\|}^2_{\gamma} &= \frac{1}{N}\sum_{v\in V} \Big{|} \frac{1}{T}\sum_{s=1}^T d(v) \Big{(}P^s \frac{N_\gamma K^\gamma}{d}\Big{)}(v) -  \frac{ \left\langle N_\gamma K^{\gamma}\right\rangle_U}{ \left\langle d\right\rangle_U} d(v) \Big{|}^2\\
&\leq \frac{D}{NT} \sum_{s=1}^T \left\| P^s \Big{(}\frac{N_\gamma K^\gamma}{d} - \frac{ \left\langle N_\gamma K^{\gamma}\right\rangle_U}{ \left\langle d\right\rangle_U} \mathbf{1} \Big{)} \right\|_{\ell^2(V,d)}^2\\
& \leq \frac{D}{NT} \sum_{s=1}^T (1-\beta)^{2s} \left\| \frac{N_\gamma K^\gamma}{d} - \frac{ \left\langle N_\gamma K^{\gamma}\right\rangle_U}{ \left\langle d\right\rangle_U} \mathbf{1}  \right\|_{\ell^2(V,d)}^2\\
&\leq \frac{D}{\beta NT} \left\| \frac{N_\gamma K^\gamma}{d} - \frac{ \left\langle N_\gamma K^{\gamma}\right\rangle_U}{ \left\langle d\right\rangle_U} \mathbf{1}  \right\|_{\ell^2(V,d)}^2\leq \frac{4D}{\beta NT} \left\| N_\gamma K^\gamma\right\|_{\ell^2(V,d)}^2,
\end{align*}
where we used the fact that $\frac{N_\gamma K^\gamma}{d} - \frac{ \left\langle N_\gamma K^{\gamma}\right\rangle_U}{ \left\langle d\right\rangle_U} \mathbf{1}$ is orthogonal to constants in the space $\ell^2(V_N, d_N)$, the assumption \EXP{}, and the fact that $d\geq 1$.

Now, since $K^\gamma$ satisfies \Hol{}, we know by Remark \ref{rem:HolStable} that $N_\gamma K^\gamma$ satisfies hypothesis \Hol{}, so that $\frac{1}{N}\left\| N_\gamma K^\gamma\right\|_{\ell^2(V)}^2 = O_{N\to +\infty, \gamma}(1)$. Therefore,
\begin{equation}\label{eq:GoingToConstant}
\lim_{T\to\infty}\sup_{\eta_0\in (0,\eta_{\mathrm{Dir}})}\limsup_{N\to\infty}\sup_{\lambda\in I}\big\|\widetilde{S}_{T,\gamma} K^{\gamma}-Y_\gamma K^\gamma \big\|_{\gamma}  = 0 .
\end{equation}

Now, the result follows by applying \eqref{eq:GoingToConstant} to $\mathring{K}^\gamma :=K^\gamma - \langle K^\gamma \rangle^\gamma \mathbf{1}$, and by noting that $Y_\gamma \mathring{K}^\gamma=0$, since $Y_\gamma  \langle K^\gamma \rangle^\gamma \mathbf{1} =  \langle K^\gamma \rangle^\gamma \frac{\langle N_\gamma \rangle_U}{\langle d \rangle_U} \frac{d}{N_\gamma}$ and, by definition, $\langle K^\gamma \rangle^\gamma= \frac{\langle N_\gamma K^\gamma \rangle_U}{\langle N_\gamma \rangle_U}$.
\end{proof}

Combining Propositions~\ref{prop:Reduc1} and \ref{prop:Reduc2}, Lemma \ref{lem:EasyVarBound}, Lemma \ref{lem:EXPisCool} and equation \eqref{eq:BoundEVar}  we obtain the following corollary.

\begin{cor}\label{cor:Reduction2}
Let $\cQ_N$ be a sequence of quantum graphs satisfying \emph{\Data{}} for each $N$, such that \emph{\EXP{}}, \emph{\BST{}}, \emph{\Green{}} and \emph{\NonDirichlet{}} hold true on the interval $I_1$. Suppose that for any $K^\gamma\in \mathscr{H}_k$ satisfying \emph{\Hol{}}, $k\ge 0$, we have 
\begin{equation}\label{eq:NBVarToShow}
\lim_{\eta_0 \downarrow 0}\lim_{N\to\infty}\mathrm{Var}^I_{\mathrm{nb},\eta_0}(K^\gamma)=0.
\end{equation}

Then, if $f_N\in L^\infty(\mathcal{G}_N)$ is a sequence of functions satisfying $\|f_N\|_{\infty}  \le 1$, \eqref{e:main'} holds. Furthermore, for any $k\geq 0$, if $\mathcal{K}_N\in \mathscr{K}_k$ is a sequence non-backtracking integral operator satisfying $|\mathcal{K}_N(x,y)|\leq 1$ for all $x,y\in \mathcal{G}_N$, then \eqref{eq:MainStatement} holds.
\end{cor}

The rest of the paper will be devoted to proving \eqref{eq:NBVarToShow}, thus establishing Theorem \ref{thm:integralver}.

\section{Contour integrals and complex analysis}\label{sec:Complex}
In this section, we shall prove Lemmas~\ref{lem:SmallO} and \ref{lem:EasyVarBound} and develop tools from complex analysis that will be used later on. The quantities we wish to estimate are expressed as sums over the eigenvalues of the quantum graph $\cQ_N$, which are the
poles of the Green function $g_N$. Thanks to Cauchy's formula, these sums will be expressed as contour integrals involving the Green functions. The manipulation of the unknown
eigenfunctions $\psi_j^{(N)}$ is thus replaced by manipulation of $g_N$. Later on, this will be replaced by the Green function $\tilde g_N$ of the universal cover, and finally
we will use that it converges in distribution to the Green function of the limiting random tree $\cT$.

If $\chi \in C_c^\infty(\R)$, we will denote by $\tilde{\chi}$  an \emph{almost analytic} extension of $\chi$, i.e., a smooth function $\tilde{\chi} : \C \mapsto \C$ such that $\tilde{\chi}(z) = \chi(z)$ for $z\in \R$, $\frac{\partial \tilde{\chi}}{\partial \overline{z}}(z) = O ((\Im z)^2)$, and 
\begin{equation}\label{eq:support extension}
\mathrm{supp }~ \tilde{\chi}\subset \{z; \Re z \in \mathrm{supp}~\chi\}.
\end{equation}
Here $\frac{\partial}{\partial\overline{z}}=\frac{1}{2}(\frac{\partial}{\partial x}+\ii\frac{\partial}{\partial y})$.
For instance, one can take $\tilde{\chi}(x+\ii y) = \chi(x) + \ii y \chi'(x) - \frac{y^2}{2} \chi''(x)$.
We refer the reader to \cite[\S~2.2]{Davies} for more details about almost analytic extensions.

Recall that if $K^\gamma\in \mathscr{H}_k$ satisfies \Hol{}, then for any $0<\eta_0<\eta_{\mathrm{Dir}}$ and $(b_1,\dots,b_k)\in \mathrm{B}_k$, $\lambda\mapsto K^{\lambda+\ii \eta_0} (b_1,\dots,b_k)$ admits a holomorphic extension to $\{|\Im z|<\eta_0/2\}$, which is denoted by $K^z_{\eta_0}(b_1,\dots,b_k)$.

\begin{prp}\label{Cor:SumToInt}
Let $\chi \in C_c^\infty(I_1)$. Define $\tilde{\chi}$ as above. Let $k,k'\in \N$ and let $K^\gamma\in \mathscr{H}_k$, and  $K'^\gamma\in \mathscr{H}_{k'}$ satisfy \emph{\Hol{}}. Then for any $0<\eta_0<\min(\frac{\eta_{\mathrm{Dir}}}{2},\eta_{I_1})$,  we have
\begin{multline}
\frac{1}{N}\sum_{j\geq 1} \sum_{(b_1; b_k)} \sum_{\substack{(b'_1; b'_{k'})\\ b'_1=b_1}} \chi(\lambda_j) K^{\gamma_j}(b_1; b_k) K'^{\gamma_j} (b'_1; b'_{k'}) \psi_j(o_{b_k}) \overline{\psi_j(o_{b_{k'}'})} \\
=-\frac{1}{N}\sum_{(b_1; b_k)} \sum_{\substack{(b'_1; b'_{k'})\\ b'_1=b_1}}\frac{1}{2\pi\ii} \int_{\Gamma_{\eta_0/4}} \tilde{\chi}(z) g_N^z(o_{b_k},o_{b'_{k'}}) K^{z}_{\eta_0}(b_1; b_k) K'^{z}_{\eta_0} (b'_1; b'_{k'})  \dd z
+ O_{N\to +\infty, \eta_0} (\eta_0),
\end{multline}
where $\Gamma_{\eta}$ is the boundary of $\Omega_{\eta} = I_1+\ii[-\eta,\eta]$.

The same formula holds if $o_{b_k}$ and/or $o_{b'_{k'}}$ are replaced by $t_{b_k}$ and/or $t_{b'_{k'}}$ in both terms.
\end{prp}

\begin{proof}
\textbf{Step 1 : From a sum to an integral}

Let $v,w\in V_N$, and let $h$ be a holomorphic function in a strip $\{z\in \C ; |\Im z |<\frac{\eta_{\mathrm{Dir}}}{2}\}$. The function $h$ may depend on $N$, $v,w$.

Using \eqref{e:Maxime}, and noting that $Z_{\lambda}^b=\overline{Z_\lambda^b}$ for $\lambda\in \R$, we have
\[
\sum_{j\ge 1} \chi(\lambda_j) h(\lambda_j) \psi_j(v) \overline{\psi_j(w)}=\sum_{j\ge 1} \chi(\lambda_j) h(\lambda_j)\langle Z_{\lambda_j}^{b_v}, \psi_j\rangle \langle Z_{\lambda_j}^{b_w}, \overline{\psi_j} \rangle,
\]
where we chose $b_v\neq b_w$ such that $v=t_{b_v}$ and $w=t_{b_w}$. This is possible since $d(v),d(w)\ge 2$ by assumption \Data{}.

We may define holomorphic functions $\widehat{Z}^b_z$ as in the end of Section~\ref{sec:ScalarProd}, on some open set $\Omega_{I_1}:=  I_1 +\ii [-\eta_{I_1},\eta_{I_1}]$ which does not depend on $N\in \N$ or $b\in B_N$.

Let $0<\eta_1< \min(\eta_{I_1},\frac{\eta_{\mathrm{Dir}}}{2})$.  Consider the rectangle $\Omega_{\eta_1} := I_1 + \ii [-\eta_1, \eta_1]$ and the curve $\Gamma_{\eta_1} := \partial \Omega_{\eta_1}$. Cauchy's integral formula $f(\lambda) = \frac{1}{2\pi \ii} \int_{\Gamma} \frac{f(z)}{z-\lambda}\,\dd z$ for analytic $f$ generalizes to
\[
f(\lambda) = \frac{1}{2\pi \ii} \bigg(\int_{\Gamma_{\eta_1}} \frac{f(z)}{z-\lambda}\,\dd z + \int_{\Omega_{\eta_1}} \frac{\partial f/ \partial \overline{z}}{z-\lambda}\,\dd z\wedge \dd \overline{z}\bigg)
\]
for $\lambda \in \Omega_{\eta_1}$, where $\dd z\wedge \dd \overline{z}=-2\ii\,\dd x\dd y$, see e.g. \cite[Theorem 1.2.1]{Hor}.

We apply this result to $\lambda = \lambda_j$ and
\[
f(z)=\tilde{\chi}(z) h(z) \langle \widehat{Z}_z^{b_v}, \psi_j\rangle \langle \widehat{Z}_z^{b_w},\overline{\psi_j} \rangle.
\]

Noting that $z\in \Omega_I \mapsto  h(z) \langle \widehat{Z}_z^{b_v}, \psi_j\rangle \langle \widehat{Z}_z^{b_w},\overline{\psi_j} \rangle$ is holomorphic, we obtain
\begin{align*}
 \chi(\lambda_j) h(\lambda_j) \psi_j(v) \overline{\psi_j(w)} &= \frac{1}{2\pi\ii}\bigg(\int_{\Gamma_{\eta_1}} \frac{\tilde{\chi}(z) h(z) \langle \widehat{Z}_z^{b_v}, \psi_j\rangle  \langle \widehat{Z}_z^{b_w},\overline{\psi_j} \rangle}{z-\lambda_j}\,\dd z \\
 &+ \int_{\Omega_{\eta_1}} \frac{h(z)  \langle \widehat{Z}_z^{b_v}, \psi_j\rangle  \langle \widehat{Z}_z^{b_w},\overline{\psi_j} \rangle}{z-\lambda_j}  \frac{\partial \tilde{\chi}}{\partial \overline{z}}\,\dd z\wedge \dd \overline{z}\bigg).
 \end{align*}

Next, we want to sum this expression over $j$. Since $(\psi_j)$ is an orthonormal basis of $L^2(\mathcal{G}_N)$, we obtain that
\begin{equation*}
 \sum_{j\geq 1} \frac{\langle \widehat{Z}_z^{b_v}, \psi_j\rangle  \langle \widehat{Z}_z^{b_w},\overline{\psi_j} \rangle}{\lambda_j-z}=   \sum_{j\geq 1} \frac{\langle \widehat{Z}_z^{b_v}, \psi_j\rangle \langle \psi_j, \overline{\widehat{Z}_z^{b_w}} \rangle}{\lambda_j-z} = \langle \widehat{Z}_z^{b_v}, (H_{\cQ_N}-z)^{-1} \overline{\widehat{Z}_z^{b_w}} \rangle\,.
\end{equation*}

Therefore, we have
\begin{align*}
\sum_j \chi(\lambda_j) h(\lambda_j) \psi_j(v) \overline{\psi_j(w)} &= \frac{-1}{2\pi\ii}\Big{(}\int_{\Gamma_{\eta_1}} \tilde{\chi}(z) h(z) \langle \widehat{Z}_z^{b_v}, (H_{\cQ_N}-z)^{-1} \overline{\widehat{Z}_z^{b_w}} \rangle \, \dd z \\
 &+ \int_{\Omega_{\eta_1}} h(z)   \langle \widehat{Z}_z^{b_v}, (H_{\cQ_N}-z)^{-1} \overline{\widehat{Z}_z^{b_w}} \rangle \frac{\partial \tilde{\chi}}{\partial \overline{z}}\,\dd z\wedge \dd \overline{z}\Big{)}.
\end{align*}

Now, we know that
\[
\big|\langle \widehat{Z}_z^{b_v}, (H_{\cQ_N}-z)^{-1} \overline{\widehat{Z}_z^{b_w}} \rangle\big| \leq \frac{C}{\Im z},
\]
where $C$ does not depend on $N\in \N$ or on $v,w$.

We deduce that 
\begin{multline*}
\bigg| \sum_j \chi(\lambda_j) h(\lambda_j) \psi_j(v) \overline{\psi_j(w)} - \frac{-1}{2\pi\ii}\int_{\Gamma_{\eta_1}} \tilde{\chi}(z) h(z) \langle \widehat{Z}_z^{b_v}, (H_{\cQ_N}-z)^{-1} \overline{\widehat{Z}_z^{b_w}} \rangle\dd z\bigg|\\
 \leq C' \eta_1 \int_{\Omega_{\eta_1}} |h(z)| \dd z \wedge \dd \overline{z}.
\end{multline*}

\textbf{Step 2 : Using the properties of $Z_z$}
We would like to replace the $\widehat{Z}_z$ by $Z_z$ in the previous formula, for the following reason.
Since $b_v\neq b_w$, the map $y \mapsto (H_{\cQ_N}-z)^{-1} (x,y)$ is an eigenfunction on $b_w$, 
so that
\begin{equation}\label{e:toberef}
\big((H_{\cQ_N}-z)^{-1} \overline{Z_z^{b_w}}\big)(x) = (H_{\cQ_N}-z)^{-1}(x,w)
\end{equation}
by \eqref{e:Maxime}. The map $x\mapsto (H_{\cQ_N}-z)^{-1}(x,w)$ is an eigenfunction on $b_v$, so that, by \eqref{e:Maxime} again, we have
\begin{equation*}
\langle Z_z^{b_v}, (H_{\cQ_N}-z)^{-1} \overline{Z_z^{b_w}} \rangle =(H_{\cQ_N}-z)^{-1}(v,w).
\end{equation*}

To estimate the cost of replacing $\widehat{Z}_z$ by $Z_z$, we write
\begin{multline}\label{e:pentecote}
\left\langle \widehat{Z}_z^{b_v}, (H_{\cQ_N}-z)^{-1} \overline{\widehat{Z}_z^{b_w}} \right\rangle  = \left\langle Z_z^{b_v}, (H_{\cQ_N}-z)^{-1} \overline{Z_z^{b_w}} \right\rangle   + \left\langle \widehat{Z}_z^{b_v} - Z_z^{b_v}, (H_{\cQ_N}-z)^{-1} \overline{Z_z^{b_w}} \right\rangle \\
+ \left\langle Z_z^{b_v}, (H_{\cQ_N}-z)^{-1} \left(\overline{\widehat{Z}_z^{b_w} - Z_z^{b_w}} \right) \right\rangle  +\left \langle \widehat{Z}_z^{b_v} - Z_z^{b_v}, (H_{\cQ_N}-z)^{-1}\left( \overline{\widehat{Z}_z^{b_w} - Z_z^{b_w}} \right) \right\rangle 
\end{multline}

By \eqref{eq:ZandHatZ}, the last term is easy to estimate
\begin{equation*}
\begin{aligned}
\left| \left\langle \widehat{Z}_z^{b_v} - Z_z^{b_v}, (H_{\cQ_N}-z)^{-1}\left( \overline{\widehat{Z}_z^{b_w} - Z_z^{b_w}} \right) \right\rangle  \right| &\le C_{I_1,\mathrm{M}}(\Im z^2) \| (H_{\cQ_N}-z)^{-1}\| \\
& \le C_{I_1,\mathrm{M}}|\Im z|.
\end{aligned}
\end{equation*}

As to the second term on the right-hand side of \eqref{e:pentecote}, we use \eqref{eq:ZandHatZ}, the Cauchy-Schwarz inequality and \eqref{e:toberef}, to see that its modulus is bounded by some constant times 
\[
|\Im z| \left\| (H_{\cQ_N}-z)^{-1}(\cdot,w)\right\|_{L^2(b_v)}  :=   |\Im z| \left(\int_{0}^{L_{b_v}} \left|(H_{\cQ_N}-z)^{-1}(x_{b_v},w)\right|^2 \dd x_{b_v}\right)^{1/2}.
\]
Since
\[
(H_{\cQ_N}-z)^{-1}(x_{b_v},w) = \frac{S_{z}(L_{b_v}-x_{b_v})}{S_{z}(L_{b_v})} g_N^z(o_{b_v},w) + \frac{S_{z}(x_{b_v})}{S_{z}(L_{b_v})} g_N^z(t_{b_v},w),
\]
we deduce from \Data{} and \NonDirichlet{} that the second term is bounded by
\[
C_{I,\mathrm{M},C_{\mathrm{Dir}}}|\Im z|(|g_N^z(o_{b_v},w)| + |g_N^z(v,w)|). 
\]

We have a similar estimate for the third term. Therefore, we have
\begin{multline}\label{eq:VWFixed}
\left| \sum_j \chi(\lambda_j) h(\lambda_j) \psi_j(v) \overline{\psi_j(w)} - \frac{-1}{2\pi\ii}\int_{\Gamma_{\eta_1}} \tilde{\chi}(z) h(z) (H_{\cQ_N}-z)^{-1}(v,w)\dd z \right|\\
\leq C \eta_1 \bigg[ \int_{\Omega_{\eta_1}} |h(z)| \dd z \wedge \dd \overline{z} + \int_{\Gamma_{\eta_1}} |h(z)| \big(  |g_N^z(o_{b_v},w)| + 2|g_N^z(v,w)|+ |g_N^z(v,o_{b_w})| +1\big) \dd z   \bigg].
\end{multline}

\textbf{Step 3 : Using \Hol}
Take $\eta_1 = \frac{\eta_0}{4}$ for $\eta_0\in (0,\eta_{I_1})$.
For each $(b_1; b_k)$, $(b_1, b'_{k'})$, we apply \eqref{eq:VWFixed} with $h(z) = K^{z}_{\eta_0}(b_1; b_k) K'^{z}_{\eta_0} (b'_1; b'_{k'}) $.

When summing over $(b_1; b_k)$, $(b_1, b'_{k'})$ and dividing by $N$, 
the first term in the remainder is bounded by
\[
C'\eta_0\sup_{\lambda\in I_1, \eta\in (-\frac{\eta_0}{4},\frac{\eta_0}{4})} \frac{1}{N}  \sum_{(b_1; b_k)} \sum_{\substack{(b'_1; b'_{k'})\\b'_1=b_1}} |K^{\lambda+ \ii \eta}_{\eta_0}(b_1; b_k) K'^{\lambda+ \ii \eta}_{\eta_0} (b'_1; b'_{k'})|,
\]
which is a $O_{N\to +\infty, \eta_0} (\eta_0)$ by the Cauchy-Schwarz inequality and \eqref{eq:APrioriBoundOperators}. 

Concerning the second term, using Cauchy-Schwarz and \eqref{eq:APrioriBoundOperators}, it can be bounded by
\begin{align*}
C'\eta_0\bigg( \sup_{\lambda\in I_1, \eta\in (-\frac{\eta_0}{4},\frac{\eta_0}{4})}  \frac{1}{N} \sum_{(b_1; b_k)} \sum_{\substack{(b'_1; b'_{k'})\\b'_1=b_1}} &\big( |g_N^z(o_{{\check{b}_k}},o_{b'_{k'}})|^2 + |g_N^z(o_{b_k},o_{b'_{k'}})|^2\\
&+|g_N^z(o_{b_k},o_{\check{b}'_{k'}})|^2+1\big) \bigg)^{1/2},
\end{align*}
where $\check{b}_k$, $\check{b}'_{k'}$ are chosen so that $t_{\check{b}_k} = o_{b_k}$ and $t_{\check{b}'_{k'}} = o_{b'_{k'}}$  but $\check{b}_k\neq \check{b}'_{k'}$. Applying Corollary~\ref{cor:ConfinedTilMay11} to $F^z(b_1;b_k)=\sum\limits_{\substack{(b'_1; b'_{k'})\\b'_1=b_1}}(...)$, we deduce this is $O_{N\to +\infty, \eta_0} (\eta_0)$. Proposition \ref{Cor:SumToInt} follows.
\end{proof}

We deduce the following corollary, which we will use several times.

\begin{cor}\label{cor:AAMagic}
Let $\chi \in C_c^\infty(I_1)$, let $K^\gamma\in \mathscr{H}_k$, and  $K'^\gamma\in \mathscr{H}_{k'}$ satisfy \emph{\Hol{}}. Then
\[
\frac{1}{N}\sum_{j\geq 1} \sum_{(b_1; b_k)}\sum_{\substack{(b'_1; b'_{k'})\\ b'_1=b_1}}  \chi(\lambda_j) K^{\gamma_j}(b_1; b_k) K'^{\gamma_j} (b'_1; b'_{k'}) \psi_j(o_{b_k}) \overline{\psi_j(o_{b'_{k'}})}=  O_{N\to +\infty, \eta_0} (1).
\]

The same result holds if $o_{b_k}$ and/or $o_{b'_{k'}}$ are replaced by $t_{b_k}$ and/or $t_{b'_{k'}}$.
\end{cor}

\begin{proof}
By Proposition~\ref{Cor:SumToInt}, up to a term which is $O_{N\to +\infty, \eta_0} (\eta_0)$  the modulus of the quantity we want to estimate is bounded by
\[
\frac{C}{N} \sup_{\lambda\in I_1} \sum_{(b_1; b_k)} \sum_{\substack{(b'_1; b'_{k'})\\ b'_1=b_1}}  \left|g_N^{\lambda\pm \ii\frac{\eta_0}{4}}(o_{b_k},o_{b'_{k'}})\right| \left|K^{\lambda\pm \ii\frac{\eta_0}{4}}_{\eta_0}(b_1; b_k)\right| \left| K'^{\lambda \pm \ii\frac{\eta_0}{4}}_{\eta_0} (b'_1; b'_{k'}) \right|.
\]
Using H\"older's inequality, \eqref{eq:APrioriBoundOperators} and Corollary~\ref{cor:ConfinedTilMay11} this is $O_{N\to +\infty, \gamma}(1)$.
\end{proof}

\begin{rem}\label{rem:CaseDiag}
The very same proof as that of Proposition~\ref{Cor:SumToInt} gives us that, if $K^\gamma \in \mathscr{H}_0$ satisfies \Hol{}, then up to an error $O_{N\to +\infty, \eta_0} (\eta_0)$ we have
\begin{align*}
&\frac{1}{N}\sum_{j\geq 1} \sum_{v\in V} \chi(\lambda_j) K^{\gamma_j}(v) |\psi_j(v)|^2 = -\frac{1}{N}\sum_{v\in V} \frac{1}{2\pi\ii} \int_{\Gamma_{\eta_0/4}} \tilde{\chi}(z) g_N^z(v,v) K^{z}_{\eta_0}(v)  \dd z,
\end{align*}
and, in particular, as in Corollary \ref{cor:AAMagic},
\begin{equation}\label{rem:coro}
\frac{1}{N}\sum_{j\geq 1} \sum_{v\in V} \chi(\lambda_j) K^{\gamma_j}(v) |\psi_j(v)|^2=  O_{N\to +\infty, \eta_0} (1).
\end{equation}
\end{rem}

\subsection{Applications}\label{subsec:ApplicHolom}
The previous results are used here to prove Lemmas~\ref{lem:SmallO} and \ref{lem:EasyVarBound}, and will be used again in the proof of Proposition~\ref{thm:upvar} later on.

\begin{proof}[Proof of Lemma \ref{lem:SmallO}]
Let $\chi \in C_c^\infty(I_1)$ be positive and equal to one on $I$.

For simplicity, we will consider only the term $r=n$, $\ell=n$ in \eqref{eq:RemainderVarInv}. All the other terms can be treated in the same fashion. By \eqref{eq:LowerNumerEigen} and Cauchy-Schwarz, this term is
\begin{multline*}
 \frac{1}{\mathbf{N}_N(I)}\sum_{\lambda_j\in I} |\langle f_j^{\ast},K_B\zeta^{\gamma_j}O_{\psi_j,\eta_0}\rangle|\\
\lesssim \bigg(\frac{1}{N}\sum_{j\ge 1} \chi(\lambda_j) \| f_j^{\ast}\|^2 \bigg)^{1/2} \bigg(\frac{1}{N}\sum_{j\ge 1} \chi(\lambda_j)\|K_B\zeta^{\gamma_j}O_{\psi_j,\eta_0}\|^2 \bigg)^{1/2}.
\end{multline*}

We have $\|f_j^*\|^2 = \sum_{b\in B}\frac{|\psi_j(o_b) - \zeta^{\gamma_j}(\hat{b})\psi_j(t_b)|^2}{S_{\lambda_j}^2(L_b)}$. Applying Corollary \ref{cor:AAMagic}, with $k=k'=1$, $K'=\mathrm{Id}$, and $K$ taking values $K^\gamma_1(b)= \frac{1}{S_{\Re(\gamma)}^2(b)}$ , $K^\gamma_2(b)= \frac{\zeta^{\gamma}(\hat{b})}{S_{\Re(\gamma)}^2(b)} $, $K^\gamma_3(b)= \frac{\overline{\zeta^{\gamma}}(\hat{b})}{S_{\Re(\gamma)}^2(b)} $, $K^\gamma_4(b)= \frac{|\zeta^{\gamma}(\hat{b})|^2}{S_{\Re(\gamma)}^2(b)} $, which all satisfy \Hol{} thanks to Remark~\ref{rem:HolStable}, we deduce that the first factor is $O_{N\to +\infty, \eta_0}(1)$. Concerning the second sum, it can be written
\[
\frac{1}{N}\sum_{j\ge 1} \sum_{(b_1; b_k)}\sum_{\substack{(b'_1; b'_{k})\\ b'_1=b_1}} \chi(\lambda_j)K(b_1;b_k)\overline{K(b_1';b_k')}\zeta^{\gamma_j} (b_k)O_{\eta_0}(b_k)\overline{\zeta^{\gamma_j} (b_k')O_{\eta_0}(b_k')} \psi_j(t_{b_k})\overline{\psi_j(t_{b_k'})} .
\]
Using Remark~\ref{rem:HolStable}, we may bound this as in the proof of Corollary \ref{cor:AAMagic}. Then the result follows from \eqref{eq:OSmall}.
\end{proof}

\begin{proof}[Proof of Lemma \ref{lem:EasyVarBound}]
Let $\mathrm{K}^\gamma\in \mathscr{H}_0$ satisfy \Hol{} and let $M_\gamma(v):= N^\gamma(v) |g_N^\gamma(v,v)|^{-1/2}$. We have by \eqref{eq:LowerNumerEigen} and Cauchy-Schwarz,
\begin{align*}
\mathrm{Var}_{\eta_0}^I (\mathrm{K^\gamma}) &= \frac{1}{\mathbf{N}_N(I)}\sum_{\lambda_j\in I} |\langle \mathring{\psi}_j, K^{\gamma_j} \mathring{\psi}_j \rangle|\\
&\lesssim \bigg(\frac{1}{N}\sum_{\lambda_j\in I} \| M_{\gamma_j}^{-1}\mathring{\psi}_j \|^2   \bigg)^{1/2} \bigg(\frac{1}{N}\sum_{\lambda_j\in I}  \| M_{\gamma_j} K^{\gamma_j} \mathring{\psi}_j \|^2 \bigg)^{1/2}.
\end{align*}

Therefore, if $\chi \in C_c^\infty(I_1)$ is positive, and equal to one on $I$, we have
\begin{align*}
\mathrm{Var}_{\eta_0}^I (\mathrm{K^\gamma})^2 &\lesssim \bigg(\frac{1}{N}\sum_{j\ge 1} \chi(\lambda_j) \sum_{v\in V} \frac{|\psi_j(v)|^2}{|M_{\gamma_j}(v)|^2}\bigg) \bigg(\frac{1}{N}\sum_{j\ge 1} \sum_{v\in V} |M_{\gamma_j}(v)K^{\gamma_j}(v)|^2 |\psi_j(v)|^2 \bigg).
\end{align*}

The first factor is $O_{N\to+\infty, \gamma}(1)$, by \eqref{rem:coro}, as $|M^\gamma|^{-2}$ satisfies \Hol{}.

Next, using Remark \ref{rem:CaseDiag}, up to a term $O_{N\to +\infty, \eta_0} (\eta_0)$, the second factor is
\[
\frac{-1}{2N \pi\ii}\sum_{v\in V_N}\int_{\Gamma_{\eta_0/4}} \tilde{\chi}(z) g_N^z(v,v) (MK\overline{MK})^{z}_{\eta_0}(v)  \dd z.
\]
Using \eqref{eq:APrioriBoundOperators3}, this may be replaced by $\frac{-1}{2N \pi\ii}\sum_{v\in V_N}\int_{\Gamma_{\eta_0/4}} \tilde{\chi}(z) g_N^z(v,v) |M^z(v)K^{z}(v)|^2  \dd z$. Indeed, dominated convergence is applicable by \eqref{eq:APrioriBoundOperators}. The modulus of this is bounded by $\frac{1}{2\pi N} \sum_{v\in V_N} \int_{\Gamma_{\eta_0/4}} |N^z(v) K^z(v)|^2 \dd z$.
\end{proof}

\section{Upper bound on the non-backtracking variance}\label{Sec:UpperBound}
Let $k\geq 1$. Given $K,K'\in \mathscr{H}_k$ and $\gamma=\lambda+\ii\eta_0\in \C^+$, we define the weighted scalar product
\begin{equation}\label{e:hnormgam}
\left(K, K'\right)_{\gamma} := \frac{1}{N}\sum_{(b_1 ; b_k)\in \mathrm{B}_k} \Im R_{\gamma}^-(t_{b_1}) \cdot K(b_1 ; b_k) \overline{K'(b_1;b_k)}  \cdot \Im R_{\gamma}^+(o_{b_k}) \,,
\end{equation}
and $\|K\|_{\gamma}^2 := \left(K, K\right)_{\gamma}$ the associated norm. The aim of this section is to prove the following proposition, which tells us that the non-backtracking quantum variance is dominated by this weighted Hilbert-Schmidt norm.

\begin{prp}\label{thm:upvar}
Let $\overline{I}\subset I_1$, with $I_1$ as in \emph{\Green{}}. There exists $C_I>0$ such that for all $k\geq 1$, and all $K^\gamma\in \mathscr{H}_k$ satisfying hypothesis \emph{\Hol{}} from Definition \ref{def:Hol}, we have
\[
\lim_{\eta_0 \downarrow 0}\limsup_{N\to\infty} \mathrm{Var}^I_{\mathrm{nb}, \eta_0}(K^{\gamma})^2 \leq C_I \lim_{\eta_0\downarrow 0}\limsup_{N\to\infty} \int_{I_1} \|K^{\lambda+\ii\eta_0}\|_{\lambda+\ii\eta_0}^2\,\dd\lambda.
\]
\end{prp}

This proposition is analogous in appearance to Lemma~\ref{lem:EasyVarBound}, but is actually much more involved because the constant $C_I$ does not depend on $k$, which is important for the next sections; it is a lot easier to derive cruder bounds depending on $k$, by arguing as in \S~\ref{subsec:ApplicHolom}.

\begin{proof}
Let $\chi\in C_c^\infty (I_1)$, with $\chi\equiv 1$ on $I$, $0\le \chi\le 1$, and let $|\alpha_{\gamma}(b)|^2 = \Im R_{\gamma}^-(t_b)$ . Denoting $\gamma_j = \lambda_j + \ii \eta_0$ and using \eqref{eq:LowerNumerEigen}, we have
\begin{equation}\label{e:mainterm}
\begin{aligned}
\mathrm{Var}_{\mathrm{nb},\eta_0}^I (K^\gamma)^2 &= \bigg(\frac{1}{\mathbf{N}_N(I)} \sum_{\lambda_j\in I} \left|\langle f_j^{\ast}, K_B^{\lambda_j+\ii\eta_0} f_j\rangle \right|\bigg)^2\\
&\leq C(I) \bigg(\frac{1}{N}\sum_{\lambda_j\in I} \|\alpha_{\gamma_j}^{-1}f_j^{\ast}\|^2\bigg) \bigg( \frac{1}{N}\sum_{\lambda_j\in I} \|\alpha_{\gamma_j} K_B^{\gamma_j} f_j\|^2\bigg)\\
&\leq C(I) \bigg( \frac{1}{N}\sum_{j\geq 1} \chi(\lambda_j)  \|\alpha_{\gamma_j}^{-1}f_j^{\ast}\|^2 \bigg) \bigg( \frac{1}{N}\sum_{j\geq 1} \chi(\lambda_j) \|\alpha_{\gamma_j} K_B^{\gamma_j} f_j\|^2\bigg).
\end{aligned}
\end{equation}

The first factor is $O_{N\longrightarrow +\infty, \gamma}(1)$ by application of Corollary~\ref{cor:AAMagic}, as in Section~\ref{subsec:ApplicHolom}. So we focus on the second factor.

\smallskip

\textbf{Step 1: From a sum to an integral}

Recalling \eqref{eq:DefKB} and definition \eqref{e:fj}, we have
\begin{multline*}
\sum_{j\ge 1}  \chi(\lambda_j)\|\alpha_{\gamma_j} K_B^{\gamma_j} f_j\|^2 = \sum_{b_1\in B_N} \sum_{(b_2;b_k), (b'_2;b'_k)\in \mathrm{B}_{k-1}^{b_1}} \sum_{j\ge 1} \chi(\lambda_j)  |\alpha_{\gamma_j}(b_1)|^2 \\
\times K^{\gamma_j} (b_1; b_k) \Big{(}\frac{\psi_j(t_{b_k})}{S_{\lambda_j}(L_{b_k})} - \frac{\zeta^{\gamma_j}(b_k) \psi_j(o_{b_k})}{S_{\lambda_j}(L_{b_k})} \Big{)} \overline{  K^{\gamma_j} (b_1;b'_k)} \Big{(}\frac{\overline{\psi_j(t_{b_{k}'})}}{S_{\lambda_j}(L_{b_{k}'})} - \frac{\overline{\zeta^{\gamma_j}(b_{k}')\psi_j(o_{b_{k}'})}}{S_{\lambda_j}(L_{b_{k}'})} \Big{)}. 
\end{multline*}

By Remark \ref{rem:HolStable}, $|\alpha_{\gamma}(b_1)|^2$, $K^{\gamma}(b_1;b_k)$, and  $\overline{K^{\gamma}(b_1; b'_k)}$ satisfy \Hol{}, so they have analytic extensions to the strip $\{|\Im z|<\eta_0/2\}$ which we denote by $f_{\alpha,\eta_0}^z(b_1)$, $K^z_{\eta_0}(b_1;b_k)$ and  $\overline{K}^z_{\eta_0}(b_1; b'_k)$, respectively (note that in general $\overline{K}^z_{\eta_0}(b_1; b'_k)\neq \overline{K^z_{\eta_0}(b_1; b'_k)}$). To lighten the expressions a bit, we denote
\[
J_{\eta_0}^z(b_1;b_k;b_k') := \frac{f_{\alpha,\eta_0}^z(b_1)K^z_{\eta_0}(b_1;b_k)\overline{K}^z_{\eta_0}(b_1; b'_k)}{S_z(L_{b_k})S_z(L_{b_k'})}\,.
\]
Next, $\overline{\zeta^{\lambda+\ii \eta_0}(b_{k}')} = \zeta^{\lambda- \ii \eta_0}(b_{k}')$ can be extended holomorphically by $\zeta^{z-\ii\eta_0}(b_{k}')$. We may thus apply Proposition \ref{Cor:SumToInt} to obtain

\begin{multline}\label{eq:FromGtoTilG}
\sum_{j\ge 1} \chi(\lambda_j)\|\alpha_{\gamma_j} K_B^{\gamma_j} f_j\|^2  = \frac{-1}{2\pi \ii} \sum_{b_1\in B_N} \sum_{(b_2;b_k), (b'_2;b'_k)\in \mathrm{B}_{k-1}^{b_1}} \int_{\Gamma_{\frac{\eta_0}{4}}} \tilde{\chi}(z) J_{\eta_0}^z(b_1;b_k;b_k')\\
  \times \Big( g_N^z(t_{b_k},t_{b'_{k}})-  \zeta^{z+\ii\eta_0}(b_k) g_N^z(o_{b_k},t_{b'_{k}})
-\zeta^{z-\ii\eta_0}(b'_k) g_N^z(t_{b_k},o_{b'_{k}})\\
+ \zeta^{z+\ii\eta_0}(b_k)\zeta^{z-\ii\eta_0}(b'_k) g_N^z(o_{b_k},o_{b'_{k}}) \Big) \dd z + N O_{N\to+\infty, \eta_0}(\eta_0).
\end{multline}

By \eqref{eq:support extension}, $\int_{\Gamma_{\eta}}$ reduces to $\int_{I_1-\ii\eta}-\int_{I_1+\ii\eta}$, which we denote by $\int_{\Gamma_{\eta}'}$, $\eta=\frac{\eta_0}{4}$.

\smallskip

\textbf{Step 2: From finite graphs to infinite trees}

The aim of this step is to use the fact that our graphs look locally like trees (assumption \BST) to replace the Green function $g^z_N$ on $\cQ_N$ in \eqref{eq:FromGtoTilG} by $\tilg^z_N$, the Green function of the universal covering tree $\widetilde{\cQ}_N$.

Given a rooted quantum tree $[\cQ,b_1]$, and two paths $p_k=(b_2;b_k)$ and $p_k' = (b_2';b_k')$ in $\mathrm{B}_{k-1}^{b_1}$, let us introduce
\begin{equation}\label{e:4greens}
\begin{aligned}
f^z([\cQ,b_1],p_k,p_k' )&:=  G^z(t_{b_k},t_{b'_{k}})-  \zeta^{z+\ii\eta_0}(b_k) G^z(o_{b_k},t_{b'_{k}})-\zeta^{z-\ii\eta_0}(b'_k) G^z(t_{b_k},o_{b'_{k}})\\
&\quad+ \zeta^{z+\ii\eta_0}(b_k)\zeta^{z-\ii\eta_0}(b'_k) G^z(o_{b_k},o_{b'_{k}})\\
\widetilde{f}^z([\cQ,b_1],p_k, p_k' )&:= \widetilde{G}^z(t_{b_k},t_{b'_{k}})-  \zeta^{z+\ii\eta_0}(b_k) \widetilde{G}^z(o_{b_k},t_{b'_{k}})
-\zeta^{z-\ii\eta_0}(b'_k) \widetilde{G}^z(t_{b_k},o_{b'_{k}})\\
&\quad+ \zeta^{z+\ii\eta_0}(b_k)\zeta^{z-\ii\eta_0}(b'_k) \widetilde{G}^z(o_{b_k},o_{b'_{k}})\,,
\end{aligned}
\end{equation}
where $\widetilde{G}^z(o_{b_k},t_{b'_k})$ is the Green's function of the universal cover $\widetilde{\cQ}$ of the given $\cQ$ (so we fix some lift $\tilde{b}_1$ and consider $\tilde{p}_k,\tilde{p}_k'\in \mathrm{B}_{k-1}^{\tilde{b}_1}$ projecting to $p_k,p_k'$; see \S~\ref{subsec:Univ}). Note that if $\cQ$ is a tree then $f^z([\cQ,b_1],p_k,p_k' ) = \widetilde{f}^z([\cQ,b_1],p_k, p_k' )$. Now
\begin{align*}
&\frac{1}{N} \Big| \sum_{b_1\in B_N} \sum_{p_k, p_k'\in \mathrm{B}_{k-1}^{b_1}} \int_{\Gamma'_{\frac{\eta_0}{4}}} \tilde{\chi}(z)  J_{\eta_0}^z(b_1;b_k;b_k')  \\
&\qquad\qquad\qquad\qquad\qquad \times \left(f^z([\cQ_N,b_1],p_k, p_k' ) - \widetilde{f}^z([\cQ_N,b_1],p_k, p_k' )\right) \dd z \Big|\\
&\leq \bigg(\frac{1}{N} \sum_{b_1\in B_N} \sum_{p_k, p_k'\in \mathrm{B}_{k-1}^{b_1}} \int_{\Gamma'_{\frac{\eta_0}{4}}} \left| \tilde{\chi}(z) J_{\eta_0}^z(b_1;b_k;b_k')\right|^2 \dd z \bigg)^{1/2}\\
&\times  \bigg(\frac{1}{N} \sum_{b_1\in B_N} \sum_{p_k, p_k'\in \mathrm{B}_{k-1}^{b_1}} \int_{\Gamma'_{\frac{\eta_0}{4}}} \left| f^z([\cQ_N,b_1],p_k, p_k' ) - \widetilde{f}^z([\cQ_N,b_1],p_k, p_k' )\right|^2 \dd z \bigg)^{1/2}.
\end{align*}

The first factor is controlled by \eqref{eq:APrioriBoundOperators}. For the second factor, we note that if
\[
F^z_1([\cQ,b_1]):= \sum_{p_k,p_k'\in \mathrm{B}_{k-1}^{b_1}} \left| f^z([\cQ,b_1],p_k, p_k' ) - \widetilde{f}^z([\cQ,b_1],p_k, p_k' )\right|^2
\] 
then the same arguments leading to \eqref{eq:stuckAtHomeForTenDays} show that
\[
\limsup_{N\to +\infty} \frac{1}{N}\int_{\Gamma'_{\frac{\eta_0}{4}}} \sum_{b_1\in B_N} F^z_1([\cQ_N,b_1]) \dd z  \le C\sup_{\lambda\in I_1}\expect_{\prob}(F^{\lambda\pm \ii \frac{\eta_0}{4}})=0,
\]
since the last expectation runs over trees $[\cQ,(b_1,x)]$ and thus $F^z ([\cQ,(b_1,x)])= 0$.

Combining all this with \eqref{eq:FromGtoTilG}, we obtain that
\begin{multline}\label{eq:FromGtoTilG2}
\limsup\limits_{N\to +\infty} \frac{1}{N} \sum_{\lambda_j\in I} \|\alpha_{\gamma_j} K_B^{\gamma_j} f_j\|^2  \leq \limsup\limits_{N\to +\infty}  \frac{-1}{2N \pi \ii} \sum_{b_1\in B_N} \sum_{p_k,p_k'\in \mathrm{B}_{k-1}^{b_1}} \int_{\Gamma'_{\frac{\eta_0}{4}}} \tilde{\chi}(z)  J_{\eta_0}^z(b_1;b_k;b_k')\\
  \times  \widetilde{f}^z([\cQ_N,b_1],p_k, p_k' ) \dd z + O(\eta_0).
\end{multline}

Let us write $\mathcal{B}_N^{1,k}:= \{b\in B_N: \rho_{G_N}(o_b)\le 2k\}$, where $\rho_{G_N}(o_b)$ is the injectivity radius of $o_b$, and $\mathcal{B}_N^{2,k}:= B_N\setminus \mathcal{B}_N^{1,k}$. By \BST{}, we have $|\mathcal{B}_N^{1,k}|=o(N)$. So by Cauchy-Schwarz, if $Y_N(b_1) = \sum_{p_k,p_k'}\int(\ldots)$, we have
\begin{equation}\label{e:bs}
\bigg|\frac{1}{N}\sum_{b_1\in \cB_N^{1,k}} Y_N(b_1)\bigg| \le \bigg(\frac{|\cB_N^{1,k}|}{N}\bigg)^{1/2}\bigg(\frac{1}{N}\sum_{b_1\in B_N} |Y_N(b_1)|^2\bigg)^{1/2} \Lim_{N\to+\infty} 0
\end{equation}
using \eqref{eq:APrioriBoundOperators}. Therefore, we may replace $\sum_{b_1\in B_N}$ above by $\sum_{b_1\in \cB_N^{2,k}}$.

\smallskip

\textbf{Step 3: Off-diagonal terms vanish}

Here we mean to show that the terms with $p_k\not= p'_k$ in \eqref{eq:FromGtoTilG2} vanish. Suppose $p'_k=(b_2';b_k')\neq (b_2;b_k)=p_k$. If $b_1\in \cB_N^{2,k}$, this implies $t_{b_k}\neq t_{b_k'}$. Using \eqref{e:greenmul} on $(v_0,\dots,v_s)$ where $v_0=t_{b_k}$, $v_1=o_{b_k}$, $v_{s-1}=o_{b_k'}$, $v_s=t_{b_k'}$ we obtain
\begin{equation}\label{eq:IAmListeningToJimmyHendrix}
\tilg_N^z(t_{b_k},t_{b'_{k}}) - \zeta^{z+\ii\eta_0}(b_k)\tilg_N^z(o_{b_k},t_{b_k'}) = \zeta^z(b_k')\tilg_N^z(t_{b_k},o_{b_k'}) - \zeta^{z+\ii\eta_0}(b_k)\zeta^z(b_k')\tilg_N^z(o_{b_k},o_{b_k'}).
%\begin{aligned}
%\tilde{g}_N^z(t_{b_k},t_{b'_{k}})-  \zeta^{z-\ii\eta_0}(b_k') \tilde{g}_N^z(t_{b_k},o_{b'_{k}})&= \left(\zeta^{z}(b_k')-\zeta^{z-\ii\eta_0}(b_k')\right) \tilde{g}_N^z(t_{b_k},o_{b'_{k}})\\
%-\zeta^{z+\ii\eta_0}(b_k)&\tilde{g}_N^z(o_{b_k},t_{b'_{k}}) + \zeta^{z+\ii\eta_0}(b_k)\zeta^{z-\ii\eta_0}(b'_k) \tilde{g}_N^z(o_{b_k},o_{b'_{k}})\\
%&=\zeta^{z+\ii\eta_0}(b_k) \left(-\zeta^z(b'_k)+\zeta^{z-\ii\eta_0}(b'_k)\right)  \tilde{g}_N^z(o_{b_k},o_{b'_{k}}).
%\end{aligned}
\end{equation}

Similarly, using \eqref{e:sym} then \eqref{e:greenmul} on $(v_s,\dots,v_0)$, then \eqref{e:sym}, we have
\begin{equation}\label{eq:IAmListeningToJimmyHendrix2}
\tilg_N^z(t_{b_k},t_{b'_{k}}) - \zeta^{z-\ii\eta_0}(b_k')\tilg_N^z(t_{b_k},o_{b_k'}) = \zeta^z(b_k)\tilg_N^z(o_{b_k},t_{b_k'}) - \zeta^{z-\ii\eta_0}(b_k')\zeta^z(b_k)\tilg^z_N(o_{b_k},o_{b_k'}).
%\begin{aligned}
%\tilde{g}_N^z(t_{b_k},t_{b'_{k}})-  \zeta^{z+\ii\eta_0}(b_k) \tilde{g}_N^z(o_{b_k},t_{b'_{k}})&= \left(\zeta^{z}(b_k)-\zeta^{z+\ii\eta_0}(b_k)\right) \tilde{g}_N^z(o_{b_k},t_{b'_{k}})\\
%-\zeta^{z-\ii\eta_0}(b'_k)&\tilde{g}_N^z(t_{b_k},o_{b'_{k}}) + \zeta^{z+\ii\eta_0}(b_k)\zeta^{z-\ii\eta_0}(b'_k) \tilde{g}_N^z(o_{b_k},o_{b'_{k}})\\
%&=\zeta^{z-\ii\eta_0}(b'_k) \left(-\zeta^z(b_k)+\zeta^{z+\ii\eta_0}(b_k)\right)  \tilde{g}_N^z(o_{b_k},o_{b'_{k}}).
%\end{aligned}
\end{equation}
%Note that the two expressions \eqref{eq:IAmListeningToJimmyHendrix} and \eqref{eq:IAmListeningToJimmyHendrix2} vanish if we set $\eta_0=0$. We now check that they also go to zero in the limit $\eta_0>0, \eta_0\To 0$.

Let us first consider the part of $\Gamma'_{\frac{\eta_0}{4}}$ where $\Im z<0$. Recalling \eqref{e:4greens}, if we use \eqref{eq:IAmListeningToJimmyHendrix} along with the Cauchy-Schwarz inequality, we obtain
\begin{multline}\label{e:doublim}
\lim_{\eta_0\downarrow 0}  \limsup_{N\rightarrow \infty} \Big|\frac{1}{N} \sum_{b_1\in \mathcal{B}_N^{2,k}} \sum_{(b_2;b_k)}\sum_{(b'_2;b'_k)\neq (b_2;b_k)} \int_{I_1} \tilde{\chi}(\lambda-\ii \eta_0/4)  J_{\eta_0}^{\lambda-\frac{\ii\eta_0}{4}}(b_1;b_k;b_k')\\
  \times  \widetilde{f}^{\lambda-\frac{\ii\eta_0}{4}}([\cQ_N,b_1],p_k, p_k' ) \dd\lambda\Big{|}\\
\leq C  \lim_{\eta_0\downarrow 0}  \limsup_{N\rightarrow \infty} \sup_{\lambda\in I_1} \Big[\frac{1}{N}\sum_{b_1\in B_N} \sum_{(b_2;b_k), (b'_2;b'_k)\in \mathrm{B}_{k-1}^{b_1}} \Big{|} J_{\eta_0}^{\lambda-\frac{\ii\eta_0}{4}}(b_1;b_k;b_k') \\
\times  \Big( |\tilde{g}_N^{\lambda-\frac{\ii \eta_0}{4}}(t_{b_k},o_{b'_{k}})| +  |\zeta^{\lambda+\frac{3\ii\eta_0}{4}}(b_k) \tilde{g}_N^{\lambda-\frac{\ii \eta_0}{4}}(o_{b_k},o_{b'_{k}})|\Big) \Big{|}^2 \Big]^{1/2}\\
\times \lim_{\eta_0\downarrow 0}  \int_{I_1}  \limsup_{N\rightarrow \infty}  \bigg\{\frac{1}{N} \sum_{b_1\in B_N} \sum_{(b_2;b_k), (b'_2;b'_k)\in \mathrm{B}_{k-1}^{b_1}} \left| \zeta^{\lambda-\frac{\ii \eta_0}{4}}(b_k')-\zeta^{\lambda-\frac{5\ii\eta_0}{4}}(b_k')\right|^2 \bigg\}^{1/2} \mathrm{d}\lambda.
\end{multline}

Thanks to \eqref{eq:APrioriBoundOperators}, the first factor is finite. Concerning the second factor, thanks to Corollary~\ref{cor:ConfinedTilMay11}, the integrand is bounded independently of $\eta_0$, and, by Proposition \ref{prop:GreenisCool3}, it goes to zero almost everywhere as $\eta_0\downarrow 0$. Therefore, by the dominated convergence theorem, the double limit \eqref{e:doublim} is zero.

For the part of $\Gamma_{\eta_0/4}'$ with $\Im z>0$, we argue similarly, using \eqref{eq:IAmListeningToJimmyHendrix2} instead of \eqref{eq:IAmListeningToJimmyHendrix}.

From what precedes we thus obtain
\begin{align*}
\lim_{\eta_0\downarrow 0} &\limsup_{N\rightarrow \infty} \frac{1}{N}\sum_{\lambda_j\in I} \chi(\lambda_j)\|\alpha_{\gamma_j} K_B^{\gamma_j} f_j\|^2   \\
&\leq \lim_{\eta_0\downarrow 0}  \limsup_{N\rightarrow \infty}  \frac{-1}{2\pi N\ii} \sum_{b_1\in \mathcal{B}^{2,k}_N} \sum_{(b_2;b_k)\in \mathrm{B}_{k-1}^{b_1}} \int_{\Gamma'_{\frac{\eta_0}{4}}} \tilde{\chi}(z)  J_{\eta_0}^z(b_1;b_k;b_k) \widetilde{f}^{z}([\cQ_N,b_1],p_k, p_k ) \dd z\\
&=:(\star)
\end{align*}

As in \eqref{e:bs}, we may replace $\sum_{b_1\in\cB_N^{2,k}}$ by $\sum_{b_1\in B}$.

\smallskip

\textbf{Step 4: Adjusting the energies.}
To finish the proof of Proposition \ref{thm:upvar}, there remains to put $(\star)$ into final form by setting all the spectral parameters equal to $\lambda+\ii\eta_0$ with $\lambda\in I_1$.

Thanks to \eqref{eq:APrioriBoundOperators3}, we have for almost all $\lambda\in I_1$,
\begin{multline*}
\lim_{\eta_0\downarrow 0} \limsup_{N\to +\infty} \bigg(\frac{1}{N} \sum_{(b_1;b_k)\in \mathrm{B}_k}\Big| \tilde{\chi}(\lambda\pm \frac{\ii\eta_0}{4}) J_{\eta_0}^{\lambda\pm\frac{\ii\eta_0}{4}}(b_1;b_k;b_k)  \\
-  \tilde{\chi}(\lambda)   \frac{\big{|} K^{\lambda+\ii\eta_0} (b_1;b_k)\big{|}^2}{S_{\lambda }^2(L_{b_k})} |\alpha_{\lambda+\ii\eta_0}(b_1)|^2\Big|^2\bigg)^{1/2} =0.
\end{multline*}

Using \eqref{eq:APrioriBoundOperators} and Remark \ref{rem:HolStable}, we may therefore apply the dominated convergence theorem to deduce that
\begin{multline*}
(\star)
=  \lim_{\eta_0\downarrow 0} \limsup_{N\rightarrow \infty}  \frac{1}{\pi N}\sum_{(b_1;b_k)\in \mathrm{B}_k}\int_{I_1}
\Big[ \chi(\lambda)   |K^{\lambda+\ii\eta_0} (b_1;b_k)|^2 |\alpha_{\lambda+\ii\eta_0}(b_1)|^2 \\
\times \Im \Big\{ \frac{ \widetilde{f}^{\lambda+\frac{\ii\eta_0}{4}}([\cQ_N,b_1],p_k, p_k )}{S_{\lambda}^2(L_{b_k})} \Big\}\Big] \dd \lambda,
\end{multline*}
where we used that $\widetilde{f}^{\lambda-\frac{\ii\eta_0}{4}}(p_k,p_k)=\overline{\widetilde{f}^{\lambda+\frac{\ii\eta_0}{4}}}(p_k,p_k)$, as readily checked.

Recalling \eqref{e:4greens}, we note that the energies in $\frac{\widetilde{f}^{\lambda+\frac{\ii\eta_0}{4}}}{S_{\lambda}^2}$ are not homogeneous. We thus use Proposition~\ref{prop:GreenisCool3} and dominated convergence to replace $\tilg^{\lambda+\frac{\ii\eta_0}{4}}\mapsto \tilg^{\lambda+\ii \eta_0}$, $\frac{\zeta^{\lambda+\frac{5\ii\eta_0}{4}}}{S_{\lambda}}\mapsto \frac{\zeta^{\lambda+\ii\eta_0}}{S_{\lambda+\ii\eta_0}}$, $\frac{\zeta^{\lambda-\frac{3\ii\eta_0}{4}}}{S_{\lambda}}\mapsto \frac{\zeta^{\lambda-\ii\eta_0}}{S_{\lambda-\ii\eta_0}}$. Let $\gamma:=\lambda+\ii\eta_0$. Using \eqref{e:sym}, we get in this fashion 
\begin{align*}
(\star)&=\lim_{\eta_0\downarrow 0} \limsup_{N\rightarrow \infty}  \frac{1}{\pi N} \sum_{(b_1;b_k)\in \mathrm{B}_k} \int_{I_1} \chi(\lambda)  \big{|}K^{\gamma} (b_1;b_k)\big{|}^2|\alpha_{\gamma}(b_1)|^2 \\
 &\qquad\times \Im \Big[ \frac{\tilde{g}_N^{\gamma}(t_{b_k},t_{b_{k}})}{S^2_{\gamma}(L_{b_k})} - 2 \Re\Big{(}\frac{\zeta^{\gamma}(b_k)}{S_{\gamma}(L_{b_k})}\Big{)}  \frac{\tilde{g}_N^{\gamma}(o_{b_k},t_{b_{k}})}{S_{\gamma}(L_{b_k})} + \Big|\frac{\zeta^{\gamma}(b_k)}{S_{\gamma}(L_{b_k})}\Big|^2  \tilde{g}_N^{\gamma}(o_{b_k},o_{b_{k}})\Big] \dd \lambda  \\
&\le \lim_{\eta_0\downarrow 0} \limsup\limits_{N\rightarrow \infty}  \frac{1}{ N} \sum_{(b_1;b_k)\in \mathrm{B}_k} \int_{I_1}   \big|K^{\gamma} (b_1;b_k)\big|^2 |\alpha_{\gamma}(b_1)|^2  \Im R_{\gamma}^+ (o_{b_k}) \dd\lambda\\
 &=   \lim_{\eta_0\downarrow 0} \limsup_{N\rightarrow \infty} \int_{I_1}  \|K^{\lambda+\ii\eta_0}\|^2_{\lambda+\ii\eta_0} \dd \lambda,
\end{align*}
where the last inequality is by \eqref{e:psiiden1}. Recalling \eqref{e:mainterm}, this completes the proof.
\end{proof}

\section{Estimating the Hilbert-Schmidt norm}\label{Sec:HSestimate}
We are now in Step (3) of the proof as described on page~\pageref{page:step3} of the Introduction. We combine the invariance of the quantum variance (Proposition \ref{p:VarInv}), Lemma \ref{lem:SmallO}, and the domination of the quantum variance by a weighted Hilbert-Schmidt norm (Proposition \ref{thm:upvar}). This yields, for any $n\in \N$, any $k\in \N$ and any family $K^\gamma$ of operators\footnote{Recall that all the operators and the quantities we manipulate here depend on $N$, although this dependence is not explicit in our notations.} in $\mathscr{H}_k$ satisfying \Hol{}:

\begin{equation}\label{eq:SummaryStep6}
\lim_{\eta_0 \downarrow 0}\limsup_{N\to\infty} \mathrm{Var}^I_{\mathrm{nb},\eta_0}(K^{\gamma})^2 \lesssim \lim_{\eta_0\downarrow 0}\limsup_{N\to\infty} \int_{I_1} \Big{\|}\frac{1}{n}\sum_{r=1}^n\cR_{n,r}^{\lambda+\ii\eta_0}K^{\lambda+\ii\eta_0}\Big{\|}_{\lambda+\ii\eta_0}^2\,\dd\lambda.
\end{equation}

We now estimate $\| \frac{1}{n}\sum_{r=1}^n \cR_{n,r}^{\gamma} K^{\gamma}\|_{\gamma}^2$ by developing the scalar product. Let $r\ge r'$, so that $n-r\le n-r'$. By~\eqref{e:rnr} and \eqref{e:hnormgam}, we have
\begin{multline*}
\left( \cR_{n,r}^{\gamma}K^{\gamma},\cR_{n,r'}^{\gamma}K^{\gamma}\right)_{\gamma} = \frac{1}{N} \sum_{(b_1;b_{n+k})\in B_{n+k}} \Im R_{\gamma}^-(t_{b_1}) |\zeta^{\gamma}(\hat{b}_2)\cdots\zeta^{\gamma}(\hat{b}_{n-r+1})|^2 \\ 
\cdot  \overline{\zeta^{\gamma}(\hat{b}_{n-r+2})\cdots\zeta^{\gamma}(\hat{b}_{n-r'+1})}K^{\gamma}(b_{n-r'+1} ;  b_{n-r'+k})\overline{K^{\gamma}(b_{n-r+1} ; b_{n-r+k})} \\
\cdot \overline{\zeta^{\gamma}(b_{n-r+k})\cdots\zeta^{\gamma}(b_{n-r'+k-1})}|\zeta^{\gamma}(b_{n-r'+k})\cdots\zeta^{\gamma}(b_{n+k-1})|^2 \Im R_{\gamma}^+(o_{b_{n+k}}) \,.
\end{multline*}
To simplify this expression, we use \eqref{e:cur1}, \eqref{e:cur2} repeatedly to obtain
\begin{multline}\label{e:aneq}
\left( \cR_{n,r}^{\gamma}K^{\gamma},R_{n,r'}^{\gamma}K^{\gamma}\right)_{\gamma} = \frac{1}{N}\sum_{(b_{n-r+1};b_{n-r'+k})\in B_{k+r-r'}} \Im R_{\gamma}^-(t_{b_{n-r+1}}) \overline{K^{\gamma}(b_{n-r+1} ; b_{n-r+k})}\\
\cdot \overline{\zeta^{\gamma}(b_{n-r+k})\cdots \zeta^{\gamma}(b_{n-r'+k-1})}
\cdot \overline{\zeta^{\gamma}(\hat{b}_{n-r+2})\cdots\zeta^{\gamma}(\hat{b}_{n-r'+1})} \\
\cdot K^{\gamma}(b_{n-r'+1} ; b_{n-r'+k}) \Im R_{\gamma}^+(o_{b_{n-r'+k}}) - \mathbf{E}_{n,r,r'}(\eta_0,K^{\gamma}) \,,
\end{multline}
where $\mathbf{E}_{n,r,r'}$ is an error term defined by
\begin{align*}
&\mathbf{E}_{n,r,r'}(\eta_0,K^{\gamma})  = \frac{\eta_0}{N}\sum_{s=2}^{n-r+1} \sum_{(b_s;b_{n+k})}\int_{0}^{L_{b_s}} |\xi_-^{\gamma}(x_{b_s})|^2\,\dd x_{b_s} \cdot |\zeta^{\gamma}(\hat{b}_{s+1})\cdots\zeta^{\gamma}(\hat{b}_{n-r+1})|^2 \\
& \qquad\cdot \overline{\zeta^{\gamma}(\hat{b}_{n-r+2})\cdots\zeta^{\gamma}(\hat{b}_{n-r'+1})}K^{\gamma}(b_{n-r'+1} ;b_{n-r'+k})\overline{K^{\gamma}(b_{n-r+1}; b_{n-r+k})}  \\
&\qquad \cdot \overline{\zeta^{\gamma}(b_{n-r+k})\cdots \zeta^{\gamma}(b_{n-r'+k-1})} \cdot |\zeta^{\gamma}(b_{n-r'+k})\cdots \zeta^{\gamma}(b_{n+k-1})|^2 \Im R_{\gamma}^+(o_{b_{n+k}}) \\
& \quad + \frac{\eta_0}{N}\sum_{s'=n-r'+k}^{n+k-1} \sum_{(b_{n-r+1};b_{s'})} \Im R_{\gamma}^-(t_{b_{n-r+1}}) \overline{\zeta^{\gamma}(\hat{b}_{n-r+2})\cdots\zeta^{\gamma}(\hat{b}_{n-r'+1})} \\
& \qquad \cdot K^{\gamma}(b_{n-r'+1} ; b_{n-r'+k})\overline{K^{\gamma}(b_{n-r+1} ; b_{n-r+k})}\overline{\zeta^{\gamma}(b_{n-r+k})\cdots\zeta^{\gamma}(b_{n-r'+k-1})}   \\
& \qquad   \cdot |\zeta^{\gamma}(b_{n-r'+k})\cdots\zeta^{\gamma}(b_{s'-1})|^2\int_{0}^{L_{b_{s'}}} |\xi_+^{\gamma}(x_{b_{s'}})|^2\,\dd x_{b_{s'}} \,,
\end{align*}
with the $\xi^\gamma_\pm$ as in \eqref{eq:DefXi}. 

Thanks to Remark \ref{rem:HolStable} and \eqref{eq:APrioriBoundOperators}, we have that for any $n\in \N$, any $r,r'\leq n$
\begin{equation}\label{eq:SmallO2}
\mathbf{E}_{n,r,r'}(\eta_0,K^{\gamma}) = O_{N\to +\infty, \gamma}^{(n)}(\eta_0).
\end{equation}

Introduce the operator acting on $\mathscr{H}_k$,
\[
(\mathbf{A}_{\gamma} K^\gamma)(b_1 ; b_k) = \sum_{b_{k+1}\in \cN_{b_k}^+} \frac{\overline{\zeta^{\gamma}(\hat{b}_2)\zeta^{\gamma}(b_k)}}{\Im R_{\gamma}^+(o_{b_k})} \Im R_{\gamma}^+(o_{b_{k+1}}) K^\gamma(b_2 ; b_{k+1}) \,.
\]
Calculating $(\mathbf{A}_{\gamma}^{r-r'}K^{\gamma})(b_{n-r+1}; b_{n-r+k})$, we find that \eqref{e:aneq} takes the form
\[
\left( \cR_{n,r}^{\gamma}K^{\gamma},\cR_{n,r'}^{\gamma}K^{\gamma}\right)_{\gamma} = \left( K^{\gamma}, \mathbf{A}_{\gamma}^{r-r'} K^{\gamma}\right)_{\gamma} - \mathbf{E}_{n,r,r'}(\eta_0,K^{\gamma}) \,.
\]

We finally introduce the operator
\[
(\cS_{u^{\gamma}}K)(b_1 ; b_k) = \frac{|\zeta^{\gamma}(b_k)|^2}{\Im R_{\gamma}^+(o_{b_k})} u^{\gamma}(b_k) \sum_{b_{k+1}\in \cN_{b_k}^+} \Im R_{\gamma}^+(o_{b_{k+1}}) K(b_2 ; b_{k+1}) \,,
\]
where $u^{\gamma}(b) = \overline{\zeta^{\gamma}(b)}\zeta^{\gamma}(b)^{-1}$ has modulus one. 
When we want to remember that $\cS_{u^{\gamma}}$ acts on $\mathscr{H}_k \simeq \C^{B_k}$, we will denote it by $\cS_{u^{\gamma}}^{(k)}$.

The advantage of this operator over $\mathbf{A}_{\gamma}$ is that, if we forget the multiplication by $u^{\gamma}$, then it is sub-stochastic: $\cS_1 \mathbf{1} \le \mathbf{1}$, using \eqref{e:ASW}. To link it to $\mathbf{A}_{\gamma}$, introduce
\[
(Z_{\gamma} K)(b_1 ; b_k) = \frac{\overline{\zeta^{\gamma}(\hat{b}_1)\cdots\zeta^{\gamma}(\hat{b}_k)}}{\overline{\zeta^{\gamma}(\hat{b}_1)\cdot \zeta^{\gamma}(\hat{b}_k)}} \cdot \overline{\tilg^{\gamma}(o_{b_k},o_{b_k})} K^{\gamma}(b_1 ; b_k) \,.
\]
In particular, if $k=1$, $(Z_{\gamma} K)(b_1) = \frac{\overline{\tilg^{\gamma}(o_{b_1},o_{b_1})}}{\overline{\zeta^{\gamma}(\hat{b}_1)}}K(b_1)$ and if $k=2$, $(Z_{\gamma}K)(b_1,b_2) = \overline{\tilg^{\gamma}(o_{b_2},o_{b_2})} K(b_1,b_2)$.

We observe that the multiplication operator $Z_{\gamma}$ conjugates $\cS_{u^{\gamma}}$ and $\mathbf{A}_{\gamma}$: using \eqref{e:zetainv},
\[
Z_{\gamma} \cS_{u^{\gamma}} Z_{\gamma}^{-1} = \mathbf{A}_{\gamma} \,.
\]
Let us introduce the following measure on $\mathrm{B}_k$
\begin{equation} \label{eq:defMeasureMu}
\mu_k^{\gamma}(b_1 ; b_k) = \frac{\Im R_{\gamma}^-(t_{b_1})}{|\zeta^{\gamma}(\hat{b}_1)|^2} \cdot |\zeta^{\gamma}(\hat{b}_1)\cdots\zeta^{\gamma}(\hat{b}_k)|^2 \cdot |\tilg^{\gamma}(o_{b_k},o_{b_k})|^2 \cdot \frac{\Im R_{\gamma}^+(o_{b_k})}{|\zeta^{\gamma}(\hat{b}_k)|^2} \,.
\end{equation}
Then, if $K, K'\in \mathscr{H}_k$, we have $\langle Z_{\gamma} K ,Z_{\gamma} K'\rangle_{\gamma} = \frac{1}{N} \langle K,K'\rangle_{\ell^2(\mathrm{B}_k,\mu_k^{\gamma})}$. Hence,
\[
\langle \cR_{n,r}^{\gamma}K^{\gamma},\cR_{n,r'}^{\gamma}K^{\gamma}\rangle_{\gamma} = \frac{1}{N}\langle Z_{\gamma}^{-1}K^{\gamma},\cS_{u^{\gamma}}^{r-r'}Z_{\gamma}^{-1}K^{\gamma}\rangle_{\ell^2(\mathrm{B}_k,\mu_k^{\gamma})} - \mathbf{E}_{n,r,r'}(\eta_0,K^{\gamma}) \,.
\]

From now on, we will write
\begin{equation*}
K_\gamma':= Z_\gamma^{-1} K^\gamma.
\end{equation*}

Note that by Remark \ref{rem:HolStable}, if $K^\gamma$ satisfies \Hol{}, then so does $K'_\gamma$.

Note we also have $\mu_k^{\gamma}(b_1 ;b_k) = \frac{\Im R_{\gamma}^-(t_{b_1})}{|\zeta^{\gamma}(b_1)|^2} |\tilg^{\gamma}(t_{b_1},t_{b_1})|^2|\zeta^{\gamma}(b_1)\cdots\zeta^{\gamma}(b_k)|^2\frac{\Im R_{\gamma}^+(o_{b_k})}{|\zeta^{\gamma}(b_k)|^2}$ using \eqref{e:greenmul} or \eqref{e:zetainv}. In particular, we see using \eqref{e:ASW} that
\begin{equation}\label{eq:AlmostStoch1}
\sum_{b_k\in\cN_{b_{k-1}}^+} \mu_k^{\gamma}(b_1;b_k) \le \mu_{k-1}^{\gamma}(b_1; b_{k-1}) \quad \text{and}\quad \sum_{b_1\in \cN_{b_2}^-} \mu_k^{\gamma}(b_1 ;  b_k) \le \mu_{k-1}^{\gamma}(b_2 ; b_k) \,.
\end{equation}
In particular,
\begin{equation}\label{e:mukdecreases}
\mu_k^{\gamma}(\mathrm{B}_k)=\sum_{(b_1 ; b_k)\in\mathrm{B}_k}\mu_k^{\gamma}(b_1 ; b_k) \le \sum_{b_1\in\mathrm{B}_1}\mu_1^{\gamma}(b_1)=\mu_1^{\gamma}(\mathrm{B}_1).
\end{equation}
Note that this inequality becomes an equality for $\gamma\in \R$ (whenever
both sides are defined).

\begin{rem}\label{rem:MeasureNuIsLessScaryThanCoronavirus}
Clearly $\mu_k^\gamma(b_1;b_k)^{\pm 1}$ belongs to $\mathscr{L}_k^\gamma$. It follows from Corollary~\ref{cor:ConfinedTilMay11} that for any $s\in\R$, we have 
\begin{equation}\label{eq:Itworksforpowersaswell}
\frac{1}{N}\sum_{b\in B} \mu_1^\gamma(b)^s = O_{N\to +\infty, \gamma}^{(s)}(1),
\end{equation}
\end{rem}

Let us come back to the problem of estimating $\Big{\|} \frac{1}{n}\sum_{r=1}^n \cR_{n,r}^{\gamma} K^{\gamma} \Big{\|}_{\gamma}^2$. Writing 
\[
\mathbf{E}'_n(\eta_0, K^\gamma) := \frac{1}{n^2}\sum_{r,r'=1}^n \mathbf{E}_{n,r,r'}(\eta_0,K^\gamma),
\]
we have
\begin{align}\label{eq:DevelopSquare}
\Big{\|} \frac{1}{n}\sum_{r=1}^n \cR_{n,r}^{\gamma}  K^{\gamma} \Big{\|}_{\gamma}^2 &= \frac{1}{n^2}\sum_{r=1}^n \| \mathcal{R}_{n,r}^\gamma K^\gamma\|_\gamma^2 + \frac{2}{n^2} \sum_{r=1}^{n-1}\sum_{r'=r+1}^n \Re \left( \mathcal{R}_{n,r}^\gamma K^\gamma, \mathcal{R}_{n,r'}^\gamma K^\gamma \right)_\gamma\nonumber\\
&= \frac{1}{n} \big{\|}K^{\gamma} \big{\|}^2_{\gamma} + \frac{2}{N n^2} \sum_{r= 1}^{n-1}\sum_{r'= r+1}^n   \Re \langle K'_{\gamma},\cS_{u^{\gamma}}^{r'-r}K'_{\gamma}\rangle_{\ell^2(\mathrm{B}_k,\mu_k^{\gamma})} - \Re \mathbf{E}'_{n}(\eta_0,K^{\gamma})\nonumber\\
&= \frac{1}{n} \big{\|}K^{\gamma} \big{\|}^2_{\gamma} + \frac{2}{N n^2} \sum_{r= 1}^{n-1}\sum_{j=1}^{k-1}   \Re \langle K'_{\gamma},\cS_{u^{\gamma}}^{j}K'_{\gamma}\rangle_{\ell^2(\mathrm{B}_k,\mu_k^{\gamma})}  \nonumber\\
&\quad+ \frac{2}{N n^2} \sum_{r= 1}^{n-k}\sum_{j=k}^{n-r}   \Re \langle K'_{\gamma},\cS_{u^{\gamma}}^{j}K'_{\gamma}\rangle_{\ell^2(\mathrm{B}_k,\mu_k^{\gamma})} -\Re \mathbf{E}'_{n}(\eta_0,K^{\gamma}) \nonumber\\
& \leq \frac{2k-1}{n} \|K^\gamma\|_\gamma^2 + \frac{2}{N n^2} \Big|\sum_{r= 1}^{n-k}\sum_{j=k}^{n-r}   \langle K'_{\gamma},\cS_{u^{\gamma}}^{j}K'_{\gamma}\rangle_{\ell^2(\mathrm{B}_k,\mu_k^{\gamma})}\Big| + |\mathbf{E}'_{n}(\eta_0,K^{\gamma})|.
\end{align}
where, for the last inequality, we used that $\cS_1^\ast$ is also sub-stochastic, cf. \eqref{eq:TrivialContractionBound}.

The first term in \eqref{eq:DevelopSquare} will go to zero as $N\to \infty$ followed by $n\to \infty$ thanks to \eqref{eq:APrioriBoundOperators}.
Therefore, combining \eqref{eq:DevelopSquare} with \eqref{eq:SummaryStep6} and \eqref{eq:SmallO2}, we obtain that
\begin{multline*}
\lim_{\eta_0 \downarrow 0}\limsup_{N\to\infty} \mathrm{Var}^I_{\mathrm{nb},\eta_0}(K^{\gamma})^2 \le \lim_{n\to \infty} \lim_{\eta_0\downarrow 0}\limsup_{N\to\infty} \sup_{\lambda\in I_1}  \frac{C_I}{N n^2} \Big|\sum_{r= 1}^{n-k}\sum_{j=k}^{n-r}   \langle K'_{\gamma},\cS_{u^{\gamma}}^{j}K'_{\gamma}\rangle_{\ell^2(\mathrm{B}_k,\mu_k^{\gamma})}\Big|\\
\le \lim_{n\to \infty} \lim_{\eta_0\downarrow 0}\limsup_{N\to\infty} \sup_{\lambda\in I_1}  \frac{C_I}{N n^2} \Big|\sum_{j=k}^{n} (n-j)  \langle K'_{\gamma},\cS_{u^{\gamma}}^{j}K'_{\gamma}\rangle_{\ell^2(\mathrm{B}_k,\mu_k^{\gamma})}\Big|.
\end{multline*}

We must therefore understand the contraction properties of $\cS_{u^{\gamma}}^{j}$. 

In this expression, $\cS_{u^{\gamma}}^{j}$ acts on $\mathscr{H}_k$; it would more adequately be written as $\big(\cS_{u^{\gamma}}^{(k)}\big)^{j}$. We will now explain why it suffices to understand the contraction properties of $\big{(}\cS_{u^{\gamma}}^{(1)}\big{)}^{j}$.

\begin{lem}\label{lem:ReducK2}
There exists operators $\mathrm{T}, \mathcal{J}^* :  \ell^2 (\mu_k^\gamma)\longrightarrow \ell^2 (\mu_1^\gamma)$ with operator norm smaller than one, such that for $j\ge k$ we have
\begin{align}\label{e:AreCharentaisesReallySoConvenient?}
\langle K'_{\gamma}, \mathcal{S}^j_{u^\gamma} K'_{\gamma} \rangle_{\ell^2(\mu_k^\gamma)} &= \left\langle K'_{\gamma}, \big{(} \mathcal{S}_{u^\gamma}^{(k)}\big{)}^{j-k+1} \mathcal{J} \mathrm{T} K'_{\gamma}  \right\rangle_{\ell^2(\mu_k^\gamma)} \\
&= \left\langle K'_{\gamma},  \mathcal{J} \big{(} \mathcal{S}^{(1)}_{u^\gamma}\big{)}^{j-k+1} \mathrm{T} K'_{\gamma}  \right\rangle_{\ell^2(\mu_k^\gamma)}= \left\langle \mathcal{J}^* K'_{\gamma}, \big{(} \mathcal{S}^{(1)}_{u^\gamma}\big{)}^{j-k+1} \mathrm{T} K'_{\gamma}  \right\rangle_{\ell^2(\mu_1^\gamma)}.
\end{align}

As a consequence,
\begin{equation*}
\begin{aligned}
&\lim_{\eta_0 \downarrow 0}\limsup_{N\to\infty} \mathrm{Var}^I_{\mathrm{nb},\eta_0}(K^{\gamma})^2 \\
&\le \lim_{n\to \infty} \lim_{\eta_0\downarrow 0}\limsup_{N\to\infty} \sup\limits_{\lambda\in I_1}  \frac{C_I}{N n^2} \Big|\sum_{j=1}^{n} (n-j)   \langle \mathcal{J}^* K'_{\gamma},\big{(}\cS_{u^{\gamma}}^{(1)}\big{)}^j \mathrm{T} K'_{\gamma}\rangle_{\ell^2(\mathrm{B}_1,\mu_1^{\gamma})}\Big|.
\end{aligned}
\end{equation*}
Furthermore, if $K_\gamma$ satisfies \emph{\Hol{}}, then $\mathrm{T}K_\gamma$ and $\mathcal{J}^* K_\gamma$ also satisfy \emph{\Hol{}}.
\end{lem}

\begin{proof}
Let us denote by $\mathcal{J} : \ell^2 (\mu_1^\gamma)\longrightarrow \ell^2 (\mu_k^\gamma)$ the map $\big{(} \mathcal{J} \phi \big{)} (b_1,\dots, b_k) = \phi(b_k)$. Using \eqref{eq:AlmostStoch1}, we see that $\|\mathcal{J}\|_{\ell^2 (\mu_1^\gamma)\rightarrow \ell^2 (\mu_k^\gamma)} \leq 1$, so that its adjoint $\mathcal{J}^* :  \ell^2 (\mu_k^\gamma)\longrightarrow \ell^2 (\mu_1^\gamma)$ satisfies the same bound. From Remark \ref{rem:HolStable}, if $K_\gamma$ satisfies \Hol{}, then $\mathcal{J}^*K_\gamma$ also satisfies \Hol{}.

Now, we have 
\[
\big{(}\mathcal{S}^{k-1}_{u^\gamma} K'_{\gamma}\big{)}(b_1 ; b_k)= \sum_{(b_{k+1} ; b_{2k-1})\in \mathrm{B}^{b_k}_{k-1}} \Lambda(b_{k} ; b_{2k-1}) K'_{\gamma}(b_k; b_{2k-1}),
\]
where 
\[
\Lambda(b_{k} ; b_{2k-1})= \prod_{\ell=k}^{2k-2}\frac{|\zeta^{\gamma}(b_{\ell})|^2}{\Im R_{\gamma}^+(o_{b_{\ell}})} u^{\gamma}(b_{\ell}) \Im R_{\gamma}^+(o_{b_{\ell+1 }}).
\]

In particular, $\big{(}\mathcal{S}^{k-1}_{u^\gamma} K'_{\gamma}\big{)}(b_1 ;b_k)$ depends only on $b_k$, so we may define $(\mathrm{T}K_\gamma')(b_k) := (\mathcal{S}^{k-1}_{u^\gamma}K'_{\gamma})(b_0 ; b_k)$. Then $\mathrm{T}K_\gamma'$ satisfies \Hol{}, and $\cJ \mathrm{T}K_\gamma' = \mathcal{S}^{k-1}_{u^\gamma}K'_{\gamma}$.

Note that $\mathcal{S}^{(k)}_{u^\gamma} \mathcal{J}= \mathcal{J} \mathcal{S}^{(1)}_{u^\gamma}$. This is just saying that, if a function on $\mathrm{B}_k$ depends only on the last variable, so will its image by $\mathcal{S}^{(k)}_{u^\gamma}$.

This proves \eqref{e:AreCharentaisesReallySoConvenient?}.

Let us now check that for any $K\in \ell^2(\mu_k^\gamma)$, we have $\|\mathrm{T} K \|_{\ell^2(\mu_1^\gamma)}\leq \|K\|_{\ell^2(\mu_k^\gamma)} $.
Noting that we have $\mu_1^\gamma(b_k) |\Lambda(b_{k} ;b_{2k-1})| = \mu_k^\gamma(b_{k} ;b_{2k-1})$ and, by \eqref{e:ASW}, that
\[
\sum_{(b_{k+1} ; b_{2k-1})\in \mathrm{B}^{b_k}_{k-1}} |\Lambda(b_{k} ;b_{2k-1})|\leq 1,
\]
we deduce that, for any $K\in \ell^2(\mu^\gamma_k)$, we have
\begin{align}\label{e:tbound}
\|\mathrm{T}K\|_{\ell^2(\mu_1^\gamma)}^2 &= \sum_{b_k\in B} \mu^\gamma_1(b_k)\Big{|}\sum_{(b_{k+1} ;b_{2k-1})\mathrm{B}^{b_k}_{k-1}} \Lambda(b_{k} ; b_{2k-1}) K(b_{k} ; b_{2k-1})\Big{|}^2\\
&\leq \sum_{b_k\in B} \mu^\gamma_1(b_k) \sum_{(b_{k+1};b_{2k-1})\mathrm{B}^{b_k}_{k-1}} |\Lambda(b_{k} ; b_{2k-1})| |K(b_{k} ;b_{2k-1})|^2\nonumber\\% ~~ \text{ by convexity of } x\mapsto x^2\\\nonumber
&= \sum_{(b_{k} ; b_{2k-1})\in \mathrm{B}_k}  \mu^\gamma_k(b_k ; b_{2k-1}) |K(b_{k} ; b_{2k-1})|^2= \|K\|_{\ell^2(\mu_k^{\gamma})}^2.\nonumber
\end{align}
This concludes the proof.
\end{proof}

\begin{rem}
Let $j\ge k$. Since $\cS_{u^{\gamma}}^jK =\cJ\cS_{u^{\gamma}}^{j-k+1}\mathrm{T}K$, we also deduce that 
\begin{equation}\label{eq:NormeS}
\|\mathcal{S}^j_{u^\gamma}\|_{\ell^2(\mu_k^\gamma) \rightarrow \ell^2(\mu_k^\gamma)} \leq \|\mathcal{S}^{j-k+1}_{u^\gamma}\|_{\ell^2(\mu_1^\gamma) \rightarrow \ell^2(\mu_1^\gamma)}\,.
\end{equation}
\end{rem}

In the sequel, we will work with 
\[
\nu_k^\gamma:= \frac{1}{\mu_k^\gamma(\mathrm{B}_k)}\mu_k^\gamma,
\]
which is a probability measure on $\mathrm{B}_k$.

Then the bounds $\|\mathrm{T}\|_{\ell^2(\mu_k^\gamma)\to\ell^2(\mu_1^\gamma)}\le 1$, $\|\cJ^*\|_{\ell^2(\mu_k^\gamma)\to\ell^2(\mu_1^\gamma)}\le 1$ become
\begin{equation}\label{eq:BornesTJ}
\begin{aligned}
\|\mathrm{T}\|_{\ell^2(\nu_k^\gamma)\to \ell^2(\nu_1^\gamma)} &\leq \sqrt{\frac{\mu_k^\gamma(\mathrm{B}_k)}{\mu_1^\gamma(\mathrm{B}_1)}}\le 1\\
\|\cJ^*\|_{\ell^2(\nu_k^\gamma)\to \ell^2(\nu_1^\gamma)} &\leq \sqrt{\frac{\mu_k^\gamma(\mathrm{B}_k)}{\mu_1^\gamma(\mathrm{B}_1)}}\le 1,
\end{aligned}
\end{equation}
where the $\le 1$ follows from \eqref{e:mukdecreases}.

To lighten the notation, we will write in the sequel
\[
\|f\|_\nu = \sqrt{\langle f, f \rangle_\nu} := \|f\|_{\ell^2(\nu_1^\gamma)},
\]
while we write $\|f\|= \langle f,f\rangle$ for the usual $\ell^2(B)$ norm with respect to the uniform scalar product.
If $H\subset \ell^2(B)$ is a linear subspace, we will denote by $P_H$ the orthogonal projection on $H$ with respect to the usual scalar product on $B$, and by $P_{H,\nu}$ the orthogonal projection with respect to the $\ell^2(\nu_1^\gamma)$ scalar product.

\section{Spectral gaps and contraction}\label{Sec:Contrac}
The aim of this section is to obtain bounds on $\|\mathcal{S}^{j}_{u^\gamma}\|_{\ell^2(\nu_1^\gamma) \rightarrow \ell^2(\nu_1^\gamma)}$. To this end, we introduce the operator
\[
(\cS_{\gamma}K)(b_1) = \frac{|\zeta^{\gamma}(b_1)|^2}{\Im R_{\gamma}^+(o_{b_1})} \sum_{b_{2}\in \cN_{b_1}^+} \Im R_{\gamma}^+(o_{b_{2}}) K(b_2) \,,
\]
which is sub-stochastic as mentioned above. Note that 
\begin{equation}
\mathcal{S}_{u^\gamma} = M_{u^\gamma} \mathcal{S}_\gamma,
\end{equation}
where $M_{u^\gamma}$ denotes the multiplication by $u^\gamma$.

 Using \eqref{e:zetainv}, one checks that the adjoint of $\mathcal{S}_\gamma$ in $\ell^2(B,\mu_1^{\gamma})$ is given by
\[
(\cS_{\gamma}^{\ast}K)(b_1) = \frac{|\zeta^{\gamma}(\hat{b}_1)|^2}{\Im R_{\gamma}^-(t_{b_1})} \sum_{b_0\in \cN_{b_1}^-} \Im R_{\gamma}^-(t_{b_0})K(b_0) \,,
\]
which is also sub-stochastic.

Since $\mathcal{S}_\gamma^*$ is sub-stochastic and $\nu$ is a probability measure, we have $\|\mathcal{S}_\gamma^*\|_{\ell^\infty(\nu)\rightarrow \ell^1(\nu)} \leq 1$, so that $\|\mathcal{S}_\gamma\|_{\ell^1(\nu)\rightarrow \ell^\infty(\nu)} \leq 1$. In particular, we have for any $p\in [1, +\infty]$
\begin{equation}\label{eq:TrivialContractionBound}
\|\mathcal{S}_{u^\gamma} \|_{\ell^p(\nu)\rightarrow \ell^p(\nu)}= \|\mathcal{S}_{\gamma} \|_{\ell^p(\nu)\rightarrow \ell^p(\nu)} \leq 1.
\end{equation}
Indeed, since $\nu$ is a probability measure, $\|\mathcal{S}_\gamma K\|_{\ell^p(\nu)} \leq \|\mathcal{S}_\gamma K\|_{\ell^{\infty}(\nu)} \leq \|K\|_{\ell^1(\nu)}\leq \|K\|_{\ell^p(\nu)}$.

We also let $\mathcal{S}:\ell^2(B)\To\ell^2(B)$ be the (non-backtracking) transfer operator, defined by
\[
(\mathcal{S}K)(b_1) = \frac{1}{q(t_{b_1})} \sum_{b_{2}\in \cN_{b_1}^+} K(b_2),
\]
where $q(v)=\deg(v)-1$.

Informally speaking, the operator $\mathcal{S}_{u^\gamma}$ can be understood as follows. Start with $\mathcal{S}$  acting on $\ell^2(B)$, and multiply it by positive weights to turn it into $\mathcal{S}_\gamma$ acting on $\ell^2(B, \nu_1^\gamma)$; then, add some phases $u^\gamma$, to turn it into  $\mathcal{S}_{u^\gamma}$. We will therefore start by recalling the contracting properties of $\mathcal{S}$; we will then understand the effect of adding some positive weights; finally, we will understand how the phases affect the contraction properties.

\subsection{Contraction properties of the transfer operator}

 Let us denote by $F$ the set of functions on $B$ which depend only on the origin:
\begin{equation*}
F:= \{ f : B\rightarrow \C \mid \forall b,b'\in B, o_b=o_{b'} \Longrightarrow f(b)= f(b') \}.
\end{equation*}
We denote by $F^\perp$ the orthogonal complement of $F$ for the usual scalar product on $\ell^2(B)$.

Note that for the usual scalar product on $\ell^2(B)$, the orthogonal projector on $F$ can be written
\begin{equation*}
(P_F K)(b) = \frac{1}{d(o_b)} \sum_{b'\in B ; o_{b'} = o_b} K(b').
\end{equation*}

If $f\in F$, then $\|\mathcal{S}f\| = \|f\|$, so that $\mathcal{S}$ enjoys no contraction properties on $F$. However, it is contracting on the orthogonal 
complement of $F$:
\begin{lem}\label{lem:EasyGapTransfer}
For all $f\in F^\perp$, we have 
\[
\|\mathcal{S}f\| \leq \frac{1}{2} \|f\|.
\]
\end{lem}
\begin{proof}
Since $P_F f =0$, we have $\mathcal{S}f(b) = \frac{-1}{q(t_b)} f(\hat{b})$, so  $\mathcal{S^*}\mathcal{S}f(b) =  \frac{1}{q(o_b)}\sum_{b_-}\frac{-1}{q(t_{b_-})}f(\widehat{b_-})=\frac{-1}{q(o_b)^2}\cB f(\hat{b})=\frac{1}{q(o_b)^2} f(b)$. By assumption \Data{}, $q(o_b)\ge 2$ $\forall b$, so $\|\mathcal{S}f\|^2\leq \frac{1}{4}\|f\|^2$.
\end{proof}

We will also need the much deeper fact that $\mathcal{S}^2$ is contracting on the space  $\mathbf{1}^\perp$, the orthogonal complement in $\ell^2(B)$ of constant functions (which is a larger space than $F^\perp$). Namely, Theorem 1.1 in \cite{A2} says that there exists $0<c(D,\beta)<1$ depending only on the bound $D$ on the degrees of the vertices of the graph, and on the spectral gap $\beta$ appearing in Hypothesis \EXP{} such that
\begin{equation}\label{eq:GapTransfer}
\forall f\in \mathbf{1}^\perp, \qquad\|\mathcal{S}^2 f\| \leq \big{(}1-c(D,\beta)\big{)} \|f\|.
\end{equation}

\subsection{Contraction properties of \texorpdfstring{$\mathcal{S}_\gamma$}{S}}
The aim of this section is to prove analogues of Lemma \ref{lem:EasyGapTransfer} and the bound \eqref{eq:GapTransfer} for $\mathcal{S}_\gamma$ acting on $\ell^2(B, \nu_1^\gamma)$. The main difficulties come, on one hand, from the fact that $\mu_1^\gamma$ is not a priori bounded from above, and could have some peaks; on the other hand, from the fact that the weights in $\mathcal{S}_\gamma$ could tend to disconnect the graph. We will therefore call ``bad'' an oriented edge, or a pair of oriented edges, where one of these events happen: for any $M>0$, we set
\begin{align*}
\Bad(M)&:= \Big{\{} b\in B\mid \nu_1^\gamma(b) > \frac{M}{N} \Big{\}}\\
\Badp(M) &:= \Big{\{} (b,b')\in B^2\mid 0<\nu_1^\gamma(b) \big{(}\mathcal{S}^*_\gamma \mathcal{S}_\gamma \big{)}(b,b') < \frac{1}{MN} \Big{\}}.
\end{align*}
Note that, if $(b,b')\in \Badp(M)$, we must have $o_b=o_{b'}$.

To formulate the analogue of Lemma \ref{lem:EasyGapTransfer}, we denote by $P_{F,\nu}$ the orthogonal projection on $F$ in the space $\ell^2(B, \nu_1^\gamma)$, and by $P_{F^\perp,\nu}:= \mathrm{1}-P_{F,\nu}$ the orthogonal projection onto $F^\perp$. The following proposition says that the bad (pairs of) edges are the only obstruction for contraction properties of $\mathcal{S}_\gamma P_{F^\perp, \nu}$. 

\begin{prp}\label{prop:ContractOrigin}
For any $M>0$, and any $K\in \mathscr{H}_1$ (possibly depending on $\gamma$), we have
\[
\| \mathcal{S}_\gamma P_{F^\perp, \nu} K \|_\nu^2 \leq \Big{(} 1- \frac{3}{4} M^{-2} \Big{)} \| P_{F^\perp, \nu} K \|_\nu^2 + C_{N,M}(K) ,
\]
and
\begin{equation}\label{eq:BoundProjFPerp}
\|P_{F^\perp,\nu} K \|_\nu^2 \leq \frac{4}{3} M^2 \big{(} \|K\|^2_\nu - \| \mathcal{S}_\gamma K\|^2_\nu + C_{N,M}(K)\big{)},
\end{equation}
where
\begin{align*}
C_{N,M}(K) &=  \frac{1}{2NM} \sum_{(b,b')\in \Badp(M)} \big{(}\mathcal{S}^*\mathcal{S}\big{)}(b,b') \big{|} K (b)- K(b')\big{|}^2 \\
&\quad + \frac{1}{M^2} \sum_{b\in \Bad(M)} \nu_1^\gamma(b) \big{|}K (b)- P_{F} K(b) \big{|}^2.
\end{align*}
\end{prp}

To formulate the analogue of \eqref{eq:GapTransfer}, we introduce another set of bad pairs of  oriented edges
\begin{align*}
\Badp_2(M) &:= \Big{\{} (b,b')\in B^2\mid  0<\nu_1^\gamma(b) \big{(} (\mathcal{S}_\gamma^*)^2 \mathcal{S}_\gamma^2 \big{)}(b,b') < \frac{1}{MN} \Big{\}}.
\end{align*}

Let $P_{\mathbf{1}^\perp,\nu}$ denote the orthogonal projection onto $\mathbf{1}^\perp$ with respect to $\ell^2(\nu_1^\gamma)$. We have

\begin{prp}\label{prop:ContractConstant}
For any $M>0$, and any $K\in \mathscr{H}_1$ (possibly depending on $\gamma$), we have
\[
\| \mathcal{S}_\gamma^2 P_{\mathbf{1}^\perp, \nu}  K \|_\nu^2 \leq \left( 1- c(D,\beta) M^{-2} \right) \| P_{\mathbf{1}^\perp, \nu} K \|_\nu^2 + C_{N,M}'(K),
\]
and 
\begin{equation}\label{eq:BoundProjConstPerp}
\| P_{\mathbf{1}^\perp,\nu}K\|^2_\nu \leq \frac{M^2}{c(D,\beta)} \big{(} \|K\|_\nu^2 - \|\mathcal{S}_\gamma^2 K \|_\nu^2 + C'_{N,M}(K
) \big{)},
\end{equation}

where
\begin{align*}
C_{N,M}'(K) &=  \frac{1}{2NM} \sum_{(b,b')\in \Badp_2(M)} \big{(}(\mathcal{S}^*)^2 \mathcal{S}^2 \big{)}(b,b') \big{|}K (b)-K(b')\big{|}^2 \\
&\quad + \frac{1}{M^2} \sum_{b\in \Bad(M)} \nu_1^\gamma(b) \big{|} K (b)- P_{\mathbf{1}} K (b) \big{|}^2.
\end{align*}
\end{prp}

Let us start by proving Proposition \ref{prop:ContractOrigin}
\begin{proof}[Proof of Proposition \ref{prop:ContractOrigin}]
Let us write $\mathcal{Q}^\gamma:= \mathcal{S}^*_\gamma \mathcal{S}_\gamma $, where the adjoint is taken for the scalar product of $\ell^2(\nu_1^\gamma)$. The operator $\mathcal{Q}^\gamma$ is self-adjoint for this scalar product, so that for all $b,b'\in B$, we have
\begin{equation}\label{eq:Qselfadj}
\nu_1^\gamma(b) \mathcal{Q}^\gamma(b,b') = \nu_1^\gamma(b')\mathcal{Q}^\gamma(b',b).
\end{equation}

We define $D^\gamma(b):= \sum_{b'\in B} \mathcal{Q}^\gamma(b,b')\leq 1$, and $\mathcal{M}^\gamma(b,b'):= D^\gamma(b) \delta_{b=b'} - \mathcal{Q}^\gamma(b,b')$. Then \eqref{eq:Qselfadj} implies the \emph{Dirichlet identity}
\begin{equation}\label{eq:DIrichletId}
\frac{1}{2} \sum_{b,b'\in B} \nu_1^\gamma(b) \mathcal{Q}^\gamma(b,b') |K(b)- K(b')|^2 = \langle K, \mathcal{M}^\gamma K \rangle_\nu,
\end{equation}
and since $\mathcal{Q}^\gamma(b,b')\neq 0 \Longleftrightarrow \left(o_b=o_{b'}\right)$, this shows in particular that 
\begin{equation}\label{eq:CriterionK}
 K\in F\iff \langle K, \mathcal{M}^\gamma K \rangle_\nu=0 \iff \mathcal{M}^\gamma K = 0,
\end{equation}
where the last equivalence comes from the fact that $\mathcal{M}^\gamma$ is a self-adjoint non-negative operator. In particular we have 
\begin{equation}\label{eq:CriterionK2}
\langle K, \mathcal{M}^\gamma K \rangle_\nu = \langle P_{F^\perp, \nu} K, \mathcal{M}^\gamma  P_{F^\perp, \nu} K \rangle_\nu\,.
\end{equation}

Thanks to \eqref{eq:DIrichletId} and \eqref{eq:CriterionK2}, we have
\begin{equation}\label{eq:NoName}
\begin{aligned}
 \| P_{F^\perp, \nu} K \|_\nu^2 - \| \mathcal{S}_\gamma P_{F^\perp, \nu} K \|_\nu^2 &=  \langle  P_{F^\perp, \nu} K ,  (\mathrm{Id} - \mathcal{Q}^\gamma)  P_{F^\perp, \nu} K  \rangle_\nu  \\
  &\geq  \langle  P_{F^\perp, \nu} K , \mathcal{M}^\gamma  P_{F^\perp, \nu} K  \rangle_\nu \\
& = \frac{1}{2} \sum_{b,b'\in B} \nu_1^\gamma(b) \mathcal{Q}^\gamma(b,b') |K(b)- K(b')|^2\\
& \geq \frac{1}{2MN} \sum_{b,b'\in B\setminus \Badp(M)}  \big{(}\mathcal{S}^*\mathcal{S}\big{)}(b,b') |K(b)- K(b')|^2,
\end{aligned}
\end{equation}
since $\big{(}\mathcal{S}^*\mathcal{S}\big{)}(b,b')\leq 1$.
As in \eqref{eq:DIrichletId}, the last member of the inequality is then equal to
\begin{align*}
\frac{1}{MN} \langle K, (I - \mathcal{S}^*\mathcal{S}) K \rangle -  \frac{1}{2MN} \sum_{b,b'\in \Badp(M)}  \big{(}\mathcal{S}^*\mathcal{S}\big{)}(b,b') |K(b)- K(b')|^2 \,.
\end{align*}

By Lemma \ref{lem:EasyGapTransfer}, and using the analog of \eqref{eq:CriterionK2}, we have 
\begin{equation}\label{eq:NoName2}
 \langle K, (I - \mathcal{S}^*\mathcal{S}) K \rangle =   \|P_{F^\perp} K\|^2 -  \|\mathcal{S} P_{F^\perp} K\|^2 \geq \frac{3}{4} \|P_{F^\perp} K\|^2.
\end{equation}

Now, by definition of $\Bad(M)$, we have
\begin{equation}\label{eq:NoName3}
\begin{aligned}
\frac{1}{N}\|P_{F^\perp} K\|^2 &\geq \frac{1}{M} \sum_{b\in B\backslash \Bad(M)} \nu_1^\gamma(b) \big{|}\big{(}P_{F^\perp} K \big{)}(b)\big{|}^2\\
&= \frac{1}{M}\|P_{F^\perp} K\|^2_\nu - \frac{1}{M} \sum_{b\in \Bad(M)} \nu_1^\gamma(b) \big{|}\big{(}P_{F^\perp} K \big{)}(b)\big{|}^2.
\end{aligned}
\end{equation}

Finally, $\|P_{F^\perp,\nu} K\|_\nu = \|P_{F^\perp,\nu} K - P_{F^\perp,\nu} P_{F} K \|_\nu =  \|P_{F^\perp,\nu} P_{F^\perp} K\|_\nu  \leq \|P_{F^\perp} K\|_\nu$. Combining this and \eqref{eq:NoName}, \eqref{eq:NoName2} and \eqref{eq:NoName3}, we obtain the first part of the statement.

Note that, in all the computations following \eqref{eq:NoName}, we have shown that 
\[
\langle   K, \mathcal{M}^\gamma K \rangle_\nu \geq \frac{3}{4 M^2} \|P_{F^\perp,\nu} K\|^2_\nu - C_{N,M}(K).
\]
But since $\langle   K , \mathcal{M}^\gamma K \rangle_\nu = \langle   K , (D^{\gamma}-\cQ^{\gamma}) K \rangle_\nu\leq \|K\|_\nu^2 - \|\mathcal{S}_\gamma K\|^2_\nu$, equation \eqref{eq:BoundProjFPerp} follows.
\end{proof}

The proof of Proposition \ref{prop:ContractConstant} is exactly the same as the previous proof, using $\mathcal{Q}_2^\gamma:= (\mathcal{S}_\gamma^*)^2 \mathcal{S}_\gamma^2$ instead of $\mathcal{Q}^\gamma$, and using the bound \eqref{eq:GapTransfer} instead of Lemma \ref{lem:EasyGapTransfer}.

Let us now estimate the quantities $C_{N,M}$ and $C'_{N,M}$.

\subsection{Estimates on the bad terms}
\begin{prp}\label{prop:BoundingCNM}
For all $t\in \N$, and all $K\in \mathscr{H}_1$ (possibly depending on $\gamma$), we have
\begin{multline}\label{eq:FirstBoundBad}
C_{N,M}(\mathcal{S}_{u^\gamma}^t K)\leq \frac{2}{NM} \left( \# \Badp(M) \right)^{1/4} \bigg(\sum_{b\in B} \frac{1}{\nu_1^\gamma(b)^2}\bigg)^{1/4} \| K\|_{\ell^4(\nu_1^\gamma)}^2 \\
+ \frac{2}{M^2} \left(\nu_1^\gamma (\Bad(M))\right)^{1/2} \bigg[ \|K\|_{\ell^4(\nu_1^\gamma)}^2  + \bigg(\sum_{b\in B} \frac{|\left( P_F \nu_1^\gamma\right)(b)|^2}{\nu_1^\gamma(b)}\bigg)^{1/4}  \|K\|_{\ell^8(\nu_1^\gamma)}^2 \bigg]
\end{multline}
where $(P_F\nu)(b)=\sum_{b'}P_F(b,b')\nu(b') = \frac{1}{d(o_b)}\sum_{o_{b'}=o_b}\nu(b')$.
Similarly,
\begin{multline}\label{eq:SecondBoundBad}
C'_{N,M}(\mathcal{S}_{u^\gamma}^t K)\leq \frac{2}{NM} \left( \# \Badp_2(M) \right)^{1/4} \bigg(\sum_{b\in B} \frac{1}{\nu_1^\gamma(b)^2}\bigg)^{1/4} \| K\|_{\ell^4(\nu_1^\gamma)}^2\\
+ \frac{2}{M^2} \left(\nu_1^\gamma (\Bad(M))\right)^{1/2} \bigg[ \|K\|_{\ell^4(\nu_1^\gamma)}^2  +  \bigg( \frac{1}{N^2} \sum_{b\in B} \frac{1}{\nu_1^\gamma(b)}\bigg)^{1/4}  \|K\|_{\ell^8(\nu_1^\gamma)}^2 \bigg].
\end{multline}
\end{prp}
\begin{proof}
Recall that $C_{N,M}(\cdot)$ is made of two terms: we will start by estimating 
\begin{multline*}
\frac{1}{2NM} \sum_{(b,b')\in \Badp(M)} \big{(}\mathcal{S}^*\mathcal{S}\big{)}(b,b') \big{|} \mathcal{S}_{u^\gamma}^t K(b)-\mathcal{S}_{u^\gamma}^t  K(b')\big{|}^2\\
\leq \frac{2}{NM} \sum_{(b,b')\in \Badp(M)} \big{(}\mathcal{S}^*\mathcal{S}\big{)}(b,b') \big{|} \mathcal{S}_{u^\gamma}^t K(b)\big{|}^2= \frac{2}{NM} \sum_{b\in B} n(b) \big{|} \mathcal{S}_{u^\gamma}^t K(b)\big{|}^2,
\end{multline*}
where
\[
n(b):= \sum_{\substack{b'\in B\\(b,b')\in \Badp(M)}} \left(\mathcal{S}^*\mathcal{S}\right)(b,b').
\]

Using the Cauchy-Schwarz inequality twice, we have

\begin{align*}
 \sum_{b\in B} n(b) \big{|}\mathcal{S}_{u^\gamma}^t  K(b)\big{|}^2 &=  \sum_{b\in B} \frac{n(b)}{\sqrt{\nu_1^\gamma(b)}} \sqrt{\nu_1^\gamma(b)} \big{|} \mathcal{S}_{u^\gamma}^t K(b)\big{|}^2 \\
 &\leq \bigg(\sum_{b\in B} n(b)^4\bigg)^{1/4} \bigg(\sum_{b\in B} \frac{1}{\nu_1^\gamma(b)^2}\bigg)^{1/4}  \bigg(\sum_{b\in B} \nu_1^\gamma(b) | \mathcal{S}_{u^\gamma}^t K(b)|^4\bigg)^{1/2}.
\end{align*}

Now, by H\"older's inequality,% we have
\begin{align}\label{e:holi}
\sum_{b\in B} n(b)^4 &= \sum_{b\in B} \bigg( \sum_{b'\in B} \mathbf{1}_{(b,b')\in \Badp(M)} \big(\mathcal{S}^*\mathcal{S}\big{)}(b,b')\bigg)^4\nonumber\\
&\leq  \sum_{b\in B} \bigg( \sum_{b'\in B} \mathbf{1}_{(b,b')\in \Badp(M)}\bigg) \bigg(\sum_{b'\in B} \big{(}(\mathcal{S}^*\mathcal{S})(b,b')\big{)}^{4/3}\bigg)^3.
\end{align}
But $\sum_{b'\in B} \big{(}(\mathcal{S}^*\mathcal{S})(b,b')\big{)}^{4/3}\leq 1$, since $\mathcal{S}$ and $\mathcal{S}^*$ are stochastic. Therefore, we have
\[
\sum_{b\in B} n(b)^4\leq \# \Badp(M).
\]

On the other hand, we have
\[
\bigg(\sum_{b\in B} \nu_1^\gamma(b) | \mathcal{S}_{u^\gamma}^t K(b)|^4\bigg)^{1/2} = \| \mathcal{S}_{u^\gamma}^t K\|_{\ell^4(\nu_1^\gamma)}^2 \leq  \| K\|_{\ell^4(\nu_1^\gamma)}^2,
\]
by \eqref{eq:TrivialContractionBound}. Therefore, the first term making up $C_{N,M}(\mathcal{S}_{u^\gamma}^t K)$ is bounded by the first term in the right-hand side of \eqref{eq:FirstBoundBad}.

Let us now consider the second term in $C_{N,M}(\mathcal{S}_{u^\gamma}^t K)$, namely 
\begin{multline*}
\frac{1}{M^2} \sum_{b\in \Bad(M)} \nu_1^\gamma(b) \big{|}  \mathcal{S}_{u^\gamma}^t  K (b)- P_{F}  \mathcal{S}_{u^\gamma}^t K (b) \big{|}^2\\
\leq \frac{2}{M^2} \sum_{b\in \Bad(M)} \nu_1^\gamma(b) \big( | \mathcal{S}_{u^\gamma}^t  K (b)|^2+ | P_{F}  \mathcal{S}_{u^\gamma}^t K (b) |^2\big).
\end{multline*}
We have
\begin{align*}
\sum_{b\in \Bad(M)} \nu_1^\gamma(b)  \left| \mathcal{S}_{u^\gamma}^t  K(b)\right|^2 & \le \left( \nu_1^\gamma (\Bad(M))\right)^{1/2} \bigg(\sum_{b\in B} \nu_1^\gamma(b) \left| \mathcal{S}_{u^\gamma}^t  K (b)\right|^4\bigg)^{1/2}\\
&\leq \left(\nu_1^\gamma (\Bad(M)\right)^{1/2} \|K\|_{\ell^4(\nu_1^\gamma)}^2
\end{align*}
by \eqref{eq:TrivialContractionBound} for $p=4$. Similarly,
\begin{align*}
\sum_{b\in \Bad(M)} \nu_1^\gamma(b)  \left| P_F \mathcal{S}_{u^\gamma}^t  K (b)\right|^2 
&\leq \left( \nu_1^\gamma \Bad(M) \right)^{1/2} \bigg(\sum_{b\in B} \nu_1^\gamma(b) \left| P_F \mathcal{S}_{u^\gamma}^t  K (b)\right|^4\bigg)^{1/2}.
\end{align*}
Using H\"older's inequality and the fact that $P_F$ is stochastic as in \eqref{e:holi}, we have
\begin{align*}
\sum_{b\in B} \nu_1^\gamma(b) \left| P_F \mathcal{S}_{u^\gamma}^t  K (b)\right|^4 &\leq \sum_{b\in B} \nu_1^\gamma(b) \sum_{b'\in B} P_F(b,b') \left| \mathcal{S}_{u^\gamma}^t  K (b')\right|^4\\
&\leq \bigg(\sum_{b'\in B} \frac{|\left( P_F \nu_1^\gamma\right)(b')|^2}{\nu_1^\gamma(b')}\bigg)^{1/2} \bigg(\sum_{b\in B} \nu_1^\gamma(b) \left|  \mathcal{S}_{u^\gamma}^t  K(b)\right|^8\bigg)^{1/2}\\
&\leq  \bigg(\sum_{b'\in B} \frac{|\left( P_F \nu_1^\gamma\right)(b')|^2}{\nu_1^\gamma(b')}\bigg)^{1/2}  \|K\|_{\ell^8(\nu_1^\gamma)}^4.
\end{align*}

The very same proof gives \eqref{eq:SecondBoundBad}, by noting that $P_{\mathbf{1}}\nu_1^\gamma = \frac{1}{|B_N|}\mathbf{1} \le \frac{1}{N} \mathbf{1}$.
\end{proof}

\begin{cor}\label{cor:BadIsntTooBad}
Let $\cQ_N$ satisfy Hypotheses \emph{\BST{}}, \emph{\Data{}} and \emph{\Green{}}.
For any $s>0$ and any $M>1$, we have, for any $K_\gamma\in \mathscr{H}_1$ satisfying \eqref{eq:APrioriBoundOperators}
\[
\begin{aligned}
\sup\limits_{t\in \N} C_{N,M}\left(\mathcal{S}^t_{u^\gamma} K_\gamma\right) &= O_{N\to + \infty, \gamma}^{(s)}(1)M^{-s}\\
\sup\limits_{t\in \N} C'_{N,M}\left(\mathcal{S}^t_{u^\gamma} K_\gamma\right) &= O_{N\to + \infty, \gamma}^{(s)}(1)M^{-s} ,
\end{aligned}
\]
where the $O$'s above depend on $s$, but not on $M$.
\end{cor}
\begin{proof}
We have to bound the different quantities appearing in the right-hand side of \eqref{eq:FirstBoundBad} and \eqref{eq:SecondBoundBad}.

Thanks to \eqref{eq:APrioriBoundOperators} and to Remark \ref{rem:MeasureNuIsLessScaryThanCoronavirus}, we know that
\[
\|K_\gamma\|_{\ell^\alpha(\nu_1^\gamma)} = \bigg(\frac{N}{\sum_b\mu_1^{\gamma}(b)}\bigg)^{1/\alpha}\bigg(\frac{1}{N}\sum_b\mu_1^{\gamma}(b)|K_{\gamma}(b)|^\alpha\bigg)^{1/\alpha} = O_{N\to + \infty, \gamma}(1).
\]

Next, by Remark~\ref{rem:MeasureNuIsLessScaryThanCoronavirus},
\[
\frac{1}{N^{\alpha+1}}\sum_{b\in B} \frac{1}{(\nu_1^\gamma(b))^\alpha} = \bigg(\frac{\sum_{b\in B} \mu_1^\gamma(b)}{N}\bigg)^\alpha \times \frac{1}{N} \sum_{b\in B} \frac{1}{(\mu_1^\gamma(b))^\alpha}=O_{N\to +\infty, \gamma}(1).
\]
 
Similarly,
\begin{align*}
\sum_{b\in B} \frac{|\left( P_F \nu_1^\gamma\right)(b)|^2}{\nu_1^\gamma(b)}&= \sum_{b\in B} \frac{1}{\nu_1^\gamma(b)} \frac{1}{d^2(o_b)} \bigg(\sum_{o_{b'}=o_b} \nu_1^\gamma(b')\bigg)^2\\
&\leq \sum_{b\in B} \frac{1}{\nu_1^\gamma(b)} \sum_{o_{b'}=o_b} \nu_1^\gamma(b')^2\leq D \bigg(\sum_{b\in B} \frac{1}{\nu_1^\gamma(b)^2}\bigg)^{1/2} \bigg(\sum_{b\in B} \nu_1^\gamma(b)^4 \bigg)^{1/2}\\
&= D \frac{N}{\sum_{b\in B} \mu_1^\gamma(b)} \bigg(\frac{1}{N}\sum_{b\in B} \frac{1}{\mu_1^\gamma(b)^2}\bigg)^{1/2} \bigg(\frac{1}{N}\sum_{b\in B} \mu_1^\gamma(b)^4 \bigg)^{1/2}.
\end{align*}

By Remark \ref{rem:MeasureNuIsLessScaryThanCoronavirus}, each factor is a $O_{N\to +\infty, \gamma}(1)$.

Recalling \eqref{eq:FirstBoundBad}, we have obtained that 
\[
C_{N,M}(\mathcal{S}_{u^{\gamma}}^t K_\gamma)\leq \bigg[\frac{1}{M}\bigg(\frac{\# \Badp(M)}{N} \bigg)^{1/4}+  \frac{1}{M^2} \left(\nu_1^\gamma (\Bad(M))\right)^{1/2}\bigg] O_{N\to +\infty, \gamma}(1)
\]
and we may get a similar bound for $C'_{N,M}(\mathcal{S}_{u^{\gamma}}^t K_\gamma)$. Therefore, the result follows from the following lemma.
\end{proof}

\begin{lem}\label{lem:moreerrorz}
Let $\cQ_N$ be a sequence of quantum graphs satisfying Hypotheses \emph{\BST{}}, \emph{\Data{}} and \emph{\Green{}}.
For any $s>1$, we have
\begin{align*}
\nu_1^\gamma (\Bad(M))&= O_{N\to + \infty, \gamma}^{(s)}(1)M^{-s}\\
\frac{\# \Badp(M)}{N} = O_{N\to + \infty, \gamma}^{(s)}(1) M^{-s},& \qquad \frac{\# \Badp_2(M)}{N} = O_{N\to + \infty, \gamma}^{(s)}(1) M^{-s}.
\end{align*}
\end{lem}
\begin{proof}
We have $\nu_1^\gamma (\Bad(M)) =\nu_1^\gamma \big( \big{\{} b\in B\mid \nu_1^\gamma(b) > \frac{M}{N} \big{\}} \big)$, so by Markov's inequality,
\[
\nu_1^\gamma (\Bad(M)) \leq \frac{N^s}{M^s} \sum_{b\in B} \nu_1^\gamma(b)^{s+1} = M^{-s} \bigg(\frac{N}{\mu_1^\gamma(B)}\bigg)^{s+1} \frac{1}{N}  \sum_{b\in B} \mu_1^\gamma(b)^{s+1}.
\]

By Remark \ref{rem:MeasureNuIsLessScaryThanCoronavirus}, this quantity is a $O_{N\to + \infty, \gamma}(1) M^{-s}$.

For the second estimate, recall that
\[
\Badp(M)\subset \Big{\{} (b,b')\in B^2,o_{b'}=o_b\mid \left((\mathcal{S}_\gamma^*\mathcal{S}_\gamma)(b,b') \nu_1^\gamma(b)\right)^{-1} > MN \Big{\}}.
\]

Therefore, by Markov inequality, we have
\begin{align*}
\# \Badp(M) &\leq \frac{1}{(MN)^s} \sum_{\substack{b,b'\in B\\o_b=o_{b'}}} \left[(\mathcal{S}_\gamma^*\mathcal{S}_\gamma)(b,b') \nu_1^\gamma(b)\right]^{-s} \\
&= \frac{1}{M^s} \bigg(\frac{\mu_1^\gamma(B)}{N}\bigg)^s \sum_{b\in B}\sum_{b',o_{b'}=o_b} \left[(\mathcal{S}_\gamma^*\mathcal{S}_\gamma)(b,b') \mu_1^\gamma(b)\right]^{-s}.
\end{align*}

Using Remark \ref{rem:HolStable}, we see that $F^\gamma(b)=\sum_{o_{b'}=o_b}[(\mathcal{S}_\gamma^*\mathcal{S}_\gamma)(b,b') \mu_1^\gamma(b)]^{-s}$ satisfies \Hol{}. Combining this fact with Remark \ref{rem:MeasureNuIsLessScaryThanCoronavirus}, we deduce that $\frac{\# \Badp(M)}{N} = O_{N\to + \infty, \gamma}^{(s)}(1) M^{-s}$.

For the last estimate we note that if $(\cS_{\gamma}^{\ast\,2}\cS_{\gamma}^2)(b,b')>0$, then there exist $b'', b_1, b'_1$ such that $b''\rightsquigarrow b_1 \rightsquigarrow b$ and $b''\rightsquigarrow b_1'\rightsquigarrow b'$, where $b_1\rightsquigarrow b_2$ means $t_{b_1}=o_{b_2}$. So now $F^\gamma(b)$ sums over $t_{b'}$ having the same grandparent as $t_b$ instead of same parent; we conclude as before.
\end{proof}

\begin{rem}\label{rem:WhenWeHaveGoodBounds}
If $K_\gamma\in \mathscr{H}_1$ satisfies $|K_\gamma(b)|\leq 1$ for all $N\in \N$, all $b\in B_N$ and all $\gamma\in I\times (0,1)$, then $\|K_\gamma\|_{\ell^p(\nu_1^\gamma)}\le \|K_\gamma\|_{\infty}\nu_1^\gamma(B)^{1/p}\le 1$, so
\begin{align*}
C'_{N,M}(K_\gamma) \leq C_{N,M}^0 &:= \frac{2}{NM} \left( \# \Badp_2(M) \right)^{1/4} \bigg(\sum_{b\in B} \frac{1}{\nu_1^\gamma(b)^2}\bigg)^{1/4} \\
&\quad + \frac{2}{M^2} \left(\nu_1^\gamma (\Bad(M))\right)^{1/2} \bigg[1  + \bigg(\frac{1}{N^2}\sum_{b\in B} \frac{1}{\nu_1^\gamma(b)}\bigg)^{1/4}  \bigg].
\end{align*}
The proof of Corollary \ref{cor:BadIsntTooBad} tells us that $C_{N,M}^0 =  O_{N\to + \infty, \gamma}^{(s)}(1) M^{-s}$.
\end{rem}

\subsection{Contraction properties of \texorpdfstring{$\mathcal{S}_{u^\gamma}$}{Su}}

In the previous subsections we showed that $\cS_{\gamma},\cS_{\gamma}^2$ are contractions on proper subspaces. The following proposition says that $\mathcal{S}_{u^\gamma}^4$ is contracting on the full space, unless the phases $u^\gamma$ satisfy very special relations. This improvement is reminiscent of Wielandt's theorem (see for instance \cite[Chapter 8]{Meyer}), saying that adding phases to a matrix with positive entries will strictly diminish its spectral radius unless the phases satisfy very special conditions.

Wielandt's theorem is insufficient for us as we need more precise information, namely a contraction on the norm of $\cS_{u^{\gamma}}^4$ instead of its spectral radius, which should moreover be uniform as the graph becomes large and $\gamma$ approaches the real axis. This will require a careful analysis of the various operators.

Let us write 
\begin{equation*}
\begin{aligned}
\omega^\gamma(b)& := -1+ \frac{|\zeta^\gamma(b)|^2}{\Im R_\gamma^+(o_b)} \sum_{b'\in \mathcal{N}_b^+} \Im R_\gamma^+(o_{b'})\\
&\ = \frac{-\Im \gamma}{\Im R_{\gamma}^+(o_b)} \int_{0}^{L_b} |\xi_+^{\gamma}(x_b)|^2\,\dd x_b,
\end{aligned}
\end{equation*}
where the second equality comes from \eqref{e:cur1}.

From the definition of $\xi_+^\gamma(x_b)$ and from Remark \ref{rem:HolStable}, one sees that $b\mapsto \int_0^{L_b} |\xi_+^\gamma(x_b)|^2 \mathrm{d}x_b$ satisfies \Hol{}, so that
\begin{equation}\label{eq:EstimateOnXi}
\|\omega^{\gamma}\|_\nu = O_{N\to +\infty, \gamma}(\Im \gamma).
\end{equation}

\begin{prp}\label{prop:ContractPhases}
Let $\gamma \in \C^+$, $M>1$, $\varepsilon\in (0,1/2)$, and let $\cQ$ be a quantum graph in the family $(\cQ_N)$. Then one of the following two options holds:
\begin{enumerate}[\rm (i)]
\item Either we have, for any $K_\gamma\in \mathscr{H}_1$ satisfying \eqref{eq:APrioriBoundOperators},
\begin{equation}\label{eq:CaseContract}
\| \mathcal{S}_{u^\gamma}^4 K_\gamma \|^2_\nu \leq (1-\varepsilon)^2 \|K_\gamma\|_\nu^2 + \tilde{C}_{N,M}(K_\gamma),
\end{equation}
where $\tilde{C}_{N,M}(K_\gamma)= \max_{j=0,1,2} \big\{ C_{N,M}(\cS_{u^\gamma}^jK_\gamma), C_{N,M}'(\cS_{u^\gamma}^jK_\gamma) \big\}$;
\item or there exists $\theta : V\longrightarrow \R$ and constants $c_0 \in \C$ with $|c_0|\leq 1$, and $\mathcal{C}(\beta,D)$ depending only on the spectral gap $\beta$ and the maximal degree $D$ of the graph such that
\begin{equation*}
\Big{\|} u^\gamma(b) - c_0 \ee^{\ii\, [\theta(o_b)-\theta(t_b)]} \Big{\|}_{\nu}^2\leq \mathcal{C}(\beta, D) M^3 \Big{[} \varepsilon^{1/2} + \|\omega^\gamma \|_{\nu} + \|\omega^\gamma \|_{\nu}^2 \Big{]} + \mathcal{C}(\beta, D) M^2 C^0_{N,M},
\end{equation*}
with $C_{N,M}^0$ as in Remark \ref{rem:WhenWeHaveGoodBounds}.
\end{enumerate}
\end{prp}

In the sequel, we will write
\begin{equation}\label{eq:defC0}
\cC_0= \cC_0(M, \varepsilon, \gamma)=\Big{(} \mathcal{C}(\beta, D) M^3 \Big{[} \varepsilon^{1/2} + \|\omega^\gamma \|_{\nu} + \|\omega^\gamma \|_{\nu}^2 \Big{]} + \mathcal{C}(\beta, D) M^2 C^0_{N,M}\Big{)}^{1/2}.
\end{equation}

If we are in case (i) of the previous proposition, we may iterate \eqref{eq:CaseContract} to obtain the following bound

\begin{cor}\label{cor:ContractPhases2}
Let $\gamma \in \C^+$, $M>1$, $\varepsilon\in (0,1/2)$, and let $\cQ$ be a quantum graph in the family $(\cQ_N)$. Suppose we are in case \emph{(i)} in Proposition~\ref{prop:ContractPhases}. Then for any $K_\gamma\in \mathscr{H}_1$ satisfying \eqref{eq:APrioriBoundOperators} and any $\ell\ge 1$, we have
\begin{equation*}
\| \mathcal{S}_{u^\gamma}^{4\ell} K_\gamma \|^2_\nu \leq (1-\varepsilon)^{2\ell} \|K_\gamma \|_\nu^2 + \ell \max_{0\leq j \leq \ell} \tilde{C}_{N,M}(\mathcal{S}_{u^\gamma}^{4j} K_\gamma).
\end{equation*}
\end{cor}

\subsubsection{Heuristics of the proof of Proposition \ref{prop:ContractPhases}}\label{sec:Heuristics}
Before proving the result, we would like to give a heuristics of why, in the second alternative, $u^\gamma$ must take this special form. To do this, we consider the case when $\gamma\in \R$ supposing that all the quantities we deal with are well-defined on the real axis, and we suppose for simplicity that there exists $K_\gamma$ such that $\|\mathcal{S}_{u^\gamma}^4 K_\gamma\|_\nu = \|K_\gamma\|_\nu$, which is the case in which (i) is the furthest from being satisfied.

By \eqref{eq:TrivialContractionBound}, we would have $\|\mathcal{S}_{u^\gamma}^j K_\gamma\|_\nu = \| K_\gamma\|_\nu$ for $j=1,2$. In particular, $\|\mathcal{S}_\gamma K_\gamma \|_\nu^2= \|K_\gamma\|_\nu^2$. By \eqref{eq:CriterionK}, this would imply that $K_\gamma\in F$. Similarly, we would have $\mathcal{S}_{u^\gamma} K_\gamma \in F$ and $\mathcal{S}_{u^\gamma}^2 K_\gamma\in F$. But, since $\gamma\in \R$ and $K_\gamma\in F$, we would have by \eqref{e:ASW} that $(\mathcal{S}_\gamma K_\gamma)(b) = K_\gamma (t_b)$, so that
\[
\left(\mathcal{S}_{u^\gamma} K_\gamma\right) (b) = \left(\mathcal{S}_{u^\gamma} K_\gamma\right)(o_b) = u^\gamma (b) K_\gamma (t_b),
\]
and, similarly,
\[
\left(\mathcal{S}^2_{u^\gamma} K_\gamma\right)(o_b) = u^\gamma (b) \left(\mathcal{S}_{u^\gamma} K_\gamma\right) (t_b).
\]
Since $|u^{\gamma}|=1$, we get in particular $\frac{|\cS_{u^{\gamma}}K_{\gamma}(o_b)|}{|K_{\gamma}(t_b)|}=\frac{|\cS_{u^{\gamma}}^2K_{\gamma}(o_b)|}{|\cS_{u^{\gamma}}K_{\gamma}(t_b)|} = 1$.

Writing $\frac{\left(\mathcal{S}_{u^\gamma} K_\gamma\right)(o_b)}{\left|\mathcal{S}_{u^\gamma} K_\gamma\right(o_b)|}=: \ee^{\ii\theta (o_b)}$, $\frac{\left(\mathcal{S}^2_{u^\gamma} K_\gamma\right)(o_b)}{\left|\mathcal{S}_{u^\gamma}^2 K_\gamma(o_b)\right|}=: \ee^{\ii\theta' (o_b)}$, $\frac{K_\gamma(t_b)}{ |K_\gamma (t_b)|}=: \ee^{\ii\theta''(t_b)}$, we obtain that
\[
u^\gamma(b)  = \ee^{\ii\theta(o_b) - \ii\theta''(t_b)} = \ee^{\ii\theta'(o_b)-\ii \theta(t_b)}.
\]
In particular, for all $b\in B$, we would have $\theta'(o_b)-\theta(o_b) = \theta(t_b)-\theta''(t_b)$. 

This quantity must be equal to a constant $c\in \R$, because the graph is not bipartite\footnote{If the graph is non-bipartite, it contains an odd cycle. As the quantity takes equal values for $b_1,b_2$ having same origin, and also for $b_1,b_2$ having same terminus, it follows that it must be constant on this cycle. From this, we readily see that it must be constant on the whole graph.} for $N$ large enough (since $(G_N)$ is expanding). Therefore, we would have
\begin{equation}\label{eq:BabyRelation}
u^\gamma(b) = \ee^{\ii c} \ee^{\ii \theta(o_b) - \ii \theta(t_b)}.
\end{equation}

This shows we are indeed in case (ii).

Let us continue these heuristics and show moreover that $\ee^{\ii c}\neq 1$ (this supplement to Proposition \ref{prop:ContractPhases} is the object of \S~\ref{sec:PhasePasUn}). This property will be essential in Section~\ref{Sec:End} to bound expressions of the form $\sum_j \mathcal{S}^j_{u^\gamma} K$.

 Suppose for contradiction that $u^\gamma(b) = \ee^{\ii \theta(o_b) - \ii \theta(t_b)}$. Writing $\zeta^\gamma(b)= \rho^\gamma(b) \ee^{\ii\varphi(b)}$, we deduce that $-2\varphi(b) = \theta(o_b) -  \theta(t_b)$. In particular, $\varphi(\hat{b}) = - \varphi(b) \mod \pi$.

By \eqref{e:zetainv}, we then have 
\begin{align*}
 \ee^{-\ii\varphi(b)} \Big[ \frac{1}{\rho^\gamma(b)} - \ee^{\ii(\varphi(b) + \varphi(\hat{b}))} \rho^\gamma(\hat{b})\Big] &= \frac{S_{\gamma}(L_b)}{\tilg^{\gamma}(t_b,t_b)}.
\end{align*}

Since the term between brackets is real and the phase of the right-hand side depends only on $t_b$ (as $S_{\gamma}$ is real), we deduce that $\ee^{-2\ii\varphi(b)}$ does not depend on $o_b$. Therefore, since $\ee^{2\ii\varphi(\hat{b})} = \ee^{-2\ii\varphi(b)}$, and since the graph is non bipartite, we deduce that $\ee^{-2\ii\varphi(b)}$ is a real constant. This constant must be one, otherwise we would have $\ee^{\ii\theta(o_b)}=-\ee^{\ii\theta(t_b)}$, contradicting non-bipartiteness. Therefore, $2\varphi(b) \equiv 0$. But this would imply that $\zeta^\gamma$ is always real, thus contradicting \Green{}, in view of \eqref{e:zetawt}.

\begin{rem}
One may wonder if \eqref{eq:CaseContract} always holds. The answer is no in general. For example take $G_N$ a family of expanders satisfying \BST, like the ones built in \cite{LuPhSa}. Consider quantum graphs on $G_N$ with all lengths equal, Kirchhoff boundary conditions, and zero potential.
This was the model studied in \cite{QEQGEQ}. 
Then all the oriented edges play exactly the same role, so $u^\gamma$ does not depend on $b$. In this case $\|\cS_{u^{\gamma}}^4K_{\gamma}\|_{\nu}^2 = \|\cS_{\gamma}^4K_{\gamma}\|_{\nu}^2$, which only contracts on subspaces, thus violating \eqref{eq:CaseContract}. Indeed, here we are in case (ii).
\end{rem}

\subsubsection{Proof of Proposition \ref{prop:ContractPhases}}

Let $K_\gamma\in \mathscr{H}_1$ satisfy \eqref{eq:APrioriBoundOperators}, and let $\gamma\in \C^+$. When we do not have $\|\mathcal{S}_{u^\gamma} K_\gamma \|_\nu^2 = \| K_\gamma\|_\nu^2$, but only $\|\mathcal{S}_{u^\gamma} K_\gamma \|_\nu^2 \approx \| K_\gamma\|_\nu^2$,  then we can still say that $\mathcal{S}_{u^\gamma}^j K_\gamma$ is close to being in $F$ for $j=0,1,2$. However, we cannot apply directly the previous argument, in which we divided by $|\mathcal{S}_{u^\gamma}^jK_\gamma(v)|$, since these could be very small, and could cause the different remainders to become large.

Therefore, we will have to show the stronger fact that $\mathcal{S}_{u^\gamma}^j K_\gamma$ is close to a function in $F$ of constant modulus. To this end, we will use several times the following lemma, which simply says that if $K$ is close to being in $F$ and is close to having constant modulus, then it is close to being a function in $F$ with constant modulus.

\begin{lem}\label{lem:Kandf}
Let us write $f_K:= P_{F,\nu} K$. We have
\begin{equation}
\label{eq:1}
\Big\| K - \|K\|_\nu \frac{f_K}{|f_K|}\Big{\|}_\nu \leq 2 \| K - f_K \|_\nu + 2 \| P_{\mathbf{1}^\perp,\nu}|K|\|_\nu,
\end{equation} 
with the convention that $\frac{f_K}{|f_K|}=1$ when $f_K$ vanishes.
\end{lem}
Note that the terms in the right-hand side of \eqref{eq:1} can be estimated by \eqref{eq:BoundProjFPerp} and \eqref{eq:BoundProjConstPerp}.
\begin{proof}
We have, by the triangle inequality
\[
\Big{\|} K - \|K\|_\nu \frac{f_K}{|f_K|}\Big{\|}_\nu \leq \| K- f_K \|_\nu + \Big{\|} f_K - \|K\|_\nu \frac{f_K}{|f_K|}\Big{\|}_\nu.
\]

Since dividing by $\frac{f_K}{|f_K|}$ does not change the $\|\cdot \|_\nu$ norm, we have
\begin{align*}
\Big{\|} f_K - \|K\|_\nu \frac{f_K}{|f_K|}\Big{\|}_\nu &= \big{\|} |f_K| - \|K\|_\nu \mathbf{1}\big{\|}_\nu\\
&\leq \big{\|} |K|- |f_K| \big{\|}_\nu + \big{\|} |K|- \|K\|_\nu \mathbf{1} \big{\|}_\nu\\
&\leq \big{\|} K- f_K \big{\|}_\nu + \big{\|} |K|- \|K\|_\nu \mathbf{1} \big{\|}_\nu.
\end{align*}

Let us write $C_K \mathbf{1}= P_\mathbf{1} |K|$. We have $ \big{\|} |K|- \|K\|_\nu \mathbf{1} \big{\|}_\nu\leq \big{\|} |K|-C_K \mathbf{1} \big{\|}_\nu + \big{|} \|K\|_\nu -C_K \big{|}\leq 2 \big{\|} |K|- C_K \mathbf{1} \big{\|}_\nu = 2 \| P_{\mathbf{1}^\perp,\nu}|K|\|_\nu$. Putting these inequalities together gives the result.
\end{proof}

\begin{proof}[Proof of Proposition \ref{prop:ContractPhases}]
Suppose that (i) does not hold: we can find $K_\gamma$ such that
\begin{equation}\label{eq:HypNonContr}
\| \mathcal{S}_{u^\gamma}^4 K_\gamma \|^2_\nu > (1-\varepsilon)^2 \|K_\gamma\|_\nu^2 + \tilde{C}_{N,M}(K_\gamma).
\end{equation}

\textbf{Step 1.}
By \eqref{eq:TrivialContractionBound}, we have $\|\cS_{\gamma}\cS_{u^{\gamma}}^jK_{\gamma}\|_{\nu}^2=\|\cS_{u^{\gamma}}^{j+1}K_{\gamma}\|_{\nu}^2 \ge \|\cS_{u^{\gamma}}^4K_{\gamma}\|_{\nu}^2$ and $\|K_{\gamma}\|_{\nu}^2\ge \|\cS_{u^{\gamma}}^jK_{\gamma}\|_{\nu}^2$ for $j=0,1,2$. Thus, \eqref{eq:HypNonContr} implies that for $j=0,1,2$,
\begin{equation}\label{eq:AlmostEq1}
\big{\|} \mathcal{S}^j_{u^\gamma} K_\gamma \big{\|}^2_\nu - \| \mathcal{S}_{\gamma} \mathcal{S}^j_{u^\gamma} K_\gamma\|_\nu^2 < 2 \varepsilon \|  \mathcal{S}^j_{u^\gamma} K_\gamma\|_\nu^2 - \tilde{C}_{N,M}(K_\gamma).
\end{equation}

Similarly, since $\|\mathcal{S}^2_\gamma |\mathcal{S}^j_{u^\gamma} K | \|^2_\nu \geq \|\mathcal{S}_{u^\gamma}^{j+2} K_\gamma\|_\nu^2$, \eqref{eq:HypNonContr} implies that for $j=0,1,2$,
\begin{equation}\label{eq:AlmostEq2}
\begin{aligned}
\| \mathcal{S}^j_{u^\gamma} K_\gamma\|_\nu^2 -\big{\|} \mathcal{S}_{\gamma}^2 | \mathcal{S}^j_{u^\gamma} K_\gamma | \big{\|}^2_\nu &< 2 \varepsilon \|\mathcal{S}^j_{u^\gamma} K_\gamma\|_\nu^2 - \tilde{C}_{N,M}(K_\gamma).
\end{aligned}
\end{equation}

Let $\delta_M:=  \frac{4 M^2}{3}$, $\delta'_M:= \frac{M^2}{c(D,\beta)}$. For $j=0,1,2$, we write $f_j:= P_{F,\nu} \mathcal{S}^j_{u^\gamma} K_\gamma$.
We now apply  Lemma \ref{lem:Kandf} to $\mathcal{S}^j_{u^\gamma} K_\gamma$, and use \eqref{eq:BoundProjFPerp}, \eqref{eq:BoundProjConstPerp}, \eqref{eq:AlmostEq1} and \eqref{eq:AlmostEq2} to obtain
\begin{equation}\label{eq:TheWeatherIsSoNiceHere}
\begin{aligned}
\Big{\|} \mathcal{S}^j_{u^\gamma} K_\gamma - \|\mathcal{S}^j_{u^\gamma} K_\gamma\|_\nu \frac{f_j}{|f_j|}\Big{\|}^2_\nu &\le 8 \|P_{F^\perp,\nu} \mathcal{S}^j_{u^\gamma} K_\gamma  \|^2_\nu + 8 \| P_{\mathbf{1}^\perp,\nu}|\mathcal{S}^j_{u^\gamma} K_\gamma|\|_\nu^2\\
&\le 8 \delta_M \big{(} 2 \varepsilon  \|\mathcal{S}^j_{u^\gamma} K_\gamma\|_\nu^2  + C_{N,M}(\mathcal{S}_{u^\gamma}^jK_\gamma) - \tilde{C}_{N,M}(K_\gamma)) \\
&+ 8 \delta_M' \big{(} 2 \varepsilon  \|\mathcal{S}^j_{u^\gamma} K_\gamma\|_\nu^2  + C_{N,M}'(\mathcal{S}_{u^\gamma}^jK_\gamma)- \tilde{C}_{N,M}(K_\gamma)) \\
&\leq 16\varepsilon (\delta_M+ \delta'_M) \|\cS_{u^{\gamma}}^jK_\gamma\|_\nu^2,
\end{aligned}
\end{equation}
thanks to the definition of $\tilde{C}_{N,M}(K_\gamma)$.

\medskip

\textbf{Step 2.} In this step, we use the fact that $\mathcal{S}_{u^\gamma}$ has a simple action on $F$.

Let us write $g_j :=  \|\cS^j_{u^\gamma} K_\gamma\|_\nu \frac{f_j}{|f_j|}$, and $R_j:= \cS^j_{u^\gamma} K_\gamma - g_j$. Then
\begin{align*}
\big{(}\mathcal{S}^{j+1}_{u^\gamma} K_\gamma \big{)}(b) &= u^\gamma(b) \big{(}\mathcal{S}_\gamma g_j\big{)}(b) + \big{(}\mathcal{S}_{u^\gamma} R_j\big{)}(b)\\
&= u^\gamma(b) g_j(t_b)+ u^\gamma(b) \omega^\gamma(b) g_j(t_b)  +  \big{(}\mathcal{S}_{u^\gamma} R_j\big{)}(b).
\end{align*}
But $\big{(}\mathcal{S}^{j+1}_{u^\gamma} K_\gamma\big{)}(b)= g_{j+1}(o_b) + R_{j+1}(b)$. So we obtain for $j=0,1$,
\begin{equation}
u^\gamma(b) = \frac{g_{j+1}(o_b)}{g_j(t_b)} + r_j(b),
\end{equation}
where $r_j= \frac{1}{g_j} \big{(} R_{j+1} - \mathcal{S}_{u^\gamma} R_j\big{)}- u^\gamma \omega^\gamma$. By \eqref{eq:TheWeatherIsSoNiceHere} and \eqref{eq:TrivialContractionBound},
\[
\|r_j\|_{\nu}\le \frac{\|R_{j+1}\|_{\nu}+\|\cS_{u^{\gamma}}R_j\|_{\nu}}{\|\cS_{u^{\gamma}}^jK_{\gamma}\|_{\nu}} + \|\omega^{\gamma}\|_{\nu}\le 8\sqrt{\varepsilon} (\delta_M+\delta'_M)^{1/2} +  \|\omega^\gamma\|_\nu.
\]

Therefore, 
\begin{equation}\label{eq:OriginAndTerminus}
\frac{g_1(t_b)}{g_0(t_b)}= \frac{g_2(o_b)}{g_1(o_b)} + r'(b),
\end{equation}
 where $\|r'\|_\nu \leq 16 \sqrt{\varepsilon} (\delta_M+\delta'_M)^{1/2} + 2 \|\omega^\gamma\|_\nu$.

\medskip

\textbf{Step 3.} In this last step, we show that $g_1(o_b) \approx c_0 g_0(o_b)$ to deduce the result.

Let us write $h_0(b) := \frac{g_1(o_b)}{g_0(o_b)}$, $h_1(b) =  \frac{g_1(t_b)}{g_0(t_b)}$, $h_2(b) =  \frac{g_2(o_b)}{g_1(o_b)}$ and $h_3(b) =  \frac{g_2(t_b)}{g_1(t_b)}$. Note that, by \eqref{eq:TrivialContractionBound} and \eqref{eq:AlmostEq1}, we have $1-2\varepsilon < \|h_j\|_\nu^2\leq 1$ for $j=0,1,2,3$.

Now using \eqref{eq:OriginAndTerminus},
\begin{align*}
\mathcal{S}_\gamma^2 h_0 &= \mathcal{S}_\gamma \big{(} h_1 + \omega^\gamma h_1 \big{)}= \mathcal{S}_\gamma \big{(} h_2 + r' + \omega^\gamma h_1 \big{)} = h_3 + r'',
\end{align*}
where $r''= \omega^\gamma h_3 + \mathcal{S}_\gamma r' + \mathcal{S}_\gamma  \omega^\gamma h_1$. In particular, $\|r''\|_\nu\leq \|r'\|_\nu + 2 \|\omega^\gamma\|_\nu$.

We deduce from \eqref{eq:BoundProjConstPerp} and Remark~\ref{rem:WhenWeHaveGoodBounds} that
\begin{align*}
\| P_{\mathbf{1}^\perp,\nu} h_0 \|^2_\nu & \leq\delta'_M \big{(} \|h_0\|_\nu^2 - \|h_3+ r'' \|_\nu^2 + C'_{N,M}(h_0) \big{)}\\
&\leq \delta'_M \big{(} 1 - \|h_3\|_\nu^2 + 2 \|r''\|_\nu - \|r'' \|_\nu^2 + C^0_{N,M}  \big{)}\\
&\leq \delta'_M \big{(} 2 \varepsilon + 2 \|r''\|_\nu  + C^0_{N,M}  \big{)}.
\end{align*}

We may thus write $h_0(b) = c_0 + r'''(b)$, where $c_0 \mathbf{1}:= P_{\mathbf{1},\nu} h_0$, so that $|c_0|\leq 1$, and $\|r'''\|_\nu^2\leq \delta'_M \big{(} 2 \varepsilon + 2 \|r''\|_\nu  + C^0_{N,M} \big{)}$. It follows that
\[
u^\gamma(b) = c_0 \frac{g_0(o_b)}{g_0(t_b)} + R(b),
\]
where $R(b)= r_0(b) + \frac{g_0(o_b)}{g_0(t_b)} r'''(b)$. Using the estimates we have on $r_0$ and $r'''$, and recalling that $\varepsilon\leq \sqrt{\varepsilon}$ (as $\varepsilon<\frac12$), we see that we may find a constant $\mathcal{C}(\beta, D)>0$ such that
\begin{equation*}
\|R\|^2_\nu \leq \mathcal{C}(\beta, D) M^3 \Big{[} \varepsilon^{1/2} + \|\omega^\gamma \|_{\nu} + \|\omega^\gamma \|_{\nu}^2 \Big{]} + \mathcal{C}(\beta, D) M^2 C^0_{N,M}.
\end{equation*}

Recalling that $|g_0(o_b)| = |g_0(t_b)|$, we write $\ee^{\ii \theta (v)} := \frac{g_0(v)}{|g_0(v)|}$, which gives us the result.
\end{proof}

\subsubsection{Properties of the phases $u^\gamma$}\label{sec:PhasePasUn}
In the previous subsection, we did not use the precise relation between $u^\gamma$ and $\zeta^\gamma$. Now, we are going to use \eqref{e:zetainv} to show that, if case (ii) of Proposition~\ref{prop:ContractPhases} occurs then $c_0$ cannot get very close to unity.

Recall that the quantity $\cC_0$ was defined in \eqref{eq:defC0}.

\begin{prp}\label{prop:PhasesReallyChange}
For any large $L>1$ there exists $\delta_0= \delta_0(N,\gamma,L)>0$ satisfying  $\liminf\limits_{\eta_0\downarrow 0}\liminf\limits_{N\to+\infty}\delta_0\ge CL^{-8}$, such that if case \emph{(ii)} of Proposition \ref{prop:ContractPhases} is satisfied, then
\[
|1-c_0| \geq \delta_0- \cC_0.
\]
\end{prp}

\begin{proof}
Let $r_0:=|1-c_0|$. Note that there exists $C_I>0$ independent of $N$ and of $b$ such that for all $\gamma\in  I + \ii\left[0,1\right]$, we have 
\begin{equation}\label{e:imsgamma}
|\Im S_\gamma (L_b)|\leq C_I \Im \gamma.
\end{equation}

By \eqref{eq:lowerDir}, we have $|\Re(S_\gamma(L_b))|\geq C'_{\mathrm{Dir}}>0$ for all $N$ and all $b\in B_N$.

Let us write $\zeta^\gamma(b) = \rho^\gamma(b) \ee^{\ii \phi^\gamma(b)}$, so that $u^\gamma(b) = \ee^{-2\ii \phi^\gamma(b)}$.

\textbf{Step 1: $\ee^{2\ii \phi^\gamma(b)}$ depends (almost) only on $t_b$.}
By assumption, we have
\begin{equation}\label{eq:ThetaRevers}
\begin{aligned}
\left\| \ee^{-2\ii \phi^\gamma(b)} - \ee^{\ii \theta(o_b)-\ii\theta(t_b)} \right\|_\nu &\leq \cC_0 + r_0\\ 
\left\| \ee^{-2\ii \phi^\gamma(\hat{b})} - \ee^{\ii \theta(t_b)-\ii\theta(o_b)} \right\|_\nu &\leq \cC_0 + r_0,
\end{aligned}
\end{equation}
so that
\[
\Big\| 1 -\ee^{-2\ii \phi^\gamma(b)- 2\ii \phi^\gamma(\hat{b})} \Big\|_\nu = \Big\| \ee^{2\ii \phi^\gamma(b)} -\ee^{-2\ii \phi^\gamma(\hat{b})} \Big\|_\nu \leq  2\cC_0 + 2 r_0.
\]

Let us write
\[
\epsilon(b) = \begin{cases} 1 &\text{ if } \Re (\ee^{-\ii \phi^\gamma(b)- \ii \phi^\gamma(\hat{b})})\geq 0 \\
-1 &\text{ if } \Re (\ee^{-\ii \phi^\gamma(b)- \ii \phi^\gamma(\hat{b})}) <0,
\end{cases}
\]
so that $ |\ee^{-\ii \phi^\gamma(b)- \ii \phi^\gamma(\hat{b})} - \epsilon(b)| \leq |1 -  \ee^{-2\ii \phi^\gamma(b)- 2\ii \phi^\gamma(\hat{b})}|$. Since $|\ee^{-\ii \phi^\gamma(b)- \ii \phi^\gamma(\hat{b})} - \epsilon(b)|\leq 2$, we deduce that for any $s \geq 2$,
\[
\Big{\|} \epsilon(b) -\ee^{-\ii \phi^\gamma(b)- \ii \phi^\gamma(\hat{b})} \Big{\|}_{\ell^s(\nu)} = \Big{\|} \epsilon(b) -\ee^{\ii \phi^\gamma(b)+ \ii \phi^\gamma(\hat{b})} \Big{\|}_{\ell^s(\nu)} \leq  2  (\cC_0 +r_0)^{2/s}.
\]

The first part of \eqref{e:zetainv} can be rewritten as
\[
\ee^{-\ii \phi^\gamma(b)}  \Big{(} \frac{1}{\rho^\gamma(b)} - \rho^\gamma(\hat{b}) \ee^{\ii \phi^\gamma(b) + \ii \phi^\gamma(\hat{b})}\Big{)} = \frac{S_{\gamma}(L_b)}{\tilg^{\gamma}(t_b,t_b)}.
\]
Therefore, we have
\begin{equation*}
\ee^{-\ii \phi^\gamma(b)}  \Big{(} \frac{1}{\rho^\gamma(b)} - \rho^\gamma(\hat{b}) \epsilon(b) \Big{)} = \frac{S_{\gamma}(L_b)}{\tilg^{\gamma}(t_b,t_b)} + R_0,
\end{equation*}
with $\|R_0\|_{\ell^s(\nu)} \leq 2 (\cC_0 + r_0)^{1/s} \|\zeta^\gamma\|_{\ell^{2s}(\nu)}$.

Let us write 
\[
G(b):=\tilg^\gamma(t_b,t_b)= |\tilg^\gamma(t_b,t_b)| \ee^{\ii \psi^\gamma(t_b)},
\]
then recalling \eqref{e:imsgamma},
\begin{equation*}
 \ee^{\ii \phi^\gamma(b)- \ii \psi^\gamma(t_b)}\frac{\Re(S_{\gamma}(L_b))}{|\tilg^{\gamma}(t_b,t_b)|} =   \Big{(} \frac{1}{\rho^\gamma(b)} - \rho^\gamma(\hat{b}) \epsilon(b) \Big{)}+ R_1,
\end{equation*}
where $\|R_1\|_{\ell^s(\nu)}\leq \|R_0\|_{\ell^s(\nu)} + C_I |\Im \gamma | \big{\|}G^{-1}\big{\|}_{\ell^s(\nu)} $.

Taking the imaginary parts, we obtain
\[
 \sin \big( \phi^\gamma(b)- \psi^\gamma(t_b)\big)\frac{\Re(S_{\gamma}(L_b))}{|\tilg^{\gamma}(t_b,t_b)|} =   \Im (R_1)(b),
\]
 so that
\[
 \Big|\sin \big( \phi^\gamma(b)- \psi^\gamma(t_b)\big)\Big| \leq  \frac{|\tilg^{\gamma}(t_b,t_b)|}{C'_{\mathrm{Dir}}}   |\Im (R_1)(b)|.
\]

Let $\delta^\gamma(b) \in \pi \Z$ be such that  $\phi^\gamma(b)- \psi^\gamma(t_b) + \delta^\gamma(b) \in \big[-\frac{\pi}{2}, \frac{\pi}{2}\big]$.
We have
\[
 \Big|\sin \big{(} \phi^\gamma(b)- \psi^\gamma(t_b)\big{)}\Big| = \Big|\sin \big( \phi^\gamma(b)- \psi^\gamma(t_b) + \delta^\gamma(b) \big)\Big| \geq \frac{2}{\pi}  \big| \phi^\gamma(b)- \psi^\gamma(t_b) + \delta^\gamma(b) \big|.
\]
Therefore,
\[
 \big{|} \phi^\gamma(b)- \psi^\gamma(t_b) + \delta^\gamma(b) \big{|} \leq  \frac{\pi}{2}\cdot \frac{|\tilg^{\gamma}(t_b,t_b)|}{C'_{\mathrm{Dir}}}   |\Im (R_1)(b)|,
\]
As $|\ee^{\ii x} - \ee^{\ii t}|\le |x-t|$, we get,
\begin{equation}\label{eq:PhiPsiClose}
\ee^{2\ii \phi^\gamma(b)} = \ee^{2\ii \psi^\gamma(t_b)} + R_2(b),
\end{equation}
with $|R_2(b)|\leq \pi \frac{|\tilg^{\gamma}(t_b,t_b)|}{C'_{\mathrm{Dir}}}   |R_1(b)|$. In particular, we have
\[
\|R_2\|_{\nu} \leq \frac{\pi}{C'_{\mathrm{Dir}}}  \|G\|_{\ell^4(\nu)} \|R_1\|_{\ell^4(\nu)}.
\]

\textbf{Step 2: $\ee^{2\ii \phi^\gamma(b)}$ is almost equal to one.}
Using \eqref{eq:ThetaRevers} and \eqref{eq:PhiPsiClose}, we obtain that
\begin{equation}\label{eq:RevInv2}
 \ee^{\ii \theta(t_b)} =  \ee^{\ii \theta(o_b)-2\ii \psi^\gamma(o_b)} + R_3(b),
\end{equation}
with $\|R_3\|_\nu \leq \|R_2\|_\nu + \cC_0 + r_0$.

Let us write $f_1(b)= \ee^{\ii \theta(o_b)}$, $f_2(b)= \ee^{\ii \theta(t_b)}$, $f_3(b) =  \ee^{\ii \theta(o_b)-2\ii \psi^\gamma(o_b)}$, $f_4(b) =  \ee^{\ii \theta(t_b)-2\ii \psi^\gamma(t_b)}$. We have
\[
\mathcal{S}_\gamma^2 f_1 = \mathcal{S}_\gamma \big{(} f_2 + \omega^\gamma f_2 \big{)}= \mathcal{S}_\gamma \big{(} f_3 + R_3 +  \omega^\gamma f_2 \big{)}=f_4 +\omega^\gamma f_4 + \mathcal{S}_\gamma \big{(} R_3 +  \omega^\gamma f_2 \big{)},
\]
so that $\|\mathcal{S}_\gamma^2 f_1 \|_\nu \geq 1- \|R_3\|_\nu - 2\|\omega^\gamma \|_{\nu}$. It follows that
\[
\|f_1\|^2_{\nu}-\|\cS_{\gamma}^2f_1\|_{\nu}^2=(\|f_1\|_{\nu}+\|\cS_{\gamma}^2f_1\|_{\nu})(\|f_1\|_{\nu}-\|\cS_{\gamma}^2f_1\|_{\nu}) \le 2\|R_3\|_{\nu} + 4\|\omega^{\gamma}\|_{\nu}.
\]

From \eqref{eq:BoundProjConstPerp} and Remark~\ref{rem:WhenWeHaveGoodBounds}, we thus get for any $L>0$,
\[
\| P_{\mathbf{1}^\perp,\nu}f_1\|^2_\nu \leq \frac{L^2}{c(D,\beta)} \big{(}2 \|R_3\|_\nu + 4\|\omega^\gamma \|_\nu + C^0_{N,L}\big{)},
\]
so there exists $s_0\in \C$ such that
\begin{equation}\label{e:r4}
f_1(b) = s_0 + R_4(b),
\end{equation}
with $|s_0|\le 1$ and $\|R_4\|_\nu^2\le \frac{L^2}{c}(2 \|R_3\|_\nu + 4\|\omega^\gamma \|_\nu + C^0_{N,M})$.

Since also $f_1(\hat{b}) = s_0+R_4(\hat{b})$, we deduce from \eqref{eq:ThetaRevers} that  $\ee^{2 \ii \varphi^\gamma(b)} = |s_0|^2 + R_5(b)$, where $\|R_5 \|_\nu \leq 2 \|R_4\|_\nu + \|R_4\hat{R}_4\|_{\nu} + r_0 + \cC_0$, where $\hat{R}_4(b)=R_4(\hat{b})$. By \eqref{e:r4}, $|\hat{R}_4(b)|\le 2$, so $\|R_5 \|_\nu \leq 4 \|R_4\|_\nu + r_0 + \cC_0$.
Writing
\[
s(b):= \begin{cases} |s_0| &\text{ if } \Re (\ee^{\ii \varphi^\gamma(b)})>0 \\
  -|s_0| &\text{ if } \Re (\ee^{\ii \varphi^\gamma(b)})\leq 0, \end{cases}
\]
we always have $|\ee^{\ii\varphi^\gamma(b)} + s(b)| \geq 1$, so that $|\ee^{\ii\varphi^\gamma(b)}-s(b)|\leq |\ee^{2\ii\varphi^\gamma(b)}-s^2(b)| = |R_5(b)|$. In particular, we get $\big{|}\Im \ee^{\ii \varphi^\gamma(b)} \big{|} \leq |R_5(b)|$ and thus $\|\Im \ee^{\ii\varphi^\gamma(b)}\|_{\nu}\le \|R_5\|_{\nu}$.

Putting it all together, we obtain 
\begin{multline*}
\big\|\Im \ee^{\ii\phi^\gamma} \big\|_\nu^2 \le \frac{32L^2}{c_{D,\beta}}\Big(\frac{2\pi}{C_{\mathrm{Dir}}'}\|G\|_{\ell^4(\nu)}\big[2(\cC_0+r_0)^{1/4}\|\zeta^\gamma\|_{\ell^8(\nu)}+C_I\Im\gamma\|G^{-1}\|_{\ell^4(\nu)}\big]\\
+2(\cC_0+r_0)+4\|\omega\|_{\nu}+C_{N,L}^0\Big)+2(r_0+\cC_0)^2.
\end{multline*}
On the other hand, by the Cauchy-Schwarz inequality,
\[
\big\|\Im \zeta^\gamma\big\|^2_\nu = \big\| |\zeta^\gamma| \Im \ee^{\ii \phi^\gamma}\big\|^2_\nu \le \big\| |\zeta^\gamma|^2 \big\|_\nu \big\|(\Im \ee^{\ii\phi^\gamma})^2 \big\|_\nu \le \big\| |\zeta^\gamma|^2 \big\|_\nu \big\|\Im \ee^{\ii\phi^\gamma} \big\|_\nu ,
\]
so
\[
1\le \big\|\Im \zeta^\gamma\big\|_{\nu}^2\cdot \big\||\Im \zeta^\gamma|^{-1}\big\|_{\nu}^2\le \big\||\Im \zeta^\gamma|^{-1}\big\|_{\nu}^2\cdot\big\| |\zeta^\gamma|^2 \big\|_\nu\cdot \big\|\Im \ee^{\ii\phi^\gamma} \big\|_\nu.
\]
Thanks to Remarks~\ref{rem:HolStable},  \ref{rem:WhenWeHaveGoodBounds} and \ref{rem:MeasureNuIsLessScaryThanCoronavirus} and to   \eqref{eq:EstimateOnXi}, we deduce that 
\[
1\leq  L^2\left[\big( r_0 + \mathcal{C}_0 \big)^{1/4} + \Im \gamma + L^{-s}\right]\times  O_{N\to +\infty, \gamma}(1),
\]
where we bounded $(r_0+\cC_0)^\alpha\le (r_0+\cC_0)^{1/4}$ for $\alpha=1,2$, which holds if $r_0+\cC_0\le 1$ (if $r_0+\cC_0>1$ then the proposition is trivially true).

Therefore, we have $(r_0+\mathcal{C}_0)^{1/4}\ge \frac{1}{L^2O_{N\to +\infty, \gamma}(1)}-\Im\gamma-L^{-s}$, so that
\[
r_0\ge \Big(\frac{1}{L^2O_{N\to +\infty, \gamma}(1)} -\Im\gamma-L^{-s}\Big)^4- \mathcal{C}_0.
\]
Taking $s=3$, the claim follows.
\end{proof}

\section{End of the proof}\label{Sec:End}
Recall that by Lemma~\ref{lem:ReducK2}, our aim is to estimate, for operators $K'_\gamma\in \mathscr{H}_k$ satisfying \Hol{}, the quantity

\[
\lim_{n\to \infty} \limsup_{\eta_0\downarrow 0}\limsup_{N\to\infty} \sup\limits_{\lambda\in I_1}  \frac{\mu_1^\gamma(\mathrm{B}_1)}{N n^2} \Big|\sum_{j=1}^{n} (n-j)   \langle \mathcal{J}^* K'_{\gamma},\cS_{u^{\gamma}}^j \mathrm{T} K'_{\gamma}\rangle_{\ell^2(\mathrm{B}_1,\nu_1^{\gamma})}\Big|.
\]

From now on, we take $M$ arbitrarily large, and take $\varepsilon_M:= M^{-8}$, $n_M=M^9$. 

We will now consider each alternative that can happen in Corollary \ref{cor:ContractPhases2}.

\medskip

\textbf{First alternative}: Suppose that case (i) of Proposition \ref{prop:ContractPhases} is satisfied for $\varepsilon = \varepsilon_M$.

We may apply Corollary~\ref{cor:ContractPhases2}, \eqref{eq:BornesTJ} and \eqref{eq:TrivialContractionBound}, to get for all $j\in \N$ 
\begin{equation}\label{e:indi}
\big\|\mathcal{S}^{j}_{u^\gamma}\mathrm{T} K'_{\gamma}\big\|_{\nu_1^\gamma} \leq  (1-\varepsilon_M)^{\lfloor\frac{j}{4}\rfloor} \|K_\gamma'\|_{\nu_k^\gamma} +  j \max_{0\leq i \leq j} \tilde{C}_{N,M}(\mathcal{S}_{u^\gamma}^{i} \mathrm{T} K'_{\gamma}).
\end{equation}

By Corollary \ref{cor:BadIsntTooBad}, we have $ \sup_{i} \tilde{C}_{N,M}(\mathcal{S}_{u^\gamma}^i \mathrm{T} K'_{\gamma})=  M^{-10} O_{N\to +\infty, \gamma} (1)$. We deduce from this and \eqref{eq:BornesTJ} that
\begin{equation}\label{e:case1}
\begin{aligned}
&\frac{\mu_1^\gamma(\mathrm{B}_1)}{N n_M^2} \Big|\sum_{j=1}^{n_M} (n_M-j)   \langle \mathcal{J}^* K'_{\gamma},\cS_{u^{\gamma}}^j \mathrm{T} K'_{\gamma}\rangle_{\nu_1^{\gamma}}\Big|\\
&\quad \le \frac{\mu_1^\gamma(\mathrm{B}_1)}{N n_M} \|K'_\gamma\|_{\nu_k^\gamma}^2\sum_{j'=1}^{n_M} (1-\varepsilon_M)^{\lfloor\frac{j}{4}\rfloor} + \frac{n_M}{N} \|K'_\gamma\|_{\nu_k^\gamma} \mu_1^\gamma(\mathrm{B}_1)  \max_{1\leq j \leq n_M} \tilde{C}_{N,M}(\mathcal{S}_{u^\gamma}^j \mathrm{T} K'_{\gamma}) \\
&\quad =   \frac{1}{M} O_{N\to +\infty, \gamma} (1),
\end{aligned}
\end{equation}
where we estimated $\sum_j(1-\varepsilon)^{\lfloor j/4\rfloor }\le 4\cdot 2^{3/4}\sum_j(1-\varepsilon)^j\le 8\varepsilon^{-1}$, as $\varepsilon\le \frac{1}{2}$, then we used \eqref{eq:APrioriBoundOperators} and Remark~\ref{rem:MeasureNuIsLessScaryThanCoronavirus}.

\medskip

\textbf{Second alternative}: Suppose that case (ii) of Proposition \ref{prop:ContractPhases} is satisfied, still for $\varepsilon = \varepsilon_M$, so that the addition of phases $u^{\gamma}$ did not improve the contraction of $\|\cS_{u^\gamma}^jK\|$. Then we will not be able to control \emph{individual} scalar products as in \eqref{e:indi}, however phases will still help to control the scalar products \emph{in mean} by inducing cancellations. 

Let us write $\theta_0(b):= \theta(o_{b})$ and $\theta_1(b)=\theta(t_b)$. Given $K\in \mathscr{H}_1$, we get
\[
\big{(}\cS_\gamma \ee^{-\ii\theta_0} K\big{)}(b) = \big{(} \ee^{-\ii\theta_1} \cS_\gamma K \big{)}(b).
\]

Therefore, 
\begin{align*}
\big{\|}\mathcal{S}_{u^\gamma} K- c_0 \ee^{\ii\theta_0} \mathcal{S}_\gamma \ee^{-\ii\theta_0} K\big{\|}_\nu & =\big{\|} \big{(}u^\gamma - c_0 \ee^{\ii(\theta_0 - \theta_1)} \big{)} \mathcal{S}_\gamma K\big{\|}_\nu   \\
&\leq \big{\|} u^\gamma - c_0 \ee^{\ii(\theta_0 - \theta_1)} \big{\|}_{\ell^4(\nu)} \big{\|}\mathcal{S}_\gamma K\big{\|}_{\ell^4(\nu)}\\
&\leq  \sqrt{2} \big{\|} u^\gamma - c_0 \ee^{\ii(\theta_0 - \theta_1)} \big{\|}_{\ell^2(\nu)}^{1/2} \big{\|}\mathcal{S}_\gamma K\big{\|}_{\ell^4(\nu)}\\
&\leq \sqrt{2\cC_0(M,\varepsilon_M,\gamma)}\cdot  \|K\|_{\ell^4(\nu)},
\end{align*}
where we used $|u^\gamma - c_0 \ee^{\ii(\theta_0 - \theta_1)}|\leq 2$ in the second inequality, \eqref{eq:TrivialContractionBound} in the last one, and $\cC_0$ is as in \eqref{eq:defC0}. 
Thanks to \eqref{eq:TrivialContractionBound}, we get for all $j\ge 1$ that 
\begin{align*}
\big\|\cS^j_{u^\gamma} K- c^j_0 \ee^{\ii\theta_0} \cS^j_\gamma \ee^{-\ii\theta_0} K\big\|_\nu &= \Big\|\sum_{i=1}^j \cS_{u^{\gamma}}^{j-i}(\cS_{u^{\gamma}}-c_0\ee^{\ii\theta_0}\cS_{\gamma}\ee^{-\ii\theta_0})(c_0 \ee^{\ii\theta_0}\cS_{\gamma}\ee^{-\ii\theta_0})^{i-1} \cdot K\Big\|_{\nu} \\
&\leq j\sqrt{2\cC_0} \cdot \|K\|_{\ell^4(\nu)}.
\end{align*}

Using \eqref{eq:TrivialContractionBound} again, it follows that for all $m$ and $t$,
\begin{multline}\label{e:packets}
\left|\sum_{j=1}^{t}\langle \cJ^{\ast}K_{\gamma}', \cS_{u^\gamma}^{j+m} \mathrm{T} K_{\gamma}'\rangle_{\nu} -\sum_{j=1}^{t}c_0^j\langle \cJ^{\ast} K_\gamma',  \ee^{\ii \theta_0}\cS_{\gamma}^{j} \ee^{-\ii \theta_0}\cS_{u^\gamma}^{m}  \mathrm{T}K_\gamma'\rangle_{\nu} \right|
\\ \leq t^2\sqrt{2\cC_0}\cdot \| \mathcal{J}^* K_\gamma'\|_{\ell^2(\nu)}  \cdot \|\mathrm{T}K_\gamma'\|_{\ell^4(\nu)} \le t^2\sqrt{2\cC_0}\cdot \| K_\gamma'\|_{\ell^2(\nu)}  \cdot \|K_\gamma'\|_{\ell^4(\nu)}
\end{multline}
where the bound on $\mathrm{T}:\ell^4\to\ell^4$ follows as in \eqref{e:tbound}--\eqref{eq:BornesTJ}, using H\"older's inequality.

As we will see below, the size $t$ of packets should be chosen so that $t \sqrt{\cC_0}$ is small as $M$ gets large. Remembering that $\cC_0\lesssim f(\beta,D)(M^{3/2}\varepsilon^{1/4}+O(\eta_0))$ and $\varepsilon=M^{-8}$,
we take $t=M^{\alpha}$ with $0<\alpha< 1/4$. We now write
\begin{align*}
\Big|\sum_{j=1}^{n_M}(n_M-j) \langle \mathcal{J}^* K'_{\gamma},\cS_{u^{\gamma}}^j \mathrm{T} K'_{\gamma}\rangle_{\nu_1^{\gamma}}\Big| &= \Big|\sum_{r=1}^{n_M}\sum_{j=1}^{n_M-r}\langle \cJ^{\ast}K_\gamma',\cS_{u^{\gamma}}^j\mathrm{T}K_\gamma'\rangle_{\nu_1^{\gamma}}\Big|\\
&\le \Big| \sum_{r=1}^{n_M}\sum_{a=0}^{\lfloor\frac{n_M-r}{t}\rfloor-1}\sum_{j=1+at}^{t(a+1)} \langle \mathcal{J}^* K'_{\gamma},\cS_{u^{\gamma}}^j \mathrm{T} K'_{\gamma}\rangle_{\nu_1^{\gamma}}\Big| + n_Mt\, \|K_{\gamma}'\|^2_{\nu} \, ,
 \end{align*}
where we estimated $| \sum_{r=1}^{n_M}\sum_{j=t\lfloor\frac{n_M-r}{t}\rfloor}^{n_M-r}\langle \mathcal{J}^* K'_{\gamma},\cS_{u^{\gamma}}^j \mathrm{T} K'_{\gamma}\rangle_{\nu_1^{\gamma}}| \le n_Mt \|K_\gamma'\|^2_{\nu}$.

Note that $\sum_{j=1+at}^{t(a+1)}\langle \cdot,\cS_{u^\gamma}^{j}\cdot\rangle = \sum_{j=1}^{t}\langle \cdot,\cS_{u^\gamma}^{j+at}\cdot\rangle$. So using \eqref{e:packets},
\begin{multline*}% \label{e:gathering}
\Big|\sum_{r=1}^{n_M}\sum_{a=0}^{\lfloor\frac{n_M-r}{t}\rfloor-1}\sum_{j=1+at}^{t(a+1)} \langle \mathcal{J}^* K'_{\gamma},\cS_{u^{\gamma}}^j \mathrm{T} K'_{\gamma}\rangle_{\nu_1^{\gamma}}\Big|\\
\le \Big|\sum_{r=1}^{n_M}\sum_{a=0}^{\lfloor\frac{n_M-r}{t}\rfloor-1}\sum_{j=1}^{t}c_0^j\langle \cJ^{\ast} K_\gamma',  \ee^{\ii \theta_0}\cS_{\gamma}^{ j} \ee^{-\ii \theta_0}\cS_{u^\gamma}^{at}  \mathrm{T}K_\gamma' \rangle_{\nu} \Big|
\\ +  n_M\cdot \frac{n_M}t \cdot t^2\sqrt{2\cC_0}\cdot \|K_\gamma'\|_{\nu}\|K_{\gamma}'\|_{\ell^4(\nu)} \, .
\end{multline*} 

Therefore,
\begin{equation}\label{e:case2}
\begin{aligned}
&\frac{\mu_1^\gamma(\mathrm{B}_1)}{N n_M^2} \Big|\sum_{j=1}^{n_M} (n_M-j)   \langle \mathcal{J}^* K'_{\gamma},\cS_{u^{\gamma}}^j \mathrm{T} K'_{\gamma}\rangle_{\nu_1^{\gamma}}\Big|\\
&\le \frac{\mu_1^\gamma(\mathrm{B}_1)}{N n_M^2} \Big|\sum_{r=1}^{n_M}\sum_{a=0}^{\lfloor\frac{n_M-r}{t}\rfloor-1}\Big\langle \cJ^{\ast} K_\gamma',  \ee^{\ii \theta_0}\sum_{j=1}^{t}c_0^j\cS_{\gamma}^{ j} \ee^{-\ii \theta_0}\cS_{u^\gamma}^{at}  \mathrm{T}K_\gamma' \Big\rangle_{\nu}\Big| \\
&\quad + \frac{\mu_1^\gamma(\mathrm{B}_1)}{N}  \cdot t\left( \sqrt{2\cC_0} \| K_\gamma'\|_{\nu_k^{\gamma}}\| K_\gamma'\|_{\ell^4(\nu)} +n_M^{-1}\|K_\gamma'\|_{\nu}^2\right)  .
\end{aligned}
\end{equation}
Thanks to \eqref{eq:APrioriBoundOperators} and Remark \ref{rem:MeasureNuIsLessScaryThanCoronavirus}, the last term goes to zero for our choice of $t$, as $N\to\infty$ followed by $\eta_0\downarrow 0$ followed by $M\to\infty$.

Now, we decompose $\ee^{-\ii\theta_0} \cS_{u^{\gamma}}^{at}\mathrm{T} K'_{\gamma}  = P_{\mathbf{1},\nu} \ee^{-\ii\theta_0} \cS_{u^{\gamma}}^{at}\mathrm{T} K'_{\gamma}  + P_{\mathbf{1}^\perp,\nu}\ee^{-\ii\theta_0} \cS_{u^{\gamma}}^{at}\mathrm{T} K'_{\gamma} $.

Consider the first term. Recall that $(\cS_\gamma\mathbf{1}- \mathbf{1})(b)= \omega^\gamma(b)$. We deduce that 
\[
\mathcal{S}_\gamma^j  P_{\mathbf{1},\nu} \ee^{-\ii\theta_0} \cS_{u^{\gamma}}^{at}\mathrm{T} K'_{\gamma} =  P_{\mathbf{1},\nu} \ee^{-\ii\theta_0} \cS_{u^{\gamma}}^{at}\mathrm{T} K'_{\gamma} + R_j,
\]
where $R_j = c_{\gamma}\sum_{k=0}^{j-1}\cS_{\gamma}^k\omega^\gamma$ if $P_{\mathbf{1},\nu} \ee^{-\ii\theta_0}\cS_{u^{\gamma}}^{at} \mathrm{T} K'_{\gamma} = c_{\gamma}\mathbf{1}$. By \eqref{eq:EstimateOnXi} and \eqref{eq:TrivialContractionBound},
\[
\|R_j\|_\nu = j\eta_0 O_{N\to +\infty, \gamma}(1).
\]
With $\cS_\gamma^j$ gone, we obtain a geometric sum, so we get
\begin{multline*}
\bigg\langle \mathcal{J}^* K'_{\gamma},\ee^{\ii\theta_0} \sum_{j=1}^{t}  c_0 ^j \mathcal{S}^j_\gamma P_{\mathbf{1},\nu} \ee^{-\ii\theta_0} \cS_{u^{\gamma}}^{at}\mathrm{T} K'_{\gamma} \bigg\rangle_{\nu_1^{\gamma}} \\
=\ee^{\ii\theta_0}\left\langle \mathcal{J}^* K'_{\gamma}, P_{\mathbf{1},\nu} \ee^{-\ii\theta_0} \cS_{u^{\gamma}}^{at}\mathrm{T} K'_{\gamma} \right\rangle_{\nu_1^{\gamma}}    \frac{c_0-c_0^{t+1}}{1-c_0} + t^2O_{N\to +\infty, \gamma}(\eta_0).
\end{multline*}
By Proposition~\ref{prop:PhasesReallyChange}, we have $|1-c_0|\geq \delta_0- \cC_0$, where $\liminf\limits_{\eta_0\downarrow 0}\liminf\limits_{N\to\infty}\delta_0\ge CL^{-8}$. Take $L=M^{1/64}$ and recall that $\cC_0 \lesssim (M^{3/2}\varepsilon_M^{1/4}+O(\eta_0)+M^{-s+1})\approx M^{-1/2}+O(\eta_0)+M^{-s+1}$. For $s=2$ we thus get $|1-c_0|\gtrsim M^{-1/8}$ as $N\to \infty$ followed by $\eta_0\downarrow 0$. Using \eqref{eq:BornesTJ}, \eqref{eq:TrivialContractionBound}, \eqref{eq:APrioriBoundOperators} and Remark \ref{rem:MeasureNuIsLessScaryThanCoronavirus}, we thus get
\begin{multline}\label{e:case2a}
\bigg| \frac{\mu_1^\gamma(\mathrm{B}_1)}{N n_M^2} \sum_{r=1}^{n_M} \sum_{a=0}^{\lfloor\frac{n_M-r}{t}\rfloor-1} \bigg\langle \mathcal{J}^* K'_{\gamma},\ee^{\ii\theta_0} \sum_{j=1}^{t}  c_0 ^j \mathcal{S}^j_\gamma P_{\mathbf{1},\nu} \ee^{-\ii\theta_0} \cS_{u^{\gamma}}^{at}\mathrm{T} K'_{\gamma} \bigg\rangle_{\nu_1^{\gamma}} \bigg|  \\
\le \frac{1}{t}\left(M^{1/8}+t^2\eta_0\right)O_{N\to\infty,\gamma}(1).
\end{multline}

Finally, let us deal with the term $P_{\mathbf{1}^\perp,\nu}\ee^{-\ii\theta_0} \cS_{u^{\gamma}}^{at}\mathrm{T} K'_{\gamma}$. By \eqref{e:cur2}, we have $\cS_{\gamma}^{\ast}\mathbf{1}=\mathbf{1} + \tilde{\omega}^{\gamma}$, where $\tilde{\omega}^{\gamma}(b) = \frac{-\Im \gamma}{\Im R_{\gamma}^-(t_b)} \int_{0}^{L_b} |\xi_-^{\gamma}(x_b)|^2\,\dd x_b$. So by iteration, $\cS_{\gamma}^{\ast\,l} \mathbf{1} = \mathbf{1} +  \sum_{s=0}^{l-1} \cS_{\gamma}^{\ast\,s}\tilde{\omega}^{\gamma}$. Hence, $\langle \mathbf{1},\mathcal{S}_{\gamma}^lJ\rangle_{\nu} = \langle \mathbf{1},J\rangle_{\nu} + \langle \sum_{s=0}^{l-1}\mathcal{S}_{\gamma}^{\ast\,s}\tilde{\omega}^{\gamma},J\rangle_{\nu}$. Denoting $\mathcal{Z}_l J :=  \tilde{\omega}^{\gamma} \sum_{s=0}^{2l-1} \mathcal{S}_{\gamma}^sJ$, we see that if $J\perp \mathbf{1}$, then $(\cS_{\gamma}^{2l}J -  \mathcal{Z}_l J) \perp \mathbf{1}$. Consequently, by Proposition~\ref{prop:ContractConstant}, for any $L_0$,
\begin{align*}
 \|\cS_{\gamma}^{2(r+1)}J\|_{\nu} &\le \|\cS_{\gamma}^{2}(\mathcal{S}_{\gamma}^{2r } - \mathcal{Z}_{r})J\|_{\nu} + \|\mathcal{Z}_{r} J\|_{\nu} \\
& \le \left(1-c(D,\beta)L_0^{-2}\right)^{1/2} \|(\cS_{\gamma}^{2r}  - \mathcal{Z}_{r})J\|_{\nu} + \,C_{N,L_0,r}'(J)^{1/2} +  \|\mathcal{Z}_{r} J\|_{\nu}\, .
\end{align*}
where $C_{N,L_0,l}'(J)= C_{N,L_0}'((\mathcal{S}_{\gamma}^{2l} - \mathcal{Z}_l )J) = O_{N\to + \infty, \gamma}(1)L^{-s}_0$ by Corollary~\ref{cor:BadIsntTooBad}. Using \eqref{eq:EstimateOnXi} and $\|(\mathcal{S}_{\gamma}^{2r} - \mathcal{Z}_{r})J\| \le \|\mathcal{S}_{\gamma}^{2r} J\| + \| \mathcal{Z}_{r}J\|$, we get by iteration
\begin{align*}
\| \mathcal{S}_{\gamma}^{2r }  J \|_{\nu} &\le \left(1-L^{-2}_0c(D,\beta)\right)^{r/2}  \| J\|_{\nu} +   \sum_{l=0}^{r-1}C_{N,L_0,l}'(J)^{1/2}  + 2  \sum_{l=0}^{r-1} \| \mathcal{Z}_{l}  J\|_{\nu}\\
&= \left(1-L^{-2}_0c(D,\beta)\right)^{r/2}  \| J\|_{\nu} + \left(rL^{-s}_0 + r^2\eta_0\right) O_{N\to + \infty, \gamma}(1)
\end{align*}
for any $J\perp \mathbf{1}$ in $\ell^2(\nu)$ satisfying \Hol{}.

Decomposing $\sum_{j=1}^tc_0^j\cS_{\gamma}^jJ = \sum_{r=1}^{\lfloor t/2\rfloor}c_0^{2r}\cS_{\gamma}^{2r}J+\sum_{r=0}^{\lfloor t/2\rfloor}c_0^{2r+1}\cS_{\gamma}\cS_{\gamma}^{2r}J$ we thus get
\begin{multline}\label{e:case2b}
 \left| \frac{\mu_1^\gamma(\mathrm{B}_1)}{N n_M^2} \sum_{r=1}^{n_M}  \sum_{a=0}^{\lfloor\frac{n_M-r}{t}\rfloor-1}\Big\langle \mathcal{J}^* K'_{\gamma},\ee^{\ii\theta_0} \sum_{j=1}^{t}  c_0 ^j \mathcal{S}^j_\gamma P_{\mathbf{1}^\perp,\nu} \ee^{-\ii\theta_0} \cS_{u^{\gamma}}^{at}\mathrm{T} K'_{\gamma} \Big\rangle_{\nu_1^{\gamma}}\right| \\
 \le \frac{1}{t}\Big( \frac{L^2_0}{c(D,\beta)}+ t^2 L^{-s}_0 + t^3\eta_0\Big)O_{N\to +\infty, \gamma}(1).
\end{multline}

Recall that $t=M^{\alpha}$, $\alpha<\frac{1}{4}$. We now choose $\alpha = 3/16$, $s = 4$ and let $L_0=M^{1/16}$. Then collecting \eqref{e:case1}, \eqref{e:case2}, \eqref{e:case2a}, \eqref{e:case2b}, we see that whether $G_N$ is in case (i) or (ii), the variance is bounded by quantities vanishing in the limit $N\To\infty$ followed by $\eta_0\downarrow 0$ followed by $M\To\infty$. This concludes the proof of Theorem~\ref{thm:qeqg}.

\appendix
\section{Further properties of Green's function on quantum trees}\label{sec:greenquan}
This appendix is devoted to a corollary of Lemma~\ref{lem:IdentitiesGreen}, which can be thought of as a kind of ``recursive relation'' for the imaginary parts of Green's functions. Its origin in the combinatorial case is the recursive relation of the spherical function of regular trees. We also discuss the proof of Lemma~\ref{lem:ASW} afterwards.

\begin{cor}          \label{cor:psiiden}
Define $\Psi_{\gamma,v}(w) = \Im G^{\gamma}(v,w)$ for $\gamma\in \C\setminus\R$. For any $b\in B$, we have
\begin{equation}     \label{eq:psiiden3}
\Psi_{\gamma,o_b}(t_b) - \overline{\zeta^{\gamma}(\hat{b})}\Psi_{\gamma,t_b}(t_b) - \zeta^{\gamma}(b)\Psi_{\gamma,o_b}(o_b) + \overline{\zeta^{\gamma}(\hat{b})}\zeta^{\gamma}(b)\Psi_{\gamma,t_b}(o_b)=-\overline{\zeta^{\gamma}(\hat{b})} \zeta^{\gamma}(b)\Im S_{\gamma}(L_b) \, ,
\end{equation}
while, if $k\geq 2$, for any $(b_1,\dots, b_k)\in \mathrm{B}_k(\mathbb{T})$, we have
\begin{equation}     \label{eq:psiiden2}
\Psi_{\gamma,o_{b_1}}(t_{b_k}) - \overline{\zeta^{\gamma}(\hat{b}_1)}\Psi_{\gamma,t_{b_1}}(t_{b_k}) - \zeta^{\gamma}(b_k)\Psi_{\gamma,o_{b_1}}(o_{b_k}) + \overline{\zeta^{\gamma}(\hat{b}_1)}\zeta^{\gamma}(b_k)\Psi_{\gamma,t_{b_1}}(o_{b_k})=0 \, .
\end{equation}
Finally, we have
\begin{multline}     \label{e:psiiden1}
\Im \left(\frac{G^{\gamma}(t_b,t_b)}{S_{\gamma}^2(L_b)}\right) - \frac{\zeta^{\gamma}(b)}{S_{\gamma}(L_b)}\Im \left(\frac{G^{\gamma}(o_b,t_b)}{S_{\gamma}(L_b)}\right)-\frac{\overline{\zeta^{\gamma}(b)}}{\overline{S_{\gamma}(L_b)}}\Im \left(\frac{G^{\gamma}(o_b,t_b)}{S_{\gamma}(L_b)}\right) \\
+ \left|\frac{\zeta^{\gamma}(b)}{S_{\gamma}(L_b)}\right|^2\Im G^{\gamma}(o_b,o_b) = \Im \left(\frac{\zeta^{\gamma}(b)}{S_{\gamma}(L_b)}\right)\le \Im R_{\gamma}^+(o_b) \,.
\end{multline}
\end{cor}
\begin{proof}
Since $G^{\gamma}(v_0,v_k) = G^{\gamma}(v_0,v_{k-1})\zeta^{\gamma}(b_k)$, taking the imaginary parts (and remembering that $\Im (zz') = z \Im z' + \overline{z'}\Im z$) yields
\begin{equation}\label{e:psipoi}
\Im G^{\gamma}(o_{b_1},t_{b_k}) - \zeta^{\gamma}(b_k)\Im G^{\gamma}(o_{b_1},o_{b_k}) = \Im \zeta^{\gamma}(b_k)\cdot \overline{G^{\gamma}(o_{b_1},o_{b_k})}\,.
\end{equation}
In particular, 
\begin{equation}\label{eq:MaxIsTired}
\Psi_{\gamma,o_b}(t_b) - \zeta^{\gamma}(b)\Psi_{\gamma,o_b}(o_b) = \Im \zeta^{\gamma}(b)\overline{G^{\gamma}(o_b,o_b)}.
\end{equation} 
Next, we have $G^{\gamma}(t_b,t_b) = \frac{G^{\gamma}(o_b,t_b)}{\zeta^{\gamma}(\hat{b})}$ and $\frac{1}{\zeta^{\gamma}(\hat{b})} = \zeta^{\gamma}(b) + \frac{S_{\gamma}(L_b)}{G^{\gamma}(o_b,o_b)}$ by \eqref{e:zetainv}. Hence, 
\begin{equation}\label{eq:MaxIsTired2}
G^{\gamma}(t_b,t_b) = \zeta^{\gamma}(b)G^{\gamma}(o_b,t_b) + S_{\gamma}(L_b)\zeta^{\gamma}(b)
\end{equation} and thus
\[
\Psi_{\gamma,t_b}(t_b) = \Im \zeta^{\gamma}(b)\Re G^{\gamma}(o_b,t_b) + \Re \zeta^{\gamma}(b)\Psi_{\gamma,o_b}(t_b) + \Im \left(\zeta^{\gamma}(b)S_{\gamma}(L_b)\right).
\]
It follows that
\begin{equation}\label{eq:psi1rel}
\Psi_{\gamma,t_b}(t_b) - \zeta^{\gamma}(b)\Psi_{\gamma,o_b}(t_b) =\Im \zeta^{\gamma}(b) \overline{G^{\gamma}(o_b,t_b)}+\Im \left(\zeta^{\gamma}(b)S_{\gamma}(L_b)\right).
\end{equation}

Using \eqref{eq:MaxIsTired2} and $\Im (zz') = z \Im z' + \overline{z'}\Im z$ again, we deduce that
\begin{align*}
& \overline{\zeta^{\gamma}(\hat{b})}[\Psi_{\gamma,t_b}(t_b) - \zeta^{\gamma}(b)\Psi_{\gamma,t_b}(o_b)] = \Im \zeta^{\gamma}(b)\overline{\zeta^{\gamma}(\hat{b})G^{\gamma}(o_b,t_b)}+\overline{\zeta^{\gamma}(\hat{b})}\Im \left(\zeta^{\gamma}(b)S_{\gamma}(L_b)\right) \\
& \quad= \Im \zeta^{\gamma}(b)\overline{\left[G^{\gamma}(o_b,o_b)-S_{\gamma}(L_b)\zeta^{\gamma}(\hat{b})\right]} +\overline{\zeta^{\gamma}(\hat{b})}\Im \left(\zeta^{\gamma}(b)S_{\gamma}(L_b)\right) \\
&\quad = \Im \zeta^{\gamma}(b)\overline{G^{\gamma}(o_b,o_b)}+\overline{\zeta^{\gamma}(\hat{b})} \zeta^{\gamma}(b)\Im S_{\gamma}(L_b).
\end{align*}
Recalling \eqref{eq:MaxIsTired}, this proves \eqref{eq:psiiden3}.

Now let $k\ge 2$. If we apply \eqref{e:psipoi}  to $(b_2,\dots,b_k)$, we obtain 
\begin{equation*}
\Im G^{\gamma}(o_{b_2},t_{b_k}) - \zeta^{\gamma}(b_k)\Im G^{\gamma}(o_{b_2},o_{b_k}) = \Im \zeta^{\gamma}(b_k)\cdot \overline{G^{\gamma}(o_{b_2},o_{b_k})}\,.
\end{equation*}
Multiplying this equation by $\overline{\zeta^{\gamma}(\hat{b}_1)}$ and using \eqref{e:greenmul},  we get
\begin{equation*}
\overline{\zeta^{\gamma}(\hat{b}_1)}\Im G^{\gamma}(o_{b_2},t_{b_k}) - \overline{\zeta^{\gamma}(\hat{b}_1)} \zeta^{\gamma}(b_k)\Im G^{\gamma}(o_{b_2},o_{b_k}) = \Im \zeta^{\gamma}(b_k)\cdot \overline{G^{\gamma}(o_{b_1},o_{b_k})}\,.
\end{equation*}
But, by \eqref{e:psipoi}, the RHS is equal to $\Im G^{\gamma}(o_{b_1},t_{b_k}) - \zeta^{\gamma}(b_k)\Im G^{\gamma}(o_{b_1},o_{b_k})$, so \eqref{eq:psiiden2} follows.

Finally, let us prove \eqref{e:psiiden1}. By \eqref{eq:MaxIsTired2},
\begin{equation*}
\frac{G^{\gamma}(t_b,t_b)}{S_{\gamma}^2(L_b)} = \frac{\zeta^\gamma(b)}{S_{\gamma}(L_b)}\frac{G^{\gamma}(o_b,t_b)}{S_{\gamma}(L_b)} + \frac{\zeta^\gamma(b)}{S_\gamma(L_b)}.
\end{equation*}
To show the equality in \eqref{e:psiiden1}, we must therefore show that 
\[
\Im \left(  \frac{\zeta^\gamma(b)}{S_{\gamma}(L_b)}\frac{G^{\gamma}(o_b,t_b)}{S_{\gamma}(L_b)}\right) =  2 \Re\left(\frac{\zeta^\gamma(b)}{S_{\gamma}(L_b)}\right)\Im \left(\frac{G^{\gamma}(o_b,t_b)}{S_{\gamma}(L_b)}\right) -  \left|\frac{\zeta^{\gamma}(b)}{S_{\gamma}(L_b)}\right|^2\Im G^{\gamma}(o_b,o_b).
\]
We have 
$\left|\frac{\zeta^{\gamma}(b)}{S_{\gamma}(L_b)}\right|^2 G^{\gamma}(o_b,o_b) = \frac{\overline{\zeta^{\gamma}(b)}}{\overline{S_{\gamma}(L_b)}} \frac{G^{\gamma}(o_b,t_b)}{S_\gamma(L_b)}$, so that 
\begin{multline*}
\Im \left(  \frac{\zeta^\gamma(b)}{S_{\gamma}(L_b)}\frac{G^{\gamma}(o_b,t_b)}{S_{\gamma}(L_b)}\right)  + \left|\frac{\zeta^{\gamma}(b)}{S_{\gamma}(L_b)}\right|^2\Im G^{\gamma}(o_b,o_b) \\
= \Im \left(  \frac{\zeta^\gamma(b)}{S_{\gamma}(L_b)}\frac{G^{\gamma}(o_b,t_b)}{S_{\gamma}(L_b)}\right) + \Im \bigg(  \frac{\overline{\zeta^\gamma(b)}}{\overline{S_{\gamma}(L_b)}}\frac{G^{\gamma}(o_b,t_b)}{S_{\gamma}(L_b)}\bigg) =  2 \Re\left(\frac{\zeta^\gamma(b)}{S_{\gamma}(L_b)}\right)\Im \left(\frac{G^{\gamma}(o_b,t_b)}{S_{\gamma}(L_b)}\right),
\end{multline*}
and the equality in \eqref{e:psiiden1} follows. For the inequality in \eqref{e:psiiden1}, we use that $\Im R_{\gamma}^+(o_b) = \Im \frac{\zeta^{\gamma}(b)-C_{\gamma}(L_b)}{S_{\gamma}(L_b)} \ge \Im \frac{\zeta^{\gamma}(b)}{S_{\gamma}(L_b)}$ by Remark~\ref{rem:hj} below.
\end{proof}

\begin{proof}[Proof of Lemma~\ref{lem:ASW}]
We only prove the ``current'' relations, see \cite[Section 2]{AC} for the remaining parts. We will also use that $\zeta^\gamma(b) = \frac{V_{\gamma;o}^+(t_b)}{V_{\gamma;o}^+(o_b)}$, as follows from \cite[Lemma 2.1]{AC}.

Since $V_{\gamma;o}^+$ satisfies the $\delta$-conditions, we have $\sum_{b^+\in \cN_b^+} R_{\gamma}^+(o_{b^+}) = R_{\gamma}^+(t_b) + \alpha_{t_b}$, so $\sum_{b^+\in \cN_b^+} \Im R_{\gamma}^+(o_{b^+}) = \Im R_{\gamma}^+(t_b)$. Similarly, as $\sum_{w_-\in \cN_w^-} U_{w_-}'(L_{w_-}) + \alpha_w U_w(0) =   U_w'(0)$, we get $\sum_{b^-\in \cN_b^-} \Im R_{\gamma}^-(t_{b^-}) = \Im R_{\gamma}^-(o_b)$.

Suppose $H f = \gamma f$ and let $J_{\gamma}(x_b) = \Im [\overline{f(x_b)}f'(x_b)]$. Then $J_{\gamma}'(x_b) = \Im [|f'(x_b)|^2 + \overline{f(x_b)}[W(x_b)f(x_b)-\gamma f(x_b)] = -\Im \gamma |f(x_b)|^2$.  Therefore,
\begin{equation}\label{eq:JConserved}
J_{\gamma}(t_b) = J_{\gamma}(o_b) - \Im \gamma \int_{0}^{L_b} |f_b(x_b)|^2 \,\dd x_b.
\end{equation}

But for $f_b(x_b)=\xi_+^{\gamma}(x_b) = \frac{V_{\gamma; o_b}^+(x_b)}{V_{\gamma; o_b}^+(o_b)}$, we have $J(t_b) = |\zeta^{\gamma}(b)|^2\Im R_{\gamma}^+(t_b)$ and $J(o_b) = \Im R_{\gamma}^+(o_b)$, so \eqref{e:cur1} follows. Equation \eqref{e:cur2} also follows by taking $f_b(x_b)= \xi_-^{\gamma}(x_b)= \frac{U_{\gamma;v}^-(x_b)}{U_{\gamma;v}^-(t_b)}$. To deduce \eqref{eq:DefXi}, express $U_{\gamma;v}^-(x_b)$ using the data at $t_b$, which reads
$U_{\gamma;v}^-(x_b) = U_{\gamma;v}^-(t_b)[S_{\gamma}'(L_b)C_{\gamma}(x_b)-C_{\gamma}'(L_b)S_{\gamma}(x_b)] + (U_{\gamma;v}^-)'(t_b)[-S_{\gamma}(L_b)C_{\gamma}(x_b)+C_{\gamma}(L_b)S_{\gamma}(x_b)]$. Recalling \eqref{e:simplif} completes the proof.
\end{proof}

\begin{rem}\label{rem:hj}
It also follows that $\Im (\frac{S_{\gamma}'(L_b)}{S_{\gamma}(L_b)}) \le 0$.

In fact, $\Im(\frac{S_{\gamma}'(L_b)}{S_{\gamma}(L_b)}) = \frac{1}{|S_{\gamma}(L_b)|^2} \Im [\overline{S_{\gamma}(L_b)}S_{\gamma}'(L_b)] \le \frac{1}{|S_{\gamma}(L_b)|^2} \Im[\overline{S_{\gamma}(0)}S_{\gamma}'(0)] =0$. If the potentials are symmetric, then $S_{\gamma}'(L_b)=C_{\gamma}(L_b)$, so we also get $\Im (\frac{C_{\gamma}(L_b)}{S_{\gamma}(L_b)}) \le 0$.
\end{rem}

\section{Proofs of the reduction}\label{app:redproofs}

\subsection{Reduction to a discrete variance}\label{sec:disred}
The aim of this subsection is to prove Proposition~\ref{prop:Reduc1}. In the course of the proof, we will omit the subscripts $N$ from $f_N$ and $\mathcal{K}_N$ to lighten notations.

\smallskip

\textbf{Proof of (1).}
Let $\psi_j$ be an eigenfunction on $\cG_N$ corresponding to $\lambda_j$. On the edge $b$, we have $\psi_{j,b}(x_b) = \psi_j(o_b) C_{\lambda_j}(x_b) + \psi'_{j}(o_b) S_{\lambda_j}(x_b)$. Evaluating at $x_b=L_b$, we get $\psi_{j,b}(x_b) = [C_{\lambda_j}(x_b)-\frac{C_{\lambda_j}(L_b)}{S_{\lambda_j}(L_b)}S_{\lambda_j}(x_b)]\psi_j(o_b) + \frac{S_{\lambda_j}(x_b)}{S_{\lambda_j}(L_b)}\psi_j(t_b)$. From \eqref{e:simplif}, we get
\begin{equation}\label{e:psijegal}
\psi_{j,b}(x_b) = \frac{S_{\lambda_j}(L_b-x_b)}{S_{\lambda_j}(L_b)} \psi_j(o_b) + \frac{S_{\lambda_j}(x_b)}{S_{\lambda_j}(L_b)}\psi_j(t_b) \,.
\end{equation}
Since $2\langle \psi_j, f\psi_j\rangle = \sum_b \int_0^{L_b}f_b(x_b)|\psi_{j,b}(x_b)|^2\,\dd x_b$, it easily follows that
\begin{equation}\label{e:1stscal}
2 \,\langle \psi_j, f\psi_j\rangle = \langle \mathring{\psi}_j, (K_{f,j}+J_{f,j}+M_{f,j}^{(1)}+M_{f,j}^{(2)})_G\mathring{\psi}_j\rangle \,,
\end{equation}
where $\mathring{\psi}_j(v)=\psi_j(v)$, $K_{f,j},J_{f,j}\in\mathscr{H}_0$ and $M_{f,j}\in \mathscr{H}_1$ are defined by
\[
K_{f,j}(v) = \sum_{b\, :\, o_b = v} \frac{1}{S^2_{\lambda_j}(L_b)}\int_0^{L_b}f_b(x_b)S^2_{\lambda_j}(L_b-x_b)\,\dd x_b \,,
\]
\[
J_{f,j}(v) = \sum_{b\, :\, t_b=v}  \frac{1}{S^2_{\lambda_j}(L_b)}\int_0^{L_b}f_b(x_b)S^2_{\lambda_j}(x_b)\,\dd x_b \,,
\]
\[
M_{f,j}^{(1)}(b) = \frac{1}{S^2_{\lambda_j}(L_b)}\int_0^{L_b}f_b(x_b)S_{\lambda_j}(L_b-x_b)S_{\lambda_j}(x_b)\,\dd x_b = M_{f,j}^{(2)}(\hat{b})\,.
\]

Let us write
\begin{multline}\label{e:2ndscal}
2\left[ f\right]^{\gamma_j} := \left\langle \mathring{\psi}_j, \left(\left[\langle K_{f,j}\rangle^{\gamma_j} + \langle J_{f,j}\rangle^{\gamma_j} + \langle M_{f,j}^{(1)}\rangle^{\gamma_j} + \langle M_{f,j}^{(2)}\rangle^{\gamma_j}\right] \mathbf{1}\right)_G\mathring{\psi}_j\right\rangle \\
= \left(\langle K_{f,j} \rangle^{\gamma_j} + \langle J_{f,j}\rangle^{\gamma_j} + 2\langle M_{f,j}\rangle^{\gamma_j} \right)\|\mathring{\psi}_j\|^2  \,,
\end{multline}
where $\langle M_{f,j}^{(1)}\rangle^{\gamma_j} = \langle M_{f,j}^{(2)}\rangle^{\gamma_j} =:\langle M_{f,j}\rangle^{\gamma_j}$ because $\Psi_{\gamma,t_b}(o_b)=\Psi_{\gamma,o_b}(t_b)$. 

From \eqref{e:1stscal} and \eqref{e:2ndscal}, we have
\begin{multline}\label{e:withav}
\frac{1}{\mathbf{N}_N(I)}\sum_{\lambda_j\in I}\left|\langle \psi_j, f\psi_j\rangle - [ f]^{\gamma_j} \right| \le \varie\left( \big{(}K_{f,\lambda}\big{)}_G-\langle K_{f,\lambda}\rangle^{\gamma} \mathbf{1}\right) + \varie\left( \big{(}J_{f,\lambda}\big{)}_G-\langle J_{f,\lambda} \rangle^{\gamma} \mathbf{1}\right) \\+ \varie\left(\big{(}M_{f,\lambda}^{(1)}\big{)}_G-\langle M_{f,\lambda}^{(1)} \rangle^{\gamma} \mathbf{1} \right) +  \varie\left(\big{(}M_{f,\lambda}^{(2)}\big{)}_G-\langle M_{f,\lambda}^{(2)} \rangle^{\gamma} \mathbf{1} \right).
\end{multline}

We now need to replace $[f]^\gamma$ by $\langle f\rangle_{\gamma}$. By \eqref{e:2ndscal},
\begin{equation}\label{e:badav}
[ f ]^{\gamma_j} = \frac{\left\|\mathring{\psi}_j^{(N)}\right\|^2}{2\sum_{v\in V_N}\Im \tilg^{\gamma_j}_N(v,v)} \sum_{b\in B} \int_0^{L_b} f_b(x_b) \Psi^{\gamma_j}_b(x_b) \,\dd x_b \,,
\end{equation}
where
\begin{multline}\label{e:psimachin}
\Psi^{\gamma_j}_b(x_b)= \\ \frac{\Im \tilg^{\gamma_j}_N(o_b,o_b)S^2_{\lambda_j}(L_b-x_b) +  \Im \tilg^{\gamma_j}_N(t_b,t_b)S_{\lambda_j}^2(x_b) + 2 \Im \tilg^{\gamma_j}_N(o_b,t_b)S_{\lambda_j}(L_b-x_b)S_{\lambda_j}(x_b)}{S_{\lambda_j}^2(L_b)}.
\end{multline}

Let us also introduce the quantity
\begin{equation}\label{e:maximf}
\mathbf{f}^\gamma =\mathbf{f}^\gamma_N = \frac{\int_{\mathcal{G}_N} f(x) \Psi^{\gamma}(x) \dd x}{\int_{\mathcal{G}_N} \Psi^{\gamma}(x) \dd x}.
\end{equation}

Then

\begin{equation}\label{eq:ThreeTerms}
\begin{aligned}
\frac{1}{\mathbf{N}_N(I)}\sum_{\lambda_j\in I}|\langle \psi_j,f\psi_j\rangle-\langle f\rangle_{\gamma_j}| &\le \frac{1}{\mathbf{N}_N(I)}\sum_{\lambda_j\in I} |\langle \psi_j,f\psi_j\rangle - [ f]^{\gamma_j}|\\
& + \frac{1}{\mathbf{N}_N(I)}\sum_{\lambda_j\in I} \Big|[ f]^{\gamma_j} - \mathbf{f}^{\gamma_j}\Big|
 + \frac{1}{\mathbf{N}_N(I)}\sum_{\lambda_j\in I} \Big| \mathbf{f}^{\gamma_j}- \langle f\rangle_{\gamma_j}\Big|.
\end{aligned}
\end{equation}

We already controlled the first term. Let us turn to the second. Noting that 
\[
[ f]^{\gamma_j} = [ \mathbf{f}^{\gamma_j} \hat{\mathbf{1}} ]^{\gamma_j}
\]
we have 
\[
\frac{1}{\mathbf{N}_N(I)}\sum_{\lambda_j\in I} \left|[ f]^{\gamma_j} - \mathbf{f}^{\gamma_j}\right| = \frac{1}{\mathbf{N}_N(I)}\sum_{\lambda_j\in I}  | [ \mathbf{f}^{\gamma_j} \hat{\mathbf{1}} ]^{\gamma_j} - \langle \psi_j, \mathbf{f}^{\gamma_j} \hat{\mathbf{1}}  \psi_j \rangle |,
\]
which can be bounded by discrete variances as in \eqref{e:withav}. For example $K_{f,j}(v)$ is now replaced by $K_{\mathbf{f}^{\gamma_j}\hat{\mathbf{1}},j}(v)$. By definition we see that $K_{\mathbf{f}^{\gamma_j}\hat{\mathbf{1}},j}(v) = \mathbf{f}^{\gamma_j}K_{\hat{\mathbf{1}},j}(v)$.

To deal with the last term, we use \cite[eq.(A.12)]{AC}, which shows that
\begin{equation}\label{e:expansame}
\begin{aligned}
\tilg^z(x_b,x_b) =& \frac{1}{S_z(L_b)}\big([S_z(L_b)C_z(x_b)-C_z(L_b)S_z(x_b)]\tilg^z(o_b,x_b) \\
&+ S_z(x_b)[\tilg^z(t_b,x_b) + S_z(L_b)C_z(x_b)-C_z(L_b)S_z(x_b)]\big) \\ 
=& \frac{S_z(L_b-x_b)\tilg^z(o_b,x_b)+S_z(x_b)\tilg^z(t_b,x_b)+S_z(x_b)S_z(L_b-x_b)}{S_z(L_b)}\\
=& \frac{S_{z}^2(L_{b}-x_{b}) \tilg^{z}(o_b,o_{b}) + 2S_{z}(L_b-x_b)S_{z}(x_{b})\tilg^{z}(o_b,t_{b}) +S_{z}^2(x_{b})\tilg^{z}(t_b,t_{b})}{S_{z}^2(L_b)}\\&+\frac{S_z(x_b)S_z(L_b-x_b)}{S_z(L_b)} ,
\end{aligned}
\end{equation}
where we used \eqref{e:simplif} in the second equality. Recalling \eqref{e:psimachin}, we deduce that
\begin{equation}\label{eq:PsiAlmostG}
|\Psi_b^{\gamma_j}(x_b)-\Im \tilg^{\gamma_j}(x_b,x_b)| \le C\eta_0 \left( |\tilg^{\gamma_j}(o_b,o_b)|+2|\tilg^{\gamma_j}(o_b,t_b)|+|\tilg^{\gamma_j}(t_b,t_b)|+1\right), 
\end{equation}
where we used that the function $S_z$ is Lipschitz continuous on $I+\ii[0,1]$, the constant $C$ is uniform in $(z,L,W)\in (\overline{I}+\ii[0,1])\times \mathrm{Lip}_{\mathrm{M}}[\mathrm{m},\mathrm{M}]$ and depends on $C_{\mathrm{Dir}}$. Comparing \eqref{eq:DefBracket} and \eqref{e:maximf} we get

\begin{multline}\label{e:avminav}
|\mathbf{f}^{\gamma}- \langle f\rangle_{\gamma}| \le \frac{|\int_{\mathcal{G}_N}f(x)(\Im \tilg^{\gamma}(x,x)-\Psi^{\gamma}(x))\dd x|}{\int_{\mathcal{G}_N}\Im \tilg(x,x)\dd x} \\
+ \Big|\int_{\mathcal{G}_N}f(x)\Psi^{\gamma}(x) \dd x\Big|\cdot \Big|\frac{1}{\int_{\mathcal{G}_N} \Im \tilg^\gamma(x,x)\dd x}-\frac{1}{\int_{\mathcal{G}_N} \Psi^\gamma(x) \dd x}\Big|\\
\le C\eta_0\frac{\sum_{b\in B_N} \left( |\tilg^{\gamma}(o_b,o_b)|+2|\tilg^{\gamma}(o_b,t_b)|+|\tilg^{\gamma}(t_b,t_b)|+1\right)}{\int_{\mathcal{G}_N}\Im \tilg^{\gamma}(x,x)\dd x}\Big(1+\frac{\int_{\mathcal{G}_N}|\Psi^{\gamma}(x)|\dd x}{|\int_{\mathcal{G}_N}\Psi^{\gamma}(x)\dd x|}\Big),
\end{multline}
where we used that $|f|\le 1$. It follows from Proposition~\ref{prop:NowGreenIsReallyCool} that
\begin{multline*}
\limsup_{N\to\infty}\sup_{\lambda\in I}|\mathbf{f}^{\gamma}-\langle f\rangle_{\gamma}| \\
\le C'\eta_0\frac{\expect_{\prob}[L_b(G^\gamma(o_b,o_b)+2|G^\gamma(o_b,t_b)|+|G^\gamma(t_b,t_b)|+1)]}{\expect_{\prob}(\Im \mathbf{G}^z)}\Big(1+\frac{\expect_{\prob}(|\Psi^{\gamma}|)}{|\expect_{\prob}(\Psi^{\gamma})|}\Big).
\end{multline*}

Note that, by \eqref{eq:PsiAlmostG},
\[
|\mathbb{E}_\mathbb{P}(\Psi^\gamma)|\geq \mathbb{E}_\mathbb{P}(\Im \tilg^\gamma(x_b,x_b)) - C \eta_0 \mathbb{E}_\mathbb{P} \left( |\tilg^{\gamma_j}(o_b,o_b)|+2|\tilg^{\gamma_j}(o_b,t_b)|+|\tilg^{\gamma_j}(t_b,t_b)|+1\right).
\]

Therefore, thanks to Proposition~\ref{prp:b7}
 $\limsup_{N\to\infty}\sup_{\lambda\in I}|\mathbf{f}^{\gamma}-\langle f\rangle_{\gamma}|$ vanishes as $\eta_0\downarrow 0$. Thus, 
\[
\lim\limits_{\eta_0\downarrow 0}\limsup\limits_{N\to\infty} \frac{1}{\mathbf{N}_N(I)}\sum_{\lambda_j\in I} \big| \mathbf{f}^{\gamma_j}- \langle f\rangle_{\gamma_j}\big|=0.
\]

Recalling \eqref{eq:ThreeTerms}, this concludes the proof of (1), up to verifying that the operators are in \Hol{}. For $K_{f,\lambda}(v)$, the bounds $|K_{f,\lambda}(v)|\le \sum_{o_b=v}\frac{1}{S_{\lambda}^2(L_b)}\int_0^{L_b}S_{\lambda}^2(L_b-x_b)\,\dd x_b$ and $|K_{f,\lambda}(v)-K_{f,\lambda'}(v)|\le \sum_{o_b=v}\int_0^{L_b}|\frac{S_{\lambda}^2(L_b-x_b)}{S_{\lambda}^2(L_b)}-\frac{S_{\lambda'}^2(L_b-x_b)}{S_{\lambda'}^2(L_b)}|\,\dd x_b$ easily imply this. In fact here we can simply bound uniformly $|K_{f,\lambda}(v)|\le C_{D,\mathrm{M},\mathrm{Dir}}$ and $|K_{f,\lambda}(v)-K_{f,\lambda'}(v)|\le C'|\lambda-\lambda'|$. The operators $J_{f,\lambda},M_{f,\lambda}$ are similar. Finally $|K_{\mathbf{f}^{\gamma}\hat{\mathbf{1}},\lambda}(v)|\le  C\frac{\int_{\cG_N}|\Psi^{\gamma}(x)|\,\dd x}{|\int_{\cG_N}\Psi^{\gamma}(x)\,\dd x|}$ and $|K_{\mathbf{f}^{\gamma},\lambda}(v)-K_{\mathbf{f}^{\gamma'},\lambda'}(v)|\le |K_{\mathbf{f}^{\gamma},\lambda}(v)-K_{\mathbf{f}^{\gamma},\lambda'}(v)| + |K_{\mathbf{f}^{\gamma},\lambda'}(v)-K_{\mathbf{f}^{\gamma'},\lambda'}(v)|$. The first is bounded by $C'|\lambda-\lambda'|\frac{\int_{\cG_N}|\Psi^{\gamma}(x)|\,\dd x}{|\int_{\cG_N}\Psi^{\gamma}(x)\,\dd x|}$, the second by $C|\mathbf{f}^{\gamma}-\mathbf{f}^{\gamma'}|\le C\frac{\int_{\cG_N}|\Psi^\gamma-\Psi^{\gamma'}|}{|\int_{\cG_N}\Psi^\gamma|}(1+\frac{\int_{\cG_N}|\Psi^\gamma|}{|\int_{\cG_N}\Psi^{\gamma'}|})$, by arguing as in \eqref{e:avminav}. Using Corollary~\ref{cor:ConfinedTilMay11} and Proposition~\ref{prop:GreenisCool3} concludes the proof.

\smallskip

\textbf{Proof of (2).}
We will suppose that $k\geq 2$. The case $k=1$ is very similar to the proof of (1). 
Let $\mathcal{K}\in \mathscr{K}_k$. We have 
\begin{align}
2\langle \psi_j, \mathcal{K} \psi_j\rangle &= \sum_{b_1\in B} \int_0^{L_{b_1}} \overline{\psi_j(x_{b_1})}(\mathcal{K} \psi_j)(x_{b_1})\,\dd x_{b_1} \nonumber \\
&= \sum_{(b_1;b_k)\in \mathrm{B}_{k}} \int_0^{L_{b_1}}\int_0^{L_{b_k}} \mathcal{K}_{b_1,b_k}(x_{b_1},y_{b_k})\bigg[\frac{S_{\lambda_j}(L_{b_1}-x_{b_1})}{S_{\lambda_j}(L_{b_1})}\overline{\psi_j(o_{b_1})}  \nonumber \\ \label{e:integraspan}
&\quad+ \frac{S_{\lambda_j}(x_{b_1})}{S_{\lambda_j}(L_{b_1})} \overline{\psi_j(t_{b_1})}\bigg] \cdot \bigg[\frac{S_{\lambda_j}(L_{b_k}-y_{b_k})}{S_{\lambda_j}(L_{b_k})}\psi_j(o_{b_k}) + \frac{S_{\lambda_j}(y_{b_k})}{S_{\lambda_j}(L_{b_k})}\psi_j(t_{b_k})\bigg]\,\dd x_{b_1} \dd y_{b_k}\,.
\end{align}

In analogy to the previous step, we define the operators $J^\gamma_{\mathcal{K}} \in \mathscr{H}_k$, $M^\gamma_{\mathcal{K},1}, M^\gamma_{\mathcal{K},2} \in \mathscr{H}_{k-1}$ and $P^\gamma_{\mathcal{K}}\in \mathscr{H}_{k-2}$ by
\begin{equation}\label{e:jkgen}
J^\gamma_{\mathcal{K}}(b_1; b_k) = \int_0^{L_{b_1}}\int_0^{L_{b_k}} \mathcal{K}_{b_1,b_k}(r_{b_1},s_{b_k})\frac{S_{\Re \gamma}(L_{b_1}-r_{b_1})S_{\Re \gamma}(s_{b_k})}{S_{\Re \gamma}(L_{b_1})S_{\Re \gamma}(L_{b_k})}\,\dd r_{b_1}\,\dd s_{b_k}\,,
\end{equation}
\begin{multline*}
M^\gamma_{\mathcal{K},1}(b_1; b_{k-1}) \\=\sum_{b_k\in \mathcal{N}_{b_{k-1}}^+} \int_0^{L_{b_1}}\int_0^{L_{b_k}} \mathcal{K}_{b_1,b_k}(r_{b_1},s_{b_k})\frac{S_{\Re \gamma}(L_{b_1}-r_{b_1})S_{\Re \gamma}(L_{b_k}-s_{b_k})}{S_{\Re \gamma}(L_{b_1})S_{\Re \gamma}(L_{b_k})}\,\dd r_{b_1}\dd s_{b_k}
\end{multline*}
\begin{equation*}
M^\gamma_{\mathcal{K},2} (b_2; b_{k}) = \sum_{b_1\in \cN_{b_2}^-} \int_0^{L_{b_1}}\int_0^{L_{b_{k}}} \mathcal{K}_{b_1,b_{k}}(r_{b_1},s_{b_{k}})\frac{S_{\Re \gamma}(r_{b_1})S_{\Re \gamma}(s_{b_{k}})}{S_{\Re \gamma}(L_{b_1})S_{\Re \gamma}(L_{b_k})}\,\dd r_{b_1}\,\dd s_{b_{k}}
\end{equation*}
\begin{multline*}
P^\gamma_{\mathcal{K}}(b_2;b_{k-1}) =\sum_{\substack{b_k\in \cN_{b_{k-1}}^+\\b_1\in \cN_{b_2}^-}}\int_0^{L_{b_1}}\int_0^{L_{b_k}} \mathcal{K}_{b_1,b_k}(r_{b_1},s_{b_k})\frac{S_{\Re \gamma}(r_{b_1})S_{\Re \gamma}(L_{b_k}-s_{b_k})}{S_{\Re \gamma}(L_{b_1})S_{\Re \gamma}(L_{b_k})}\,\dd r_{b_1}\dd s_{b_k}\,.
\end{multline*}

Then expanding \eqref{e:integraspan} gives
\begin{multline*}
2 \langle \psi_j, \mathcal{K} \psi_j\rangle = \sum_{(b_1;b_k)\in \mathrm{B}_k} \overline{\psi_j(o_{b_1})} J^\gamma_{\mathcal{K}}(b_1; b_k) \psi_j(o_{b_k}) \\
+ \sum_{(b_1;b_{k-1})\in \mathrm{B}_{k-1}} \overline{\psi_j(o_{b_1})} M^\gamma_{\mathcal{K},1}(b_1; b_{k-1}) \psi_j(o_{b_{k-1}}) \\
 + \sum_{(b_2;b_k)\in \mathrm{B}_{k-1}} \overline{\psi_j(o_{b_2})} M^\gamma_{\mathcal{K},2}(b_2; b_k) \psi_j(o_{b_k}) + \sum_{(b_2;b_{k-1})\in \mathrm{B}_{k-2}} \overline{\psi_j(o_{b_2})} P^\gamma_{\mathcal{K}}(b_2; b_{k-1}) \psi_j(o_{b_{k-1}}).
 \end{multline*}

In other words, setting $M^\gamma_{\mathcal{K}}=M^\gamma_{\mathcal{K},1}+M^\gamma_{\mathcal{K},2}$, we have
\begin{equation}\label{e:genk}
2 \langle \psi_j, \mathcal{K}\psi_j\rangle = \langle \mathring{\psi}_j, \big(J^{\lambda_j}_{\cK}+ M^{\lambda_j}_{\cK}+ P^{\lambda_j}_{\cK} \big)_G  \mathring{\psi}_j\rangle \,.
\end{equation}

We define
\begin{align*}
2\, [ \mathcal{K}]^{\gamma_j} &:= \left[\langle J^{\gamma_j}_{\mathcal{K}}\rangle^{\gamma_j} + \langle M^{\gamma_j}_{\mathcal{K}}\rangle^{\gamma_j} + \langle P^{\gamma_j}_{\mathcal{K}}\rangle^{\gamma_j} \right] \cdot \|\mathring{\psi}_j\|^2\\
&= \langle \mathring{\psi}_j, \left([\langle J^{\gamma_j}_{\mathcal{K}}\rangle^{\gamma_j} + \langle M^{\gamma_j}_{\mathcal{K}}\rangle^{\gamma_j} + \langle P^{\gamma_j}_{\mathcal{K}}\rangle^{\gamma_j}]\mathbf{1}\right)_G \mathring{\psi}_j \rangle\,,
\end{align*}
so that
\begin{equation}\label{eq:DiscVarRed}
\begin{aligned}
&\frac{1}{\mathbf{N}_N(I)}\sum_{\lambda_j\in I} \left|\langle \psi_j, \mathcal{K}\psi_j\rangle - [ \mathcal{K}]^{\gamma_j}\right| \\
&\le \mathrm{Var}^I_{\eta_0} \big{(}\big{(}J_\mathcal{K}^{\gamma}\big{)}_G-\langle J_\mathcal{K}^{\gamma} \rangle^{\gamma}\big{)} + \mathrm{Var}^I_{\eta_0}\big{(}\big{(}M_\mathcal{K}^{\gamma}\big{)}_G-\langle M_\mathcal{K}^{\gamma}\rangle^{\gamma}\big{)}+  \mathrm{Var}^I_{\eta_0} \big{(}\big{(}P_\mathcal{K}^{\gamma}\big{)}_G - \langle P_\mathcal{K}^{\gamma}\rangle^{\gamma}\big{)}\,.
\end{aligned}
\end{equation}

Recalling the definition \eqref{eq:DefBracket2} and the fact that $\Psi_{\gamma,v}(w) = \Im \tilg_N^{\gamma}(v,w)$, we  note that
\[
[ \mathcal{K}]^{\gamma_j} = \frac{\left\| \mathring{\psi}_j^{(N)}\right\|^2}{2\sum_{v\in V_N}\Im \tilg_N^{\gamma_j}(v,v)}\sum_{(b_1;b_k)\in \mathrm{B}_k} \int_0^{L_{b_1}} \int_0^{L_{b_k}} \mathcal{K}_{b_1,b_k}(x_{b_1},y_{b_k})\Psi^{\gamma_j}(x_{b_1},y_{b_k})\,\dd x_{b_1}\,\dd y_{b_k} ,
\]
where
\begin{align*}
&\Psi^{\gamma_j}(x_{b_1},y_{b_k}) = \frac{S_{\lambda_j}(L_{b_1}-x_{b_1})S_{\lambda_j}(y_{b_k})}{S_{\lambda_j}(L_{b_1})S_{\lambda_j}(L_{b_k})} \Im \tilg^{\gamma_j}(o_{b_1},t_{b_k})\\
&\qquad + \frac{S_{\lambda_j}(L_{b_1}-x_{b_1})S_{\lambda_j}(L_{b_k}-y_{b_k})}{S_{\lambda_j}(L_{b_1})S_{\lambda_j}(L_{b_k})}\Im \tilg^{\gamma_j}(o_{b_1},o_{b_k}) \\
&\qquad + \frac{ S_{\lambda_j}(x_{b_1})S_{\lambda_j}(y_{b_k})}{S_{\lambda_j}(L_{b_1})S_{\lambda_j}(L_{b_k})}\Im \tilg^{\gamma_j}(t_{b_1},t_{b_k}) + \frac{S_{\lambda_j}(x_{b_1})S_{\lambda_j}(L_{b_k}-y_{b_k})}{S_{\lambda_j}(L_{b_1})S_{\lambda_j}(L_{b_k})}\Im \tilg^{\gamma_j}(t_{b_1},o_{b_k}) \,.
\end{align*}

Just as in the proof of (1), we introduce
\[
\mathbf{K}^\gamma := \frac{\sum_{(b_1;b_k)\in \mathrm{B}_k}\int_0^{L_{b_1}}\int_0^{L_{b_k}} K(x_{b_1},y_{b_k})\Psi(x_{b_1},y_{b_k})\dd x_{b_1} \dd y_{b_k} }{\int_{\mathcal{G}_N}\Psi(x)\dd x}.
\]

We have
\begin{multline}\label{eq:ThreeTerms2}
\frac{1}{\mathbf{N}_N(I)}\sum_{\lambda_j\in I}|\langle \psi_j,\mathcal{K}\psi_j\rangle-\langle \mathcal{K}\rangle_{\gamma_j}| \le \frac{1}{\mathbf{N}_N(I)}\sum_{\lambda_j\in I} \big|\langle \psi_j,\mathcal{K}\psi_j\rangle - [ \mathcal{K}]^{\gamma_j}\big|\\
 + \frac{1}{\mathbf{N}_N(I)}\sum_{\lambda_j\in I} \big|[ \mathcal{K}]^{\gamma_j} - \mathbf{K}^{\gamma_j}\big|
 + \frac{1}{\mathbf{N}_N(I)}\sum_{\lambda_j\in I} \big| \mathbf{K}^{\gamma_j}- \langle \mathcal{K}\rangle_{\gamma_j}\big|.
\end{multline}

The first term is bounded by \eqref{eq:DiscVarRed}. The second term is estimated as in the proof of (1), noting that $[ \mathcal{K} ]^{\gamma_j} = [ \mathbf{K}^{\gamma_j} \hat{\mathbf{1}} ]^{\gamma_j}$.

To deal with the last term in \eqref{eq:ThreeTerms2}, we use \cite[eq.(A.11)]{AC} and argue as in the proof of (1), to see that 
\begin{multline*}
|\Psi^{\gamma_j}(x_{b_1},y_{b_k})-\Im \tilg^{\gamma_j}(x_{b_1},y_{b_k})| \\
\le C\eta_0\left(|\tilg^{\gamma_j}(o_{b_1},o_{b_k})|+|\tilg^{\gamma_j}(o_{b_1},t_{b_k})|+|\tilg^{\gamma_j}(t_{b_1},o_{b_k})|+|\tilg^{\gamma_j}(t_{b_1},t_{b_k})|\right).
\end{multline*}

As in (1), we deduce that $\frac{1}{\mathbf{N}_N(I)}\sum_{\lambda_j\in I} \big| \mathbf{K}^{\gamma_j}- \langle \mathcal{K}\rangle_{\gamma_j}\big| \to 0$, concluding the proof of the bound. The proof of property \Hol{} also goes as before.

\subsection{Reduction to non-backtracking variances}\label{sec:nonbared}

The aim of this subsection is to prove Proposition~\ref{prop:Reduc2}. Recall definition~\eqref{e:fj} of $f_j,f_j^\ast$.

\smallskip

\textbf{Proof of (1).} Since $\overline{\psi_j(o_b)}\psi_j(t_b)\in\R$ then $\overline{\psi_j(t_b)}\psi_j(o_b) = \overline{\psi_j(o_b)}\psi_j(t_b)$. Hence,
\[
\overline{f_j^{\ast}(b)}f_j(b) = \frac{(1+\overline{\zeta^\gamma(\hat{b})}\zeta^\gamma(b))\overline{\psi_j(o_b)}\psi_j(t_b)-\zeta^\gamma(b)|\psi_j(o_b)|^2-\overline{\zeta^\gamma(\hat{b})}|\psi_j(t_b)|^2}{S_{\lambda}^2(L_b)} \,.
\]
Thus,
\[
\overline{f_j^{\ast}(\widehat{b})}f_j(\widehat{b}) = \frac{(1+\overline{\zeta^\gamma(b)}\zeta^\gamma(\hat{b}))\overline{\psi_j(o_b)}\psi_j(t_b)-\zeta^\gamma(\hat{b})|\psi_j(t_b)|^2-\overline{\zeta^\gamma(b)}|\psi_j(o_b)|^2}{S_{\lambda}^2(L_b)} \,.
\]
Letting $(L_1J)(b) = \frac{1}{1+\overline{\zeta^\gamma(\hat{b})}\zeta^\gamma(b)}(TJ)(b)$ and $(L_2J)(\hat{b}) = \frac{1}{1+\overline{\zeta^\gamma(b)}\zeta^\gamma(\hat{b})}(TJ)(b	)$, we get
\begin{multline*}
(L_1J)(b)\overline{f_j^{\ast}(b)}f_j(b) - (L_2J)(\hat{b})\overline{f_j^{\ast}(\widehat{b})}f_j(\widehat{b})  \\
= \frac{(TJ)(b)}{S_{\lambda}^2(L_b)}\bigg\{|\psi(t_b)|^2\cdot 2\ii \Im \bigg(\frac{\zeta^\gamma(\hat{b})}{1+\zeta^\gamma(\hat{b})\overline{\zeta^\gamma(b)}}\bigg) - |\psi(o_b)|^2\cdot 2\ii\Im\bigg(\frac{\zeta^\gamma(b)}{1+\zeta^\gamma(b)\overline{\zeta^\gamma(\hat{b})}}\bigg)\bigg\}\,.
\end{multline*}
Let $TJ(b)=\beta_b J(t_b)$ with $\beta_b$ to be chosen later. Then,
\begin{align*}
&\langle f_j^{\ast}, [(L_1-L_2)J]_B f_j\rangle = \sum_{b\in B} (L_1J)(b)\overline{f_j^{\ast}(b)}f_j(b) - \sum_{b\in B} (L_2J)(\widehat{b})\overline{f_j^{\ast}(\widehat{b})}f_j(\widehat{b})\\
& \quad = 2\ii \sum_{b\in B} \frac{\beta_bJ(t_b)}{S_{\lambda}^2(L_b)}\bigg\{|\psi(t_b)|^2\cdot\Im \bigg(\frac{\zeta^\gamma(\hat{b})}{1+\zeta^\gamma(\hat{b})\overline{\zeta^\gamma(b)}}\bigg) - |\psi(o_b)|^2\cdot \Im\bigg(\frac{\zeta^\gamma(b)}{1+\zeta^\gamma(b)\overline{\zeta^\gamma(\hat{b})}}\bigg)\bigg\}\\
&\quad = 2\ii \sum_{b\in B} \frac{|\psi(o_b)|^2}{S_{\lambda}^2(L_b)}\Im\bigg(\frac{\zeta^\gamma(b)}{1+\zeta^\gamma(b)\overline{\zeta^\gamma(\hat{b})}}\bigg)\left\{\beta_{\widehat{b}}J(o_b) -  \beta_bJ(t_b)\right\}.
\end{align*}

But by \eqref{e:zetainv}, $\frac{\zeta^\gamma(b)}{1+\zeta^\gamma(b)\overline{\zeta^\gamma(\hat{b})}} = \frac{1}{\frac{1}{\zeta^\gamma(b)}+\overline{\zeta^\gamma(\hat{b})}} = \frac{1}{2\Re \zeta^\gamma(\hat{b})+ \frac{S_{\gamma}(L_b)}{\tilg^{\gamma}(t_b,t_b)}}$, so
\begin{align*}
&\Im \bigg(\frac{\zeta^\gamma(b)}{1+\zeta^\gamma(b)\overline{\zeta^\gamma(\hat{b})}}\bigg) = \frac{-\Im(\frac{S_{\gamma}(L_b)}{\tilg^{\gamma}(t_b,t_b)})}{|2\Re \zeta^\gamma(\hat{b})+ \frac{S_{\gamma}(L_b)}{\tilg^{\gamma}(t_b,t_b)}|^2} = \frac{\frac{N_\gamma(t_b)\Re S_{\gamma}(L_b)}{|\tilg^\gamma(t_b,t_b)|^2}-\frac{\Re\tilg^\gamma(t_b,t_b)\Im S_\gamma(L_b)}{|\tilg^\gamma(t_b,t_b)|^2}}{|2\Re \zeta^\gamma(\hat{b})+ \frac{S_{\gamma}(L_b)}{\tilg^{\gamma}(t_b,t_b)}|^2} \\
& \qquad = \frac{|\zeta^\gamma(b)|^2}{|1+\zeta^\gamma(b)\overline{\zeta^\gamma(\hat{b})}|^2}\cdot \frac{N_\gamma(t_b)\Re S_{\gamma}(L_b)}{|\tilg^\gamma(t_b,t_b)|^2}\bigg(1-\frac{\Re \tilg^\gamma(t_b,t_b)\Im S_\gamma(L_b)}{N_\gamma(t_b)\Re S_\gamma(L_b)}\bigg).
\end{align*}

Choose $\beta_{\widehat{b}} = \frac{S_{\lambda}^2(L_b)|\tilg^{\gamma}(t_b,t_b)|^2|1+\zeta^\gamma(b)\overline{\zeta^\gamma(\hat{b})}|^2}{\Re S_\gamma(L_b)|\zeta^\gamma(b)|^2N_{\gamma}(t_b)}$, so by \eqref{e:zetainv} $\beta_b=\frac{S_{\lambda}^2(L_b)|\tilg^{\gamma}(t_b,t_b)|^2|1+\zeta^\gamma(b)\overline{\zeta^\gamma(\hat{b})}|^2}{\Re S_\gamma(L_b)|\zeta^\gamma(b)|^2N_{\gamma}(o_b)} = \frac{N_{\gamma}(t_b)}{N_\gamma(o_b)}\beta_{\widehat{b}}$. We thus get
\begin{align*}
&\langle f_j^{\ast}, [(L_1-L_2)J]_B f_j\rangle = 2\ii \sum_b |\psi(o_b)|^2\Big(1-\frac{\Im S_\gamma(L_b)\Re \tilg^\gamma(t_b,t_b)}{\Re S_\gamma(L_b)N_\gamma(t_b)}\Big)\Big\{J(o_b) -   \frac{N_{\gamma}(t_b)}{N_\gamma(o_b)}J(t_b)\Big\}\\
&=2\ii\Big(\sum_{o_b}|\psi(o_b)|^2J(o_b)d(o_b) - \sum_{o_b}|\psi(o_b)|^2\frac{1}{N_\gamma(o_b)}\sum_{t_b\sim o_b}N_\gamma(t_b)J(t_b) + \langle \mathring{\psi}_j,(\cE_{\gamma_j}J)_G\mathring{\psi}_j\rangle\Big)\\
&= 2\ii \{\langle \mathring{\psi}, [(I-P_{\gamma})dJ]_G\mathring{\psi}\rangle + \langle \mathring{\psi},(\cE_{\gamma_j}J)_G\mathring{\psi}_j\rangle\},
\end{align*}
since $P_{\gamma}=\frac{d}{N_\gamma}P\frac{N_\gamma}{d}$ and $(\cE_{\gamma}J)(o_b) = \sum_{t_b\sim o_b}[\frac{\Im S_{\gamma}(L_b)\Re \tilg^{\gamma}(t_b,t_b)}{\Re S_{\gamma}(L_b)}(\frac{J(t_b)}{N_{\gamma}(o_b)}-\frac{J(o_b)}{N_\gamma(t_b)})]$. Thus,
\[
\mathrm{Var}^I_{\eta_0}[(I-P_{\gamma})J] \le \mathrm{Var}^I_{\mathrm{nb},\eta_0}[(L_1-L_2)d^{-1}J] + \mathrm{Var}^I_{\eta_0}(\cE_\gamma d^{-1}J) \,.
\]
But $P_{\gamma}(\mathcal{S}_{T,\gamma}K) = \frac{1}{T} \sum_{s=1}^T(T-s+1)P_{\gamma}^s K = \mathcal{S}_{T,\gamma}K - K + \widetilde{\mathcal{S}}_{T,\gamma}K$, so
\[
K = (I-P_{\gamma})S_{T,\gamma}K + \widetilde{S}_{T,\gamma}K \,.
\]
Hence,
\begin{align}\label{e:endz}
\mathrm{Var}^I_{\eta_0}(K)  &\le \mathrm{Var}^I_{\eta_0}[(I-P_{\gamma})S_{T,\gamma} K] + \mathrm{Var}^I_{\eta_0}[ \widetilde{S}_{T,\gamma}K] \\
& \le \mathrm{Var}^I_{\mathrm{nb},\eta_0}[(L_1-L_2)d^{-1}S_{T,\gamma} K] + \mathrm{Var}^I_{\eta_0}[\cE_\gamma d^{-1}S_{T,\gamma}K] +\mathrm{Var}^I_{\eta_0}[ \widetilde{S}_{T,\gamma}K]  \,.\nonumber
\end{align}
Finally, $L_1K'(b) - L_2K'(b) = \frac{(TK')(b) - (TK')(\hat{b})}{1+\zeta^\gamma(b)\overline{\zeta^\gamma(\hat{b})}} = \frac{\beta_b K'(t_b) - \beta_{\widehat{b}}K'(o_b)}{1+\zeta^\gamma(b)\overline{\zeta^\gamma(\hat{b})}}$, so
\[
(L_1 - L_2)K'(b) = \frac{\beta_{\widehat{b}}}{1+\zeta^\gamma(b)\overline{\zeta^\gamma(\hat{b})}}\Big(\frac{N_{\gamma}(t_b)}{N_{\gamma}(o_b)}K'(t_b)-K'(o_b)\Big) = (\cL^\gamma K')(b)
\]
using the above relation on $\beta_b$. Replacing $K$ by $J-\langle J\rangle^{\gamma}$ in \eqref{e:endz} completes the proof.

\smallskip

\textbf{Proof of (2)}
Define $\cT^{\gamma}:\mathscr{H}_1 \to \mathscr{H}_1$ and $\cO^{\gamma}:\mathscr{H}_1\to \mathscr{H}_0$ by
\[
(\cT^{\gamma}K)(b) = \frac{S_{\Re \gamma}^2(L_b)}{1+\overline{\zeta^{\gamma}(\hat{b})}\zeta^{\gamma}(b)} K(b) \, ,
\]
\begin{equation}\label{e:ogam}
(\cO^{\gamma}K)(v) = - \sum_{b\in B ; t_b = v} \frac{\overline{\zeta^{\gamma}(\hat{b})}}{S_{\Re \gamma}^2(L_{b})} K(b) - \sum_{b\in B ; o_b = v} \frac{\zeta^{\gamma}(b)}{S_{\Re \gamma}^2(L_{b})} K(b) \, .
\end{equation}
For $k\ge 2$, define $\cT_k^{\gamma}:\mathscr{H}_k\to\mathscr{H}_k$, $\mathcal{O}_k^{\gamma}:\mathscr{H}_k \to \mathscr{H}_{k-1}$ and $\mathcal{P}_k^{\gamma}:\mathscr{H}_k \to \mathscr{H}_{k-2}$ by
\[
(\cT_k^{\gamma}K)(b_1,\dots, b_k) = S_{\Re \gamma}(L_{b_1})S_{\Re \gamma}(L_{b_k}) K(b_1,\dots, b_k) \,,
\]

\[
(\cO_k^{\gamma}K)(b_1,\dots, b_{k-1}) = -\sum_{b_0\in \mathcal{N}^-_{b_1}} \overline{\zeta^{\gamma}(\hat{b}_0)} K(b_0,\dots, b_{k-1})  - \sum_{b_k \in \cN^+_{b_{k-1}}} K(b_1,\dots, b_k) \zeta^{\gamma}(b_k)\, ,
\]
\[
(\cP_k^{\gamma}K)(b_2,\dots,b_{k-1}) = \sum_{b_1\in \cN^-_{b_2}} \sum_{b_k\in  \cN^+_{b_{k-1}}} \overline{\zeta^{\gamma}(\hat{b}_1)}K(b_1,\dots,b_k) \zeta^{\gamma}(b_k) \, .
\]

Then we have

\begin{lem}                                \label{lem:passup}
For any $K\in \mathscr{H}_1$ satisfying \eqref{eq:APrioriBoundOperators},
\[
\mathrm{Var}^I_{\eta_0}(K-\langle K \rangle^{\gamma}) \le \mathrm{Var}^I_{\mathrm{nb},\eta_0}(\cT^{\gamma}K) + \mathrm{Var}^I_{\eta_0}(\cO^{\gamma}\cT^{\gamma}K - \langle \cO^{\gamma}\cT^{\gamma}K\rangle^{\gamma}) + O_{N\to +\infty, \gamma} (\eta_0)  \, .
\]
For any $K \in \mathscr{H}_k$, $k\ge 2$, we have
\[
\mathrm{Var}^I_{\eta_0}(K-\langle K \rangle^{\gamma}) \le \mathrm{Var}^I_{\mathrm{nb},\eta_0}(\cT_k^{\gamma}K) + \mathrm{Var}^I_{\eta_0}(\mathcal{O}^{\gamma}_kK  - \langle \mathcal{O}_k^{\gamma}K\rangle^{\gamma} ) +\mathrm{Var}^I_{\eta_0}( \mathcal{P}^{\gamma}_kK - \langle\mathcal{P}_k^{\gamma}K\rangle^{\gamma}) \, .
\]
\end{lem}
Before proving the lemma, we note that it implies point (2) of Proposition~\ref{prop:Reduc2}. Indeed, by hypothesis, $J$ satisfies \Hol{}. Applying the lemma several times if necessary, starting from $J$ we obtain variances as in point (2), with operators $L_p^\gamma $, $F_p^\gamma $ consisting of compositions of $\cL^\gamma, S_{T,\gamma},\cO_k^\gamma,\cT^\gamma$,$\ldots$ etc, all of which preserve \Hol{} by Remark~\ref{rem:HolStable}, since they add/multiply by terms in $\mathscr{L}_k^\gamma$.
\begin{proof}[Proof of Lemma~\ref{lem:passup}]
Let $K \in \mathscr{H}_k$, $k \ge 1$. Then
\begin{multline}\label{e:addednum}
 \langle f_j^{\ast},K_B f_j\rangle = \sum_{(b_1;b_k) \in \mathrm{B}_k} \frac{K(b_1;b_k)}{S_{\lambda_j}(L_{b_1})S_{\lambda_j}(L_{b_k})}  \Big{(} \overline{\psi_j(o_{b_1})}\psi_j(t_{b_k}) - \overline{\zeta^{\gamma_j}(\hat{b}_1) \psi_j(t_{b_1})}\psi_j(t_{b_k}) \\
 - \overline{\psi_j(o_{b_1})}\zeta^{\gamma_j}(b_k) \psi_j(o_{b_k}) + \overline{\zeta^{\gamma_j}(\hat{b}_1)\psi_j(t_{b_1})} \zeta^{\gamma_j}(b_k) \psi_j(o_{b_k}) \Big{)} \, .
\end{multline}
Assume $k=1$. Define $\mathcal{U}^{\gamma}:\mathscr{H}_1 \to \mathscr{H}_1$ by
\[
(\cU^{\gamma}K)(b_1) = \frac{1 + \overline{\zeta^{\gamma} (\hat{b}_1)} \zeta ^{\gamma} (b_1)}{S_{\Re \gamma}^2(L_{b_1})} K(b_1)  \, ,
\]
Since $\overline{\psi_j(o_b)}\psi_j(t_b)\in \R$, we have $\overline{\psi_j(t_b)}\psi_j(o_b)=\overline{\psi_j(o_b)}\psi_j(t_b)$ and thus
\begin{equation}\label{e:areceqn}
\langle f_j^{\ast},K_Bf_j\rangle = \langle \mathring{\psi}_j, (\cU^{\gamma_j}K+\cO^{\gamma_j}K)_G\mathring{\psi}_j\rangle \, .
\end{equation}

We now note that
\begin{equation}\label{eq:AlmostSameAverage}
\langle \cU^{\gamma}K\rangle^{\gamma} = -\langle \cO^{\gamma}K \rangle^{\gamma} + O_{N\to +\infty, \gamma} (\eta_0) \, .
\end{equation}

Indeed, we have (cf.~\eqref{eq:DefBracket2})
\[
\langle \cU^{\gamma}K\rangle^{\gamma}  = \frac{1}{\sum_{v\in V}\Psi_{\gamma,v}(v)} \sum_{b\in B} \frac{1+\overline{\zeta^{\gamma}(\hat{b})}\zeta^{\gamma}(b)}{S_{\Re \gamma}^2(L_b)}K(b) \Psi_{\gamma,o_b}(t_b).
\]

On the other hand, recalling \eqref{e:ogam},
\begin{align*}
- \langle \cO^{\gamma}K\rangle^{\gamma} &= \frac{-1}{\sum_{v\in V}\Psi_{\gamma,v}(v)} \sum_{v\in V} \Psi_{\gamma,v}(v) (\cO^{\gamma}K)(v)\\
&= \frac{1}{\sum_{v\in V}\Psi_{\gamma,v}(v)} \sum_{b\in B} \frac{\overline{ \zeta^{\gamma}(\hat{b})}\Psi_{\gamma,t_b}(t_b)+\zeta^{\gamma}(b)\Psi_{\gamma,o_b}(o_b)}{S_{\Re \gamma}^2(L_b)}\,K(b) \\
& = \langle \cU^{\gamma}K\rangle^{\gamma}  + \frac{1}{\sum_{v\in V}\Psi_{\gamma,v}(v)} \sum_b \frac{\overline{\zeta^{\gamma}(\hat{b})}\zeta^{\gamma}(b)\Im S_{\gamma}(L_b)}{S_{\Re \gamma}^2(L_b)}K(b) \,,
\end{align*}
where we used \eqref{eq:psiiden3} and the fact that $\Psi_{\gamma,o_b}(t_b)=\Psi_{\gamma,t_b}(o_b)$ by \eqref{e:sym}. Since $|\Im S_{\gamma}(L_b)|\le C_I\eta_0$, applying Cauchy-Schwarz and Corollary~\ref{cor:ConfinedTilMay11} now proves \eqref{eq:AlmostSameAverage}. 

Recalling \eqref{e:areceqn}, we thus showed that for any $K \in \mathscr{H}_1$,
\begin{align*}
\mathrm{Var}^I_{\eta_0}(K-\langle K \rangle^{\gamma}) & = \mathrm{Var}^I_{\eta_0}\left(\cU^{\gamma}\cT^{\gamma}K - \langle \cU^{\gamma}\cT^{\gamma}K \rangle^{\gamma}\right)  \\
& \le \mathrm{Var}^I_{\mathrm{nb},\eta_0}(\cT^{\gamma}K) + \mathrm{Var}^I_{\eta_0}(\cO^{\gamma}\cT^{\gamma}K - \langle \cO^{\gamma}\cT^{\gamma}K \rangle^{\gamma}) + O_{N\to +\infty, \gamma} (\eta_0)  \, . 
\end{align*}

Now let $k \ge 2$. Recalling \eqref{e:addednum} we have
\[
\langle f_j^{\ast},(\cT_k^{\gamma_j}K)_Bf_j\rangle = \langle  \mathring{\psi}_j, \left(K+\mathcal{O}^{\gamma_j}_kK+\mathcal{P}^{\gamma_j}_kK\right)_G  \mathring{\psi}_j\rangle \, ,
\]
so that
\[
\mathrm{Var}^I_{\eta_0}(K-\langle K \rangle^{\gamma}) \le \mathrm{Var}^I_{\mathrm{nb},\eta_0}(\cT_k^{\gamma}K) +\mathrm{Var}^I_{\eta_0}(\mathcal{O}^{\gamma}_kK + \mathcal{P}^{\gamma}_kK + \langle K\rangle^{\gamma}) \, .
\]

The result hence follows if we show that
\begin{equation}                  \label{eq:gammav}
\langle K\rangle^{\gamma} = -\langle \mathcal{O}_k^{\gamma}K\rangle^\gamma - \langle\mathcal{P}_k^{\gamma}K\rangle^{\gamma} \, .
\end{equation}
Indeed, 
\begin{multline*}
 -\langle \mathcal{O}_k^{\gamma}K\rangle^\gamma -\langle \mathcal{P}_k^{\gamma}K\rangle^{\gamma} = \frac{1}{\sum_{v\in V}\Psi_{\gamma,v}(v)}\sum_{(b_0;b_{k-1})\in \mathrm{B}_k} \overline{\zeta^{\gamma}(\hat{b}_0)} K(b_0 ; b_{k-1}) \Psi_{\gamma,o_{b_1}}(t_{b_{k-1}}) \\
 + \frac{1}{\sum_{v\in V}\Psi_{\gamma,v}(v)}\sum_{(b_1 ; b_k) \in \mathrm{B}_k} K(b_1;b_k)\zeta^{\gamma}(b_k)\Psi_{\gamma,o_{b_1}}(t_{b_{k-1}})  \\
- \frac{1}{\sum_{v\in V}\Psi_{\gamma,v}(v)} \sum_{(b_1 ; b_k) \in \mathrm{B}_k} \overline{\zeta^{\gamma}(\hat{b}_1)}K(b_1 ; b_k)\zeta^{\gamma}(b_k) \Psi_{\gamma,o_{b_2}}(t_{b_{k-1}}) \, ,
\end{multline*}
so \eqref{eq:gammav} follows from \eqref{eq:psiiden2}.
\end{proof}

\section{Control of averaged Green functions}\label{sec:BS}

Throughout the proof we have to control averages of spectral quantities on the finite graphs $(\cQ_N)$. The averages have two forms, either the energy is fixed and the averaging is on the vertices/edges of the graph, or we have a double averaging over the graph and the energy. We discuss the first case here which is technically simpler and illustrates directly how we use our main assumptions. The idea is always as follows: we use the Benjamini-Schramm convergence to replace the averages on the finite graph by some $\expect_{\prob}$-average over the limiting random graph (\S~\ref{s:HelloMaximeAreYouHavingFunReadingTheSourceFile?}), then use the moment condition \Green{} to show these limit averages do not explode (\S~\ref{sec:greencons}).  The second case is handled in Section~\ref{sec:Complex} using results of the present Appendix.

\subsection{Consequences of \protect\BST{}}
\label{s:HelloMaximeAreYouHavingFunReadingTheSourceFile?}
A rooted quantum graph $(\cQ,\mathbf{x_0})$ is a quantum graph $\cQ = (V,E,L,W,\alpha)$ with a root $\mathbf{x_0}=(b_0,x_0)\in \cG$. Let $\mathbf{Q}_{\ast}^{D,\rmm,\rmM}$ be the set of (isomorphism classes of) rooted quantum graphs $[\cQ,\mathbf{x_0}]$ satisfying \Data{}. The set $\mathbf{Q}_{\ast}^{D,\rmm,\rmM}$ is endowed with a metric inspecting the local structure of the quantum graph, see \cite{BSQG}. Convergence of the sequence $(\cQ_N)$ in the sense of Benjamini-Schramm means that the probability measures
\[
\nu_{\cQ_N} :=\frac{1}{2\cL(\cQ_N)}\sum_{b_0\in B}\int_0^{L_{b_0}}\delta_{[\cQ_N,(b_0,x_0)]}\,\dd x_0 = \frac{1}{\cL(\cQ_N)}\int_{\cG_N} \delta_{[\cQ_N,\mathbf{x_0}]}\,\dd\mathbf{x_0}
\]
on $\mathbf{Q}_{\ast}^{D,\rmm,\rmM}$ converge weakly to a probability measure $\prob$ on $\mathbf{Q}_{\ast}^{D,\rmm,\rmM}$.
Such convergence means that if we take an $r$-ball $B(\mathbf{x_0},r)$ uniformly at random in $\cG_N$ (this is encoded by $\nu_{\cQ_N}$), it will resemble a $\prob$-random ball in $\mathbf{Q}_{\ast}^{D,\rmm,\rmM}$. In formulas, the weak convergence $\nu_{\cQ_N}\xrightarrow{w} \prob$ means that for any continuous $F:\mathbf{Q}_{\ast}^{D,\rmm,\rmM}\To\C$, we have
\[
\lim_{N\to\infty} \frac{1}{\cL(\cQ_N)}\int_{\cG_N}F([\cQ_N,\mathbf{x_0}])\,\dd\mathbf{x_0} = \int_{\mathbf{Q}_{\ast}^{D,\rmm,\rmM}}F([\cQ,\mathbf{x_0}])\,\dd\prob([\cQ,\mathbf{x_0}]) =: \expect_{\prob}[F]\,.
\]

Assumption \BST{} says that the limiting law $\prob$ is supported on the subset $\mathbf{T}_{\ast}^{D,\rmm,\rmM}$ of quantum trees.

In \cite{BSQG} we proved that, for any fixed $z\in\C^+$, the important function $\mathbf{G}^z:\mathbf{Q}_{\ast}^{D,\rmm,\rmM}\To \C$ given by $\mathbf{G}^z :[\cQ,\mathbf{x_0}]\mapsto G^z(\mathbf{x_0},\mathbf{x_0})$ is continuous. So by weak convergence,
\begin{equation}\label{e:greenco}
\lim_{N\to\infty} \frac{1}{\cL(\cQ_N)}\int_{\cG_N} g_N^z(\mathbf{x_0},\mathbf{x_0})\,\dd\mathbf{x_0} = \int_{\mathbf{T}_{\ast}^{D,\rmm,\rmM}}G^z(\mathbf{x_0},\mathbf{x_0})\,\dd\prob([\cQ,\mathbf{x_0}]) = \expect_{\prob}(\mathbf{G}^z).
\end{equation}

This was used to prove that the empirical spectral measures of $\cQ_N$ converge vaguely to some averaged spectral measure $\expect_{\prob}(\mu_{\mathbf{x_0}})$. In particular, \cite[Theorem 3.11, Lemma 7.3]{BSQG} imply that for any bounded interval,
\begin{equation}\label{eq:UpperNumberEigen}
\limsup_{N\to +\infty} \frac{\mathbf{N}_N(I)}{|V_N|} \leq C_I.
\end{equation}

More importantly, if $\mathbb{P}( H_\cQ \text{ has spectrum in } I) >0$, then by \cite[Lemma 3.12]{BSQG}, there exists $C'_I>0$ such that
\begin{equation}\label{eq:LowerNumerEigen}
\liminf_{N\to +\infty} \frac{\mathbf{N}_N(I)}{|V_N|} \geq C'_I.
\end{equation}

\begin{rem}
The bound \eqref{eq:LowerNumerEigen} holds in particular on the interval $I_1$ where \Green{} holds. In fact, we have $\int_{I_1}\expect_{\prob}(\Im R_{\lambda+\ii 0}^+(o_b)^{-s}+\Im R_{\lambda+\ii 0}^+(o_{\hat{b}})^{-s})\dd\lambda<\infty$ by Fatou's lemma and \Green{}. Here the integral is more precisely on $I_1\cap \mathfrak{S}$, where $\mathfrak{S}$ is a set of full Lebesgue measure on which all limits $R_{\lambda+\ii 0}^+(o_b)$ exist for $\prob$-a.e. $[\cQ,(b,x_b)]$. The existence of $\mathfrak{S}$ follows from the Herglotz property of $R_z^+(o_b)$, see e.g. Remark~\ref{rem:limits}. It follows that $\int_{I_1}\expect_{\prob}(\Im R_{\lambda+\ii 0}^+(o_b)+\Im R_{\lambda+\ii 0}^+(o_{\hat{b}}))\dd\lambda>0$. Using \cite[Lemma A.2]{AC}, if follows that there is an $L^2$ function $f_b=\phi_{\lambda;b}^-$ such that $\int_{I_1}\expect_{\prob}(\Im \langle f_b, G^{\lambda+\ii 0}f_b\rangle+\Im \langle f_{\hat{b}}, G^{\lambda+\ii 0}f_{\hat{b}}\rangle)\,\dd\lambda >0$, hence $\expect_{\prob}(\mu_{f_b}^{H_{\cQ}}(I_1)+\mu_{f_{\hat{b}}}^{H_{\cQ}}(I_1))>0$, where $\mu_f^{H_Q}$ is the spectral measure of $H_{\cQ}$ at $f$ and we used the standard inversion formula of Borel transform along with Fatou's lemma and the spectral theorem $\langle f,G^z f\rangle = \int_{\R} \frac{1}{t-z}\,\dd\mu_f^{H_\cQ}(t)$. It follows that $\mu_{f_b}^{H_{\cQ}}(I_1)+\mu_{f_{\hat{b}}}^{H_{\cQ}}(I_1)>0$ with positive probability and thus $H_{\cQ}$ has spectrum in $I_1$ with positive probability.
\end{rem}

We are now interested in finding the limits of more general spectral quantities than \eqref{e:greenco}. We first make the observation that, by definition of universal covers and the local metric on $\mathbf{Q}_{\ast}^{D,\rmm,\rmM}$, we have $d([\widetilde{\cQ}_0,\widetilde{\mathbf{x_0}}],[\widetilde{\cQ}_1,\widetilde{\mathbf{x_1}}]) \le d([\cQ_0,\mathbf{x_0}],[\cQ_1,\mathbf{x_1}])$. It follows that the map $\mathbf{Q}_{\ast}^{D,\rmm,\rmM}\ni [\cQ,\mathbf{x_0}]\mapsto \tilde{g}^z(\mathbf{x_0},\mathbf{x_0})$ is also continuous. So by weak convergence,
\begin{equation}\label{e:greenco2}
\lim_{N\to\infty} \frac{1}{\cL(\cQ_N)}\int_{\cG_N} \tilg_N^z(\mathbf{x_0},\mathbf{x_0})\,\dd\mathbf{x_0} = \int_{\mathbf{T}_{\ast}^{D,\rmm,\rmM}}G^z(\mathbf{x_0},\mathbf{x_0})\,\dd\prob([\cQ,\mathbf{x_0}]) = \expect_{\prob}(\mathbf{G}^z).
\end{equation}
 Here we just have $G^z$ on the RHS because the universal cover of a tree is the tree itself.

We will often need to control combinatorial averages like $\frac{1}{N}\sum_{v\in V_N} g_N^z(v,v)$, for $z\in \C^+$. For this, we consider the map $\phi^z:[\cQ,(b,x_b)]\mapsto \frac{G^z(o_b,o_b)}{L_bd(o_b)} + \frac{G^z(o_{\hat{b}},o_{\hat{b}})}{L_b d(o_{\hat{b}})}$. Note that this map is constant over the edge $b$. Here and in what follows, we duplicate all terms (for instance we use both $o_b$ and $o_{\hat{b}}$) just to ensure the map is well-defined, since $[\cQ,(b,x_b)]$ can also be represented as $[\cQ,(\hat{b},L_b-x_b)]$.

%{\color{red}{I don't understand why we distinguish between $t_b$ and $o_{\hat b}$.}}

By arguments similar (and simpler) to \cite[Section 6]{BSQG}, we see that $\phi^z$ is continuous on $\mathbf{Q}_{\ast}^{D,\rmm,\rmM}$. In fact, consider $f^z(x) := G^z(o_b,x)$ for $x\in b$. This solves $Hf^z = zf^z$, so it can be expressed in some basis of solutions $E_b^z$ and $E_{\hat{b}}^z$. We choose them to satisfy $E_b^z(0)=1$ and $(E_b^z)'(0)=-\ii\sqrt{z}$. This yields $f^z(x) = a(b)E_b^z(x) + a(\hat{b})E_{\hat{b}}^z(L_b-x)$ for some coefficients $a(b)$ which depend on $f^z$ only through its initial values $f^z(0)$, $(f^z)'(0)$. In particular $G^z(o_b,o_b) = a(b)+a(\hat{b})E_{\hat{b}}^z(L_b)$. Consequently $\phi^z[\cQ,(b,x_b)] = \frac{a(b)+a(\hat{b})E_{\hat{b}}^z(L_b)}{L_bd(o_b)} + \frac{a(\hat{b})+a(b)E_{b}^z(L_b)}{L_b d(o_{\hat{b}})}$. Now the arguments of \cite[Section 4,6]{BSQG} show that $\phi^z$ is indeed continuous.
For later use (Proposition~\ref{prop:NowGreenIsReallyCool}), we mention that the arguments of \cite{BSQG} imply more generally that it is continuous in $(z,[\cQ,(b,x)])$. For now fix $z$.
Since $\nu_{\cQ_N}\xrightarrow{w}\prob$, this implies that $\frac{1}{\cL(\cQ_N)}\int_{\cG_N} \phi^z[\cQ_N,\mathbf{x}]\,\dd \mathbf{x} \To \expect_{\prob}(\phi^z)$. But an easy expansion gives $\int_{\cG_N} \phi^z[\cQ_N,\mathbf{x}]\,\dd \mathbf{x} = \sum_{v\in V_N} g_N^z(v,v)$. Similarly $N = \frac{1}{2}\sum_{b\in B_N}\int_0^{L_b}(\frac{1}{d(o_b)L_b}+\frac{1}{d(o_{\hat{b}})L_b})\,\dd x_b$. Hence,
\begin{equation}\label{e:greenco3}
\frac{1}{N}\sum_{v\in V_N} g_N^z(v,v) \To \frac{1}{\expect_{\prob}(\frac{1}{d(o_b)L_b}+\frac{1}{d(o_{\hat{b}})L_b})}\expect_{\prob}\Big(\frac{G^z(o_b,o_b)}{d(o_b)L_b}+\frac{G^z(o_{\hat{b}},o_{\hat{b}})}{d(o_{\hat{b}})L_b}\Big)\,.
\end{equation}
The same limit holds for $\tilg_N^z(v,v)$. In a similar way, the map $[\cQ,(b,x_b)]\mapsto \frac{G^z(o_b,t_b)}{L_b}+\frac{G^z(o_{\hat{b}},t_{\hat{b}})}{L_b}$ can be expressed as $\frac{a(b)E_b^z(L_b)+a(\hat{b})}{L_b}+\frac{a(\hat{b})E_{\hat{b}}^z(L_b)+a(b)}{L_b}$, showing that it is continuous and we deduce as before that $\frac{1}{N}\sum_{b\in B_N} g_N^z(o_b,t_b) \To  \frac{2}{\expect_{\prob}(\frac{1}{d(o_b)L_b}+\frac{1}{d(o_{\hat{b}})L_b})}\expect_{\prob}(\frac{G^z(o_b,t_b)}{L_b})$, using \eqref{e:sym}.

Recalling that $\zeta^z(b) = \frac{G^z(o_b,t_b)}{G^z(o_b,o_b)}$, we immediately deduce the continuity of $[\cQ,(b,x_b)]\mapsto \frac{\zeta^z(b)}{L_b}+\frac{\zeta^z(\hat{b})}{L_b}$, thus obtaining $\frac{1}{N}\sum_{b\in B_N} \zeta^z(b)\To  \frac{1}{\expect_{\prob}(\frac{1}{d(o_b)L_b}+\frac{1}{d(o_{\hat{b}})L_b})}\expect_{\prob}\big(\frac{\zeta^z(b)+\zeta^z(\hat{b})}{L_b}\big)$.

It follows from \cite[p.10]{PT87} that $[\cQ,(b,x_b)]\mapsto (C_z(L_b),S_z(L_b))$ is continuous. Indeed, this is just saying that $(C_z(L_b),S_z(L_b))$ depends continuously on $(W_b,L_b)$. So using \eqref{e:zetawt}, we deduce that $\frac{1}{N}\sum_{b\in B_N}R_z^+(o_b)\To  \frac{1}{\expect_{\prob}(\frac{1}{d(o_b)L_b}+\frac{1}{d(o_{\hat{b}})L_b})}\expect_{\prob}\big(\frac{R_z^+(o_b)+R_z^+(o_{\hat{b}})}{L_b}\big)$.

Similarly considering $[\cQ,(b,x_b)] \mapsto \frac{1}{L_b}(\zeta^z(b)\sum_{b^+\in \cN_b^+}\zeta^z(b^+)+\zeta^z(\hat{b})\sum_{\hat{b}^+\in \cN_{\hat{b}}^+}\zeta^z(\hat{b}^+))$, $\frac{1}{N}\sum_{(b_0,b_1)\in\mathrm{B}_2}\zeta^z(b_0)\zeta^z(b_1)\To c_{d,L}\expect_{\prob}(\frac{1}{L_b}(\zeta^z(b)\sum_{b^+\in \cN_b^+}\zeta^z(b^+)+\zeta^z(\hat{b})\sum_{\hat{b}^+\in \cN_{\hat{b}}^+}\zeta^z(\hat{b}^+))$. Note that $\cN_{\hat{b}}^+ = \{\widehat{b'}:b'\in\cN_b^-\}$.

Throughout the paper, we need continuity of more general functionals, depending on several bonds. This is why we introduce the following class.

\begin{defa}\label{def:GeneralContinuousOp}
If $k\in \N$, we define $\mathscr{L}_k^\gamma\subset \mathscr{H}_k$ to be the set of all $F^\gamma(b_1,\dots,b_k)$ that are sums and products of the following quantities, their inverses and complex conjugates:
\begin{itemize}
\item $S_\gamma(L_{b_i})$, $S_{\Re \gamma}(L_{b_i})$, $\int_0^{L_{b_i}} S^n_\gamma(x_{b_i}) S^m_\gamma(L_{b_i}-x_{b_i}) \dd x_i$, $\int_0^{L_{b_i}} S^n_{\Re{\gamma}}(x_{b_i}) S^m_{\Re{\gamma}}(L_{b_i}-x_{b_i}) \dd x_i$ for $m,n\in \N\cup \{0\}$ and $i\in \{1,\dots,k\}$.
\item $\zeta^\gamma(b_i)$, $\zeta^\gamma(\hat{b_i})$, $|1+ \overline{\zeta^\gamma(\hat{b_i})} \zeta^\gamma(b_i)|$,  $R_\gamma^{+}(o_{b_i})$, $R_\gamma^{-}(t_{b_i})$, $\Im R_\gamma^+(o_{b_i})$, $\Im R_\gamma^-(t_{b_i})$.
\item $\tilg_N^\gamma(v,w)$, $g_N^\gamma(v,w)$, $\Im g_N^\gamma(v,w)$, $v,w\in \{o_{b_i},t_{b_i}\}_{1\le i\le k}$.
\item $\int_{0}^{L_{b_i}} \Im \tilg(x_{b_i},x_{b_i}) \dd x_{b_i}$, $\int_{0}^{L_{b_i}} \frac{1}{\Im \tilg(x_{b_i},x_{b_i})} \dd x_{b_i}$ for $i\in \{1,\dots,k\}$.
\end{itemize}
\end{defa}

\begin{rem}\label{rem:limits}
Consider $F^\gamma\in \mathscr{L}_{k+1}^\gamma$, $\gamma\in\C^+$. We note that there is a Lebesgue-null set $\mathcal{A}\subset \R$ such that, for each $\lambda\notin \mathcal{A}$, there exists a set $\Omega_\lambda\subset \mathbf{Q}_*^{D,\rmm,\rmM}$ with $\mathbb{P}(\Omega_\lambda)=1$ having the following property: if $[\cQ,(b_0,x_0)]\in \Omega_\lambda$, $(b_1,\dots, b_k)\in \mathrm{B}_k^{b_0}$ , then $\lim\limits_{\eta_0 \downarrow 0} F^{\lambda + \ii \eta_0} (b_0,\dots,b_k)$ exists and is finite.

Indeed, if $F^\gamma\in \mathscr{L}_{k+1}^\gamma$, then for every fixed $[\cQ,(b_0,x_0)]$ and every $(b_1,\dots, b_k)\in \mathrm{B}_k^{b_0}$, the map $F^{\lambda + \ii \eta_0} (b_0,\dots,b_k)$ is a product of Herglotz functions, their complex conjugate, and functions having limits on $\R$.  We deduce that the limit $\lim\limits_{\eta_0\downarrow 0} F^{\lambda + \ii \eta_0} (b_0,\dots,b_k)$ exists and is finite for all $(b_1,\dots, b_k)\in \mathrm{B}_k^{b_0}$ and all $\lambda \in \R \setminus \mathcal{A}_{[\cQ,b_0]}$ for some Borel set $\mathcal{A}_{[\cQ,b_0]}$ of measure zero (see e.g. \cite[Corollary 3.29]{Teschl}). By Fubini's theorem, the set
\[
\mathcal{C}=\{([\cQ,(b_0,x_0)], \lambda)\in \mathbf{Q}_*^{D,\rmm,\rmM}\times \R : \lambda \in \mathcal{A}_{[\cQ,b_0]} \}
\]
has $(\mathbb{P}\otimes \mathrm{Leb})$-measure zero, and hence, for almost all $\lambda\in \R$, the set
\[
O_{\lambda} = \{[\cQ,(b_0,x_0)]\in \mathbf{Q}_*^{D,\rmm,\rmM} : ([\cQ,(b_0,x_0)],\lambda)\in\mathcal{C} \}
\]
has $\prob$-measure zero. Taking $\Omega_{\lambda}=O_{\lambda}^c$ proves our claim.
\end{rem} 

\begin{prp}\label{prop:NowGreenIsReallyCool}
Let $\cQ_N$ be a sequence of quantum graphs satisfying \emph{\Data{}} and converging to $\prob$ in the sense of Benjamini-Schramm. Let $\Xi \subset \C^+$ be compact and let $F^\gamma\in\mathscr{L}^\gamma_{k+1}$. Then uniformly in $\gamma\in\Xi$, we have
\begin{multline}\label{eq:stuckAtHomeForTenDays}
\frac{1}{N}   \sum_{(b_0,\dots,b_k)\in \mathrm{B}_{k+1}} F^\gamma(b_0,\dots,b_k) \\
\Lim_{N\to\infty} c_{d,L}\mathbb{E}_\mathbb{P}\bigg[  \sum_{(b_1,\dots,b_k)\in \mathrm{B}^{b_0}_k} F^\gamma (b_0,\dots,b_k) + \sum_{(b_{-k},\dots,b_{-1})\in \mathrm{B}_{k,b_0}} F^\gamma (\hat{b}_0,\dots,\hat{b}_{-k})\bigg],
\end{multline}
where $c_{d,L} = \frac{1}{\expect_{\prob}(\frac{1}{d(o_b)L_b}+\frac{1}{d(o_{\hat{b}})L_b})}$.
\end{prp}
\begin{proof}
Let $\tilde{F}^\gamma(b_0) = \sum_{(b_1;b_k)\in \mathrm{B}_{k}^{b_0}} F^\gamma(b_0;b_k)$. The previous arguments show that $\tilde{F}^\gamma(b_0)$ is continuous on $\mathbf{Q}_\ast^{D,m,M}$ for fixed $\gamma$, so \eqref{eq:stuckAtHomeForTenDays} holds. For the Green's functions, e.g. $g_N^{\gamma}(o_{b_1},t_{b_p})$, one just expands $f^\gamma(x) = G^\gamma(o_{b_1},x)$ for $x\in b_p$ using the basis of $E_{b_p}^\gamma,E_{\hat{b}_p}^\gamma$ as before. It remains to justify uniformity in $\gamma$.

Let us consider the limit in \eqref{e:greenco3}, the general case is similar. We first observe that $\phi:(\gamma,[\cQ,(b,x_b)])\mapsto \frac{G^{\gamma}(o_b,o_b)}{L_bd(o_b)}+\frac{G^{\gamma}(o_{\hat{b}},o_{\hat{b}})}{L_bd(o_{\hat{b}})}$ is uniformly continuous on $ \Xi \times \mathbf{Q}_{\ast}^{D,\rmm,\rmM}$. In fact, since $\phi$ does not depend on the value of $x\in b$, it is enough to study uniform continuity on the compact set $\Xi\times\mathbf{Q}_{\ast}^{D,\rmm,\rmM,\delta}$, where $\mathbf{Q}_{\ast}^{D,\rmm,\rmM,\delta}:= \{ (\cQ, (b,x_b))\in \mathbf{Q}_{\ast}^{D,\rmm,\rmM} ~ |~ x_b\in [\delta,L_{b}-\delta]  \}$ (see \cite[Lemma 3.6]{BSQG}). We are thus reduced to checking continuity in $(\gamma,[\cQ,(b,x_b)])$, which holds as pointed out before \eqref{e:greenco3}.

Using uniform continuity on $\Xi\times \mathbf{Q}_{\ast}^{D,\rmm,\rmM}$, it easily follows that the sequence $F_N(\gamma) = \frac{1}{\cL(\cQ_N)}\int_{\cG_N}\phi(\gamma,[\cQ_N,\mathbf{x_0}])\,\dd\mathbf{x_0}$ is uniformly equicontinuous. The uniform convergence of $F_N(\gamma)$ follows.
\end{proof}

\subsection{Consequences of \protect\Green{}}\label{sec:greencons}
Let $\overline{I}\subset I_1$ as in \NonDirichlet{} and $z\in \C^+$. If all $\cQ_N$ satisfy \Data{}, we may fix $c_1,c_2,c_3>0$ such that for all $z\in I+\ii[0,1]$, $b\in \cup_N B_N$,
\begin{equation}\label{e:cj}
c_1\le |S_z(L_b)|\le c_2 \qquad \text{and} \qquad |C_z(L_b)|\le c_3 \,.
\end{equation}

Using \cite[Corollary 2.5]{AC} along with Proposition~\ref{prop:NowGreenIsReallyCool}, we see that \Green{} implies that for any $s>0$,
\[
\sup_{\lambda\in I_1,\eta\in (0,1)} \expect_{\prob}(|G^z(o_b,o_b)|^s + |G^z(o_{\hat{b}},o_{\hat{b}})|^s)< \infty\,, \quad \sup_{\lambda\in I_1,\eta\in (0,1)} \expect_{\prob}(|\hat{\zeta}^z(b)|^s + |\hat{\zeta}^z(\hat{b})|^s) <\infty
\]
where $z:=\lambda+\ii\eta$ and
\[
\sup_{\lambda\in I_1,\eta\in (0,1)} \expect_{\prob}(|\hat{R}_z^+(o_b)|^s+|\hat{R}_z^+(o_{\hat{b}})|^s)< \infty\,.
\]
Using \eqref{e:2} then \eqref{e:zetainv}, we also deduce that
\[
\sup_{\lambda\in I_1,\eta\in (0,1)} \expect_{\prob}(|G^z(o_b,o_b)|^{-s} + |G^z(o_{\hat{b}},o_{\hat{b}})|^{-s}) <\infty\,, \ \sup_{\lambda\in I_1,\eta\in (0,1)} \expect_{\prob}(|\hat{\zeta}^z(b)|^{-s} + |\hat{\zeta}^z(\hat{b})|^{-s}) <\infty.
\]
Next, \eqref{e:midgreen} and the Herglotz property imply that $\Im G^z(o_b,o_b)\ge \Im R_z^+(o_b)|G^z(o_b,o_b)|^2$, so we deduce that
\[
\sup_{\lambda\in I_1,\eta\in (0,1)}\expect_{\prob}(|\Im G^z(o_b,o_b)|^{-s} + |\Im G^z(o_{\hat{b}},o_{\hat{b}})|^{-s})<+\infty .
\]

We have thus obtained bounds on the powers of all the quantities listed in Definition~\ref{def:GeneralContinuousOp}, except on the negative powers of $|1+ \overline{\zeta^\gamma(\hat{b})} \zeta^\gamma(b)|$ and the positive powers\footnote{The negative powers of $\int_{0}^{L_{b}} \frac{1}{\Im \tilg(x_{b},x_{b})} \dd x_{b}$ are easily bounded using Jensen's inequality $(\frac{1}{L_b}\int_0^{L_b}\frac{1}{\Im\tilg(x_b,x_b)}\,\dd x_b)^{-1}\le \frac{1}{L_b}\int_0^{L_b}\Im\tilg(x_b,x_b)\,\dd x_b$.} of $\int_{0}^{L_{b}} \frac{1}{\Im \tilg(x_{b},x_{b})} \dd x_{b}$.

These are dealt with using Lemmas \ref{lem:InvZeta} and \ref{lem:InvG} below, which imply that for all $s>0$,
\begin{equation}\label{e:maximesolved}
\sup_{\lambda\in I_1,\eta\in (0,1)}\expect_{\prob}\left[|1+ \overline{\zeta^\gamma(\hat{b})} \zeta^\gamma(b)|^{-s}\right] <+\infty .
\end{equation}
and
\begin{equation*}
\sup_{\lambda\in I_1,\eta\in (0,1)}\expect_{\prob}\left[\left( \int_0^{L_b} \frac{1}{ \Im \tilg^\gamma(x_b,x_b)} \mathrm{d}x_b\right)^{s}\right] <+\infty .
\end{equation*}

\begin{lem}\label{lem:InvZeta}
We have
\[
\frac{1}{|1+\overline{\zeta^\gamma(\hat{b})} \zeta^\gamma(b)|} \leq \max \left( \frac{|G^\gamma(o_b,o_b)G^\gamma(t_b,t_b)|}{|G^\gamma(o_b,t_b)|^2}, \frac{|G^\gamma(o_b,o_b)G^\gamma(t_b,t_b)|}{\Im G^\gamma(o_b,o_b) \Im G^\gamma(t_b,t_b)}\right).
\]
\end{lem}
To deduce \eqref{e:maximesolved}, note that $ \frac{|G^\gamma(o_b,o_b)G^\gamma(t_b,t_b)|}{|G^\gamma(o_b,t_b)|^2} = \frac{1}{|\zeta^\gamma(b)\zeta^\gamma(\hat{b})|}$.
\begin{proof}
We have
$ \overline{\zeta^\gamma(\hat{b})} \zeta^\gamma(b) = \frac{|G^\gamma(o_b,t_b)|^2}{G^\gamma(o_b,o_b) \overline{G^\gamma(t_b,t_b)}}$, so $\frac{1}{1+\overline{\zeta^\gamma(\hat{b})} \zeta^\gamma(b)} = \frac{G^\gamma(o_b,o_b) \overline{G^\gamma(t_b,t_b)}}{G^\gamma(o_b,o_b) \overline{G^\gamma(t_b,t_b)} + |G^\gamma(o_b,t_b)|^2}$.

Let us write $G^\gamma(o_b,o_b)= x+\ii y$, $\overline{G^\gamma(t_b,t_b)}= x'-\ii y'$, so that 
\[
G^\gamma(o_b,o_b) \overline{G^\gamma(t_b,t_b)} + |G^\gamma(o_b,t_b)|^2 =|G^\gamma(o_b,t_b)|^2 + xx' + yy' + \ii (yx' - y'x).
\]
We know by Lemma \ref{lem:ASW} that $y,y' \geq 0$. 
\begin{itemize}
\item If $xx' \geq - yy'$, we have $|G^\gamma(o_b,o_b) \overline{G^\gamma(t_b,t_b)} + |G^\gamma(o_b,t_b)|^2| \geq |G^\gamma(o_b,t_b)|^2$.

\item  If $xx' \leq - yy'$, then %we have $|yx' - y'x| \geq y \left(|x'| + \frac{(y')^2}{|x'|}\right).$ The map $(0,+\infty) \ni t \mapsto t + \frac{(y')^2}{t}$ is minimal when $t= y'$, 
$(yx'-y'x)^2\ge 2(yy')^2$, so  $|yx' - y'x| \ge \sqrt{2} y y'$. Hence,
\[
|G^\gamma(o_b,o_b) \overline{G^\gamma(t_b,t_b)} + |G^\gamma(o_b,t_b)|^2| \geq \Im G^\gamma(o_b,o_b) \Im G^\gamma(t_b,t_b).     \qedhere
\]
\end{itemize}
\end{proof}

\begin{lem}\label{lem:InvG}
Let $\Xi\subset \C^+$ be a bounded set. Suppose that $\mathcal{Q}$ satisfies 
\emph{\Data{}}. We may find a constant $C'_\Xi$ such that for all $b\in B$ and all $\gamma\in \Xi$, we have 
\begin{equation}\label{eq:BoundIntInverse}
\int_0^{L_b} \frac{1}{ \Im \tilg^\gamma(x_b,x_b)} \mathrm{d}x_b \leq C'_\Xi \frac{(1+|R_\gamma^+(o_b)|)^6}{|\zeta^\gamma(b)|^6 (\Im R_\gamma^+(t_b))^3}  + C'_\Xi \frac{(1+|R_\gamma^-(t_b)|)^6}{|\zeta^\gamma(\hat{b})|^6 (\Im R_\gamma^-(o_b))^3}.
\end{equation}
\end{lem}
\begin{proof}
Let $b\in B$, $x\in b$. Recalling the notations and results of Section \ref{subsecGreen}, we have $- \tilg(x,x)= \frac{1}{R_{\gamma}^+(x) + R_{\gamma}^-(x)}$,
so that by the Herglotz property,
\begin{equation}\label{eq:UnSurImG}
 \frac{1}{\Im \tilg^\gamma(x,x)} = \frac{|R_{\gamma}^+(x) + R_{\gamma}^-(x)|^2}{\Im (R_{\gamma}^+(x) + R_{\gamma}^-(x))} \leq \frac{2|R_{\gamma}^+(x)|^2}{\Im R_{\gamma}^+(x)}  + \frac{2|R_{\gamma}^-(x)|^2}{\Im R_{\gamma}^-(x)}\,.
\end{equation}

Since $V_{\gamma; o_b}^+(t_b)=\zeta^\gamma(b)$, arguing as for \eqref{eq:JConserved} with $f(x)=V_{\gamma; o_b}^+(x)$, we have
\begin{equation}\label{eq:JoliEncadrement}
|\zeta^\gamma(b)|^2 \Im R_\gamma^+(t_b)\leq | V_{\gamma;o_b}^+(x)|^2 \Im R_\gamma^+(x).
\end{equation}

Noting that $V_{\gamma;o_b}^+(y) = C_\gamma(y) + R_\gamma^+(o_b) S_\gamma(y)$, we deduce that for any compact set $\Xi\subset \C$, there exists $C_\Xi>0$ depending only on the constants in \Data{} and on $\Xi$ such that for all $z\in \Xi$, all $N\in \N$ and all $b\in B_N$, we have 
\begin{equation}\label{eq:BoundsOnV}
| V_{\gamma;o_b}^+(y)| + | V_{\gamma;o_b}^{+\prime}(y)| \leq C_\Xi (1+|R_\gamma^+(o_b)|).
\end{equation}

In particular, using \eqref{eq:JoliEncadrement}, we get
\begin{equation}\label{eq:LowerImR}
 \Im R_\gamma^+(x)  \geq \frac{|\zeta^\gamma(b)|^2 \Im R_\gamma^+(t_b)}{C_\Xi^2 (1+|R_\gamma^+(o_b)|)^2}.
\end{equation}

On the other hand, we have, using \eqref{eq:BoundsOnV},

\[
|V_{\gamma;o_b}^+(x)|^2 \Im R_ \gamma^+(x) = \Im (\overline{V^+_{\gamma;o_b}(x)} V_{\gamma;o_b}^{+\prime}(x)) \leq |V_{\gamma;o_b}^+(x)V_{\gamma;o_b}^{+\prime}(x)| \leq C_\Xi(1+|R_\gamma^+(o_b)|)|V_{\gamma;o_b}^+(x)|,
\]
so, by \eqref{eq:JoliEncadrement}, we get
\[
|V_{\gamma;o_b}^+(x)| \geq \frac{|\zeta^\gamma(b)|^2 \Im R_\gamma^+(t_b)}{C_\Xi (1+|R_\gamma^+(o_b)|)}.
\]

All in all, we have
\[
|R_\gamma^+(x)| = \frac{|(V_{\gamma;o_b}^+)'(x)|}{|V_{\gamma;o_b}^+(x)|}\leq  \frac{C^2_\Xi (1+|R_\gamma^+(o_b)|)^2}{|\zeta^\gamma(b)|^2 \Im R_\gamma^+(t_b)}.
\]

The same proof gives a similar bound for $|R_\gamma^-(x)|$ and a bound similar to \eqref{eq:LowerImR} for  $\Im R_\gamma^-(x)$, exchanging the roles of $o_b$ and $t_b$. Therefore, combining this with \eqref{eq:UnSurImG} and \eqref{eq:LowerImR}, we obtain \eqref{eq:BoundIntInverse}.
\end{proof}

The previous considerations imply that the limits \eqref{eq:stuckAtHomeForTenDays} are controlled:

\begin{prp}\label{prp:b7}
Suppose \emph{\BST{}}, \emph{\Data{}}, \emph{\NonDirichlet{}} and \emph{\Green{}} are satisfied, let $\overline{I}\subset I_1$ be compact, and $F^{\gamma}\in \mathscr{L}_{k+1}^\gamma$. Then
\[
\sup_{\lambda\in I,\eta\in (0,\eta_{\mathrm{Dir}})}\mathbb{E}_\mathbb{P}\bigg[  \sum_{(b_1,\dots,b_k)\in \mathrm{B}^{b_0}_k} F^\gamma (b_0,\dots,b_k) + \sum_{(b_{-k},\dots,b_{-1})\in \mathrm{B}_{k,b_0}} F^\gamma (\hat{b}_0,\dots,\hat{b}_{-k})\bigg]<+\infty.
\]
\end{prp}
\begin{proof}
By definition $F^{\gamma}(b_0;b_k)$ is a sum or product of quantities in Definition~\ref{def:GeneralContinuousOp}. If $F^\gamma(b_0;b_k) = \prod_{i=0}^k T^\gamma_i(b_i)$, then using H\"older's inequality with $\sum_{i=0}^k \frac{1}{p_i}=1$, we have
\[
\Big|\frac{1}{N}\sum_{(b_0;b_k)} F^\gamma(b_0;b_k)\Big| \le \prod_{i=0}^k \Big(\frac{1}{N}\sum_{(b_0;b_k)} |T^\gamma_i(b_i)|^{p_i}\Big)^{1/p_i} \le C_{D,k} \prod_{i=0}^k\Big(\frac{1}{N}\sum_{b\in B_N}|T^\gamma_i(b)|^{p_i}\Big)^{1/p_i}.
\]
Taking the limit $N\To+\infty$, Proposition~\ref{prop:NowGreenIsReallyCool} thus gives
\begin{multline*}
\mathbb{E}_\mathbb{P}\bigg[  \sum_{(b_1,\dots,b_k)\in \mathrm{B}^{b_0}_k} \hat{F}^\gamma (b_0,\dots,b_k) + \sum_{(b_{-k},\dots,b_{-1})\in \mathrm{B}_{k,b_0}} \hat{F}^\gamma (\hat{b}_0,\dots,\hat{b}_{-k})\bigg]\\
\le C_{D,k}\prod_{i=0}^k\Big(\mathbb{E}_\mathbb{P}\left[ T^\gamma_i (b_0) + T^\gamma_i(\hat{b}_0)\right]^{p_i}\Big)^{1/p_i}.
\end{multline*}
We are thus reduced to the case $k=0$. But we already showed above how to bound each of the quantities $T^\gamma(b) = \zeta^\gamma(b)$, $R_\gamma^{+}(o_b)$,... etc.\ uniformly. $R^+_\gamma(t_b)$ is also controlled since $ R_\gamma^+(t_b) = \sum_{b^+\in\mathcal{N}_b^+}  R_\gamma^+(o_{b^+}) - \alpha_{t_b}$. $R^-_\gamma(t_b) = R^+_\gamma(o_{\widehat{b}})+\frac{C_\gamma(L_b)-S_\gamma'(L_b)}{S_\gamma(L_b)}$ is controlled as well.

The case when $F^\gamma(b_0;b_k) = G^\gamma(o_{b_0},t_{b_k})$ is controlled as above, we just view it as the limit of $\tilg_N^\gamma(o_{b_0},t_{b_k})$, which has a product form by \eqref{e:greenmul}.
\end{proof}

Combining Proposition~\ref{prop:NowGreenIsReallyCool} and \ref{prp:b7}, we deduce:

\begin{cor}\label{cor:ConfinedTilMay11}
Suppose $\cQ_N$ satisfy \emph{\BST{}}, \emph{\Data{}}, \emph{\NonDirichlet{}} and \emph{\Green{}}. Let $k\geq 0$, and let $F^{\gamma} \in \mathscr{L}_{k+1}^{\gamma}$. Then we have
\[
\sup_{\eta_0\in (0,\eta_{\mathrm{Dir}})}  \limsup_{N\to +\infty}  \frac{1}{N}  \sup_{\eta_1\in (-\frac{\eta_0}{2},\frac{\eta_0}{2})}\sup_{\lambda\in I_1} \sum_{(b_1,\dots,b_k)\in \mathrm{B}_k} | F^{\lambda+\ii\eta_0+\ii\eta_1}(b_0,\dots,b_k)| < +\infty.
\]
\end{cor}

\begin{prp}\label{prop:GreenisCool3}
Suppose $\cQ_N$ satisfy \emph{\BST{}}, \emph{\Data{}}, \emph{\NonDirichlet{}} and \emph{\Green{}}. Let $k\geq 0$, and let $F^{\gamma} \in \mathscr{L}_{k+1}^{\gamma}$. For almost all $\lambda\in I_1$, for all $s>0$ and all $t,t'\in (0,2)$, these operators also satisfy
\[
 \lim\limits_{\eta_0\downarrow 0} \limsup\limits_{N\to +\infty}  \frac{1}{N}   \sum_{(b_0,\dots,b_k)\in \mathrm{B}_{k+1}} \left| F^{\lambda+t\ii\eta_0}(b_0,\dots,b_k) - F^{\lambda+t'\ii\eta_0}(b_0,\dots,b_k)\right|^s =0.
\]
\end{prp}
\begin{proof}
As in \eqref{eq:stuckAtHomeForTenDays}, we have for every $\eta_0>0$,
\begin{align*}
&\lim\limits_{N\to +\infty}  \frac{1}{N}   \sum_{(b_0;b_k)\in \mathrm{B}_{k+1}} \left| F^{\lambda+t\ii\eta_0}(b_0;b_k) - F^{\lambda+t'\ii\eta_0}(b_0;b_k)\right|^s\\
& = c_{D,L}\mathbb{E}_\mathbb{P}\bigg[ \sum_{(b_1;b_k)\in \mathrm{B}_{k}^{b_0}} \left| F^{\lambda+t\ii\eta_0}(b_0;b_k) - F^{\lambda+t'\ii\eta_0}(b_0;b_k)\right|^s \\ &\qquad+\sum_{(b_{-k};b_{-1})\in \mathrm{B}_{k,b_0}}
\left| F^{\lambda+t\ii\eta_0}(\hat{b}_{0};\hat{b}_{-k}) - 
  F^{\lambda+t'\ii\eta_0}(\hat{b}_{0};\hat{b}_{-k})\right|^s\bigg].
\end{align*}
By Remark~\ref{rem:limits}, we know that for $\lambda\in I_1\cap \cA^c$, the integrand converges to zero as $\eta_0\downarrow 0$, $\prob$-a.s. Moreover, it follows from Proposition~\ref{prp:b7} that for any $p>1$, the integrand is bounded in $L^p(\mathbb{P})$ independently of $\eta_0$.

The result then follows by the classical fact that, if $(f_n)$ is a sequence of measurable functions on some probability space such that $f_n(x)$ converges to zero almost everywhere and $\sup\limits_n \|f_n\|_{L^p} <\infty$ for some $p>1$, then $\lim\limits_{n\to \infty} \|f_n\|_{L^1} = 0$ (indeed, the $p$-moment assumption implies the sequence is uniformly integrable). 
\end{proof}

\providecommand{\bysame}{\leavevmode\hbox to3em{\hrulefill}\thinspace}
\providecommand{\MR}{\relax\ifhmode\unskip\space\fi MR }
% \MRhref is called by the amsart/book/proc definition of \MR.
\providecommand{\MRhref}[2]{%
  \href{http://www.ams.org/mathscinet-getitem?mr=#1}{#2}
}
\providecommand{\href}[2]{#2}

\end{document}